\documentclass[a4paper]{article} 

\pdfoutput=1

 \usepackage{amsmath,amsfonts,amssymb}

\newcommand{\sign}     {\mbox{sign}}

\newcommand{\const}{\mbox{\rm const.}}

\newcommand{\inv}{^{-1}}

\newcommand{\cov}{\mbox{\rm cov}}

\newcommand{\diag}{\mbox{\rm Dg}}

\newcommand{\Rset}{{\mathbb R}}

\newcommand{\e}   {\mbox{\rm e}}

\newcommand{\Hmat}{{\bf H}}

\newcommand{\trace}{\mbox{\rm Tr}}

\newcommand{\trans}{^{\top}}
\newcommand{\half}{\frac{1}{2}}

\newtheorem{theorem}{Theorem}[section]

\newtheorem{proof}[theorem]{Proof}

\newtheorem{lemma}{Lemma}
\newtheorem{proposition}{Proposition}

\DeclareMathAlphabet{\mathpzc}{OT1}{pzc}{m}{it}
\DeclareMathAlphabet{\mathcalligra}{T1}{calligra}{m}{n}

\newcommand{\Nmat}{{\bf N}}

\newcommand{\kcat}     {k^{\rm cat}}
\newcommand{\kcatplus} {k^{+}_{\rm cat}}
\newcommand{\kcatminus}{k^{-}_{\rm cat}}
\newcommand{\km}       {K_{\rm M}}
\newcommand{\keq}      {K_{\rm eq}}

\newcommand{\kcatplusl}{k_{{\rm cat,} l}^{+}}

\newcommand{\kmS}  {K_{\rm S}}
\newcommand{\kmP}  {K_{\rm P}}
\newcommand{\kmi}  {K_{i}}
\newcommand{\kmj}  {K_{j}}
\newcommand{\kmli} {K_{li}}

\usepackage{setspace} 

\usepackage{enumitem}
\usepackage{xr}
\usepackage{color}
\usepackage{amsmath,amsfonts,amssymb}
\usepackage{epsfig}
\usepackage{url}
\usepackage{graphicx}
\usepackage{colortbl}
\definecolor{grey}{rgb}{0.9, 0.9, 0.9}

\newcommand{\coout}[1]{}
\newcommand{\co}   [1]{}
\newcommand{\todo} [1]{}
\newcommand{\wolf} [1]{}

\newcommand{\myparagraph}[1]{\vspace{-3mm}\paragraph{#1}}

\newcommand{\Ntot} {{{\bf N}_{\rm tot}}}

\newcommand{\ratelaw}{r}

\newcommand{\hminus} {h}

\newcommand{\enzymemetcost}{q^{\rm enz}}

\newcommand{\fluxcost}  {a}
\newcommand{\metbene}   {b^{\rm (c)}}

\newcommand{\cv}{{\bf c}}
\newcommand{\vv}{{\bf v}}
\newcommand{\yv}{{\bf y}}

\newcommand{\uv}{{\bf u}}

\newcommand{\kv}{{\bf k}}

\newcommand{\av}{{\bf a}}

\newcommand{\xiv}{{\boldsymbol \xi}}
\newcommand{\muv}{{\boldsymbol \mu}}

\newcommand{\pv}{{\bf p}}

\newcommand{\jv}{{\bf j}}

\newcommand{\sv}{{\bf s}}
\newcommand{\nv}{{\bf n}}

\newcommand{\Esc}     {\bf {\mathcal E}}

\newcommand{\Cmat}  {{\bf C}}

\usepackage{hyperref}
\hypersetup{colorlinks = true, allcolors = {blue}}

\definecolor{lightyellow}{rgb}{1,1,1}
\definecolor{orange}{rgb}{1,0.75,0.3}
\definecolor{brown}{rgb}{0.87,0.6,0.23}
\definecolor{pink}{rgb}{1,0.7,0.7} 
\definecolor{purple}{rgb}{0.8,0.7,1}
\definecolor{lightblue}{rgb}{0.65,0.7,1}
\definecolor{verylightblue}{rgb}{0.92,0.93,1}
\definecolor{darkcyan}{rgb}{0.,0.6,0.9}
\definecolor{darkblue}{rgb}{0.1,0,0.7}
\definecolor{darkmagenta}{rgb}{0.6,0,0.4}
\definecolor{c1}{rgb}{0,0,0.8}
\definecolor{c2}{rgb}{0.25,0.,0.6}
\definecolor{c3}{rgb}{0.6,0.,0.3} 
\definecolor{c4}{rgb}{1,0,0}
\definecolor{c5}{rgb}{.5,0.2,0}
\definecolor{lightblack}{rgb}{.5,0.1,0.2}

\renewcommand{\co}[1]   {}
\renewcommand{\coout}[1]{}
\renewcommand{\wolf}[1] {}
\renewcommand{\todo}[1] {}
\newcommand{\elad}[1]   {}
\newcommand{\arrenalt}[1]  {}
\newcommand{\avialt}[1]    {}
\newcommand{\arren}[1]  {}
\newcommand{\avi}   [1]  {}
\newcommand{\ron}[1]    {}
\newcommand{\dan}[1]    {}

\renewcommand{\myparagraph}[1]{}

\renewcommand{\Esc}{{\mathcal E}}
\renewcommand{\kcat}{k_{\rm cat}}

\renewcommand{\kcatplus}{k_{\rm cat}^+}
\renewcommand{\kcatplusl}{k_{{\rm cat},l}^+}
\renewcommand{\kcatminus}{k_{\rm cat}^-}

\newcommand{\Stol}{S_{\rm tol}}
\newcommand{\hessian}{\Hmat_\enzymemetcost}
\newcommand{\rrs}{\fluxcost_{\rm v}}
\newcommand{\rrspw}{\fluxcost_{\rm pw}}
\newcommand{\rrl}{\fluxcost_{v_l}}
\newcommand{\rrlmin}{\rrl^{\rm cat}}
\newcommand{\rrpw}  {\rrs^{\rm pw}}

\newcommand{\enzyme}{E}
\newcommand{\el}{\enzyme_{l}}
\newcommand{\enzymev}{{\boldsymbol \enzyme}}

\newcommand{\metabolitepolytope}{{\mathcal P}}

\newcommand{\vPW}{v_{\rm pw}}
\newcommand{\activity}{A}

\newcommand{\hel}   {\hminus_{E_l}}
\newcommand{\he}    {\hminus_{\rm E}}
\renewcommand{\enzymemetcost} {q}
\newcommand{\enzymemetcostl}  {\enzymemetcost_l}
\newcommand{\enzymemetcostmax}{\enzymemetcost^{\rm max}}
\newcommand{\concS}{s}
\newcommand{\concP}{p}
\newcommand{\vscaled}{v'}

\newcommand{\Hm}{{\bf H}}
\newcommand{\mytextonehalf}{{\frac{1}{2}}}

\newtheorem{corollary}[theorem]{Corollary}
\newcommand{\qedhere}{$\square$}

\begin{document}

\title{The protein cost of metabolic fluxes: \\ prediction from 
enzymatic rate laws and cost minimization}
\date{} 

\author{Elad Noor$^1$, Avi Flamholz$^2$, Arren Bar-Even$^3$, Dan
  Davidi$^4$, Ron Milo$^4$, Wolfram Liebermeister$^5$}

\maketitle

{\vspace{-5mm} \footnotesize
 $^1$Institute of Molecular Systems Biology,
Eidgen\"ossische Technische Hochschule Z\"urich, Switzerland,
$^2$Department of Molecular and Cellular Biology, University of
California, Berkely, California, United States of America, $^3$Max
Planck Institute for Molecular Plant Physiology, Golm, Germany,
$^4$Department of Plant Sciences, The Weizmann Institute of Science,
Rehovot, Israel, $^5$Institute of Biochemistry, Charit\'e --
Universit\"atsmedizin Berlin, Germany}

\begin{abstract}

Bacterial growth depends crucially on metabolic fluxes, which are
limited by the cell's capacity to maintain metabolic enzymes.  The
necessary enzyme amount per unit flux is a major determinant of
metabolic strategies both in evolution and bioengineering. It depends
on enzyme parameters (such as $\kcat$ and $K_M$ constants), but also
on metabolite concentrations.  Moreover, similar amounts of different
enzymes might incur different costs for the cell, depending on
enzyme-specific properties such as protein size and half-life.  Here,
we developed enzyme cost minimization (ECM), a scalable method for
computing enzyme amounts that support a given metabolic flux at a
minimal protein cost.  The complex interplay of enzyme and metabolite
concentrations, e.g. through thermodynamic driving forces and enzyme
saturation, would make it hard to solve this optimization problem
directly. By treating enzyme cost as a function of
metabolite levels, we formulated ECM as a numerically tractable,
convex optimization problem.  Its tiered approach allows for building
models at different levels of detail, depending on the amount of
available data.  Validating our method with measured metabolite and
protein levels in \emph{E. coli} central metabolism, we found typical
prediction fold errors of 3.8 and 2.7, respectively, for the two kinds
of data. ECM can be used to  predict enzyme levels and protein cost in
natural and engineered pathways, establishes a direct connection between
protein cost and thermodynamics, and provides a physically plausible
and computationally tractable way to include enzyme kinetics into
constraint-based metabolic models, where kinetics have usually been
ignored or oversimplified.
\end{abstract}


\section{Introduction}

\myparagraph{\ \\Fluxes and enzyme levels as a research question} The
biochemical world is remarkably diverse, and this is only the tip of
the iceberg as new pathways and chemicals are still discovered
routinely. Even for model organisms, such as \textit{E. coli}, the
exhaustive mapping of their metabolic network is (almost) complete
only on the stoichiometric level, but far from perfect when it comes
to our understanding of metabolic fluxes, how they are dynamically
realized, and how they support cell fitness \cite{szzh:12}.
Furthermore, the rational designing of novel and efficient metabolic
pathways is still a big challenge and metabolic engineering projects
require considerable efforts even for relatively simple metabolic
tasks.  Among the different possible criteria \cite{bnlm:10}, one key
to understanding the choices of metabolic routes, both in naturally
evolved and engineered organisms, may be enzyme cost.  Quite often,
cells use metabolic pathways in ways that seem irrational, such as in
the case of aerobic fermentation (known as the Crabtree effect in
yeast or the Warburg effect in cancer cells \cite{wapn:24}).  However,
apparently yield-inefficient fluxes can sometimes be explained by an
economic use of enzyme resources \cite{fnbl:13,bhoz:15}.  Pathway
structures that require too much enzyme per unit flux will be
outcompeted during evolution and will not be efficient for use in
biotechnological applications.  Thus, a quantitative analysis of
resource investment in enzyme production, predicting the amount of
enzyme needed to support a given flux, as well as the optimal enzyme
levels along pathways, would be valuable steps towards a rational
design of metabolic pathways.

To understand why specific enzymes or pathways occupy larger or
smaller areas of the proteome \cite{lnfd:14}, we could proceed in two
steps, determining first the metabolic fluxes and then enzyme levels
needed to realise these fluxes. Metabolic fluxes can be measured based
on isotope-labeled tracer experiments in combination with
computational modeling.  Methods for flux prediction \emph{ab initio}
rely on mechanistic aspects (chemical mass balances and kinetics) and
economic aspects (cost and benefit of pathway fluxes) and combine them
in different ways.  Constraint-based methods like Flux Balance
Analysis (FBA) determine fluxes by requiring steady states -- i.e.,
fluxes must be such that internal metabolite levels remain constant in
time -- and assuming that natural selection maximizes some benefit
function (e.g., maximal yield of biomass). Different optimality
criteria for fluxes can be combined in a multi-objective optimization
\cite{scks:07,szzh:12}.  In some cases, the second law of
thermodynamics is used to put further constraints on fluxes or
metabolite levels \cite{belq:02,bbcq:04,faqb:05,fmsy:11}. Some
extensions of FBA \cite{hjbh:06,hebh:07,hohh:07} use metabolite
log-concentrations as extra variables and constrain fluxes to flow
only in the direction of thermodynamic driving forces, i.e., towards
lower chemical potentials. This links flux directions to reactant
concentrations, and by including bounds on metabolite levels, flux
directions become restricted.  These links between fluxes and
metabolite concentrations hold independently of specific reaction
kinetics.  The relationship between fluxes and metabolite
concentrations can be used also in opposite direction -- i.e.  given
all flux directions, certain metabolite profiles can be excluded
\cite{hohh:07}. The set of feasible metabolite profiles can be
depicted as a polytope in the space of metabolites'
log-concentrations.  To further narrow down the metabolite
concentration profiles, the Max-min Driving Force (MDF) method
\cite{nbfr:14} chooses profiles that ensure sufficient driving forces,
thus keeping reactions distant from chemical equilibrium.

\begin{figure}[t!]
  \begin{center}
\parbox{15.5cm}{
(a) \hspace{4cm} (b)\\
\parbox{17cm}{
    \includegraphics[width=4.5cm]{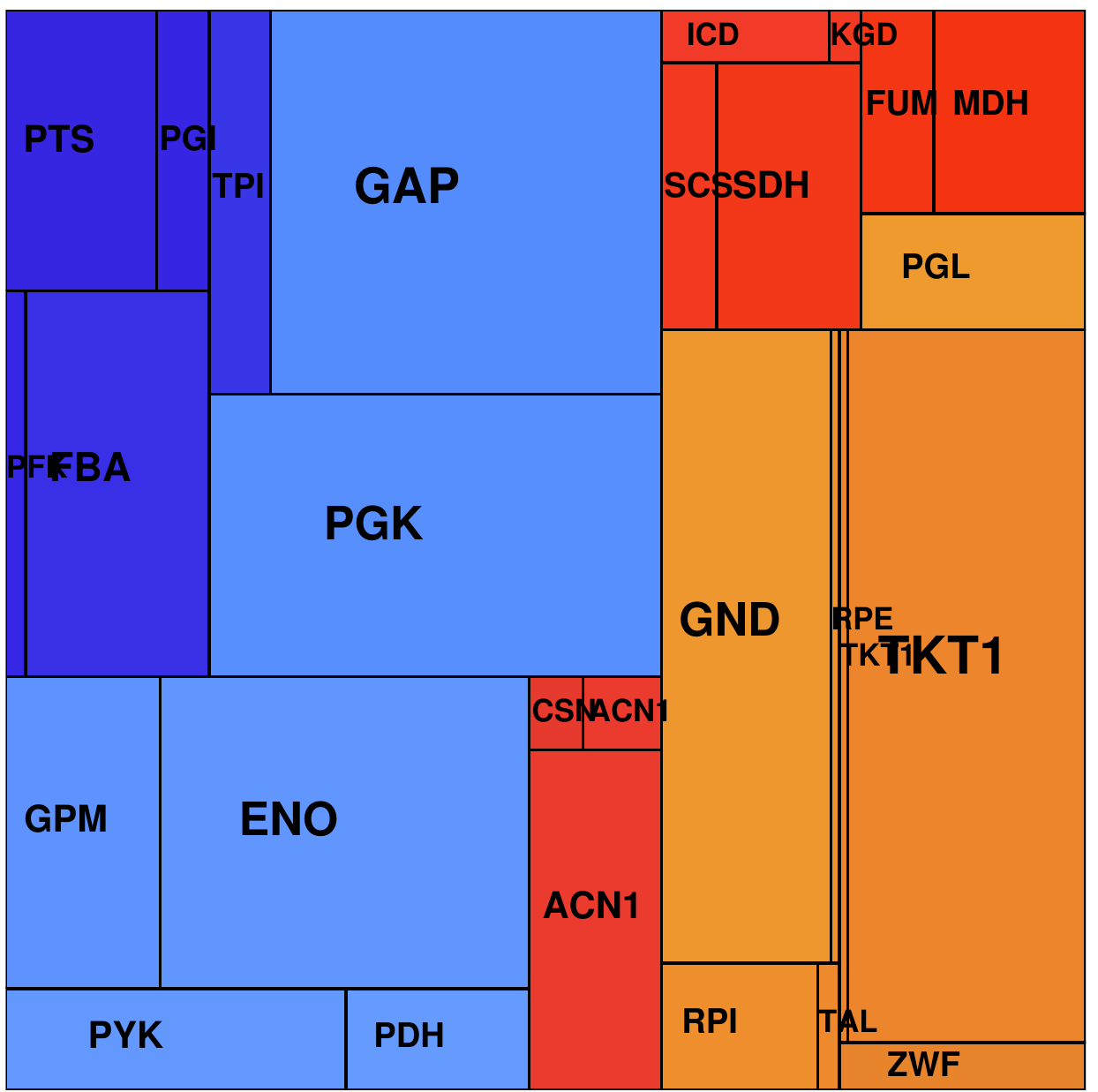}
  \hspace{3mm}
    \includegraphics[width=10.5cm]{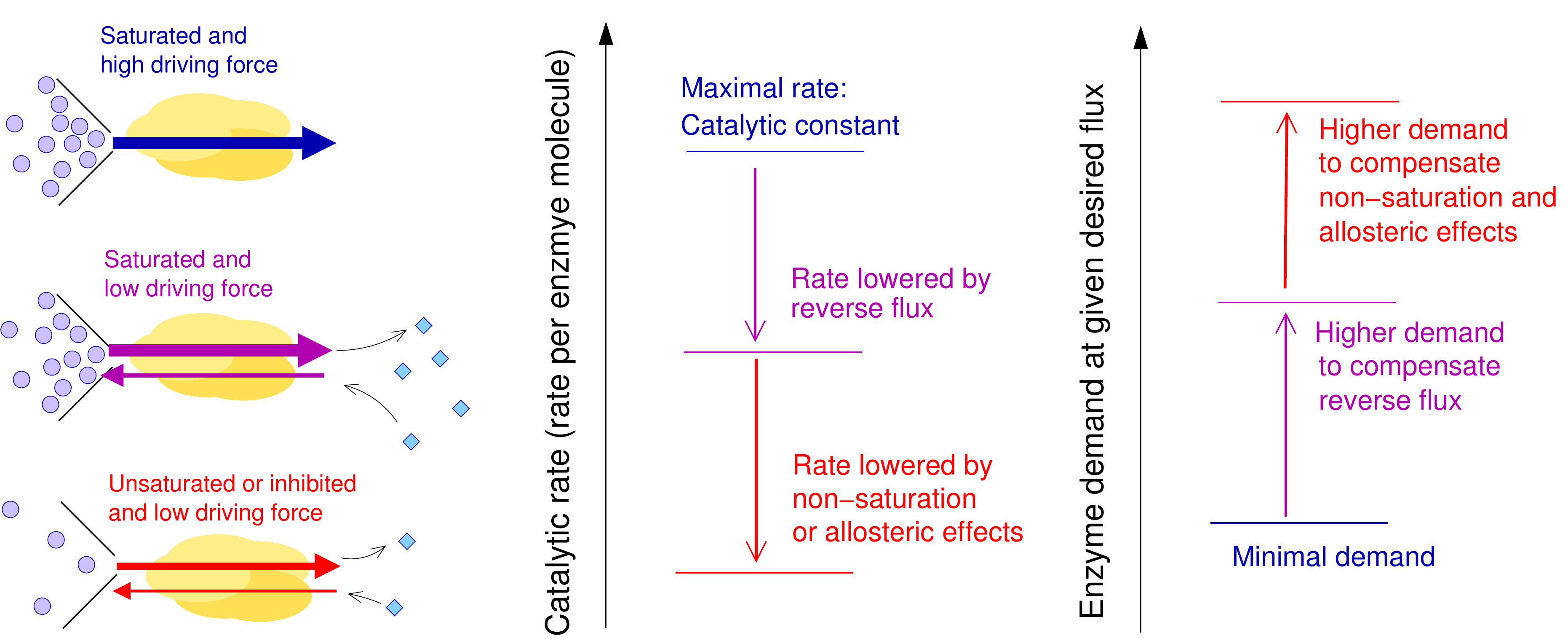}}
}  \end{center}
  \caption{Enzyme cost in metabolism.  (a) Measured enzyme levels in
    \emph{E.~coli} central metabolism (molecule counts displayed as
    rectangle areas).  Colors correspond to the network graphics in
    Figure \ref{fig:EcoliPredictions}. To predict such protein levels,
    and to explain the differences between enzymes, we start from
    known metabolic fluxes and assume that these fluxes are realised
    by a cost-optimal distribution of enzyme levels.  (b)
    Enzyme-specific flux depends on a number of physical factors.
    Under ideal conditions, an enzyme molecule catalyses its reaction
    at a maximal rate given by the enzyme's forward catalytic constant
    (top left). The rate is reduced by microscopic reverse fluxes
    (center left) and by incomplete saturation with substrate (causing
    waiting times between reaction events) or by allosteric inhibition
    or incomplete activation (bottom left).  With lower catalytic
    rates (center), realizing the same metabolic flux requires larger
    amounts of enzyme (right).}
  \label{fig:efficiencies}
\end{figure}

\myparagraph{Enzyme cost} Typically, constraint-based methods do well
in defining a space of feasible fluxes and assessing their benefits,
but much less in predicting the necessary enzyme levels and the cost
of making and maintaining the enzymes.  Thus, flux prediction and in
designing efficient pathways in bioengineering, how can we estimate
the protein demand of a reaction or pathway, needed to sustain a
desired flux? It is often assumed that the flux in a
reaction is proportional to the enzyme level. FBA methods use this
assumption to translate enzyme expression, as a proxy for protein
burden, into flux bounds or linear flux cost functions
\cite{schp:08}. For practical reasons (computational tractability and
lack of detailed knowledge), flux cost are often represented by the
sum of absolute fluxes \cite{holz:04,lhcl:10} To obtain better proxies
of protein demand or related cellular burdens, fluxes have been
weighted by ``flux burdens'' that account for different catalytic
constants $\kcat$ \cite{bvem:07,bnlm:10}, protein size and lifetime
\cite{horh:11}, or equilibrium constants \cite{holz:04}.
\myparagraph{Factors affecting enzyme demand} In reality, however,
enzyme demand does not only depend on fluxes, but also on metabolite
levels, which in turn are determined by the non-linear kinetics of all
enzymes. Therefore, it is not only the choice of numerical cost
weights, but the very relation between enzyme amounts and fluxes that
needs to be clarified. \arrenalt{nice logic}

For a simple estimate, we can assume that each enzyme molecule works
at its maximal rate, the catalytic constant $\kcat$. In this case,
enzyme demand is given by the flux divided by the catalytic constant
\cite{bvem:07,bnlm:10}. To translate enzyme demand into cost, the
different sizes or effective lifetimes of enzymes can be considered
\cite{horh:11}. The notion of Pathway Specific Activity \cite{bnlm:10}
applies this principle to measure the efficiency of entire pathways
(while assuming that enzyme levels are optimally distributed), and
provides a direct way to compare between alternative pathways.
However, by assuming that enzymes operate at their maximal capacity,
we underestimate the true enzyme demand (see Figure
\ref{fig:efficiencies}).  Enzymes typically do not operate at full
capacity. This is due to backward fluxes, incomplete substrate
saturation, allosteric regulation, and regulatory post-translational
modifications. Below, we will refer to allosteric regulation only, but
other types of posttranslational regulation, e.g., by phosphorylation,
could be treated similarly. The relative backward fluxes depend on the
ratio between product and substrate concentrations, called mass-action
ratio. Whenever the mass-action ratio deviates from its equilibrium
value, called equilibrium constant, this deviation can be
conceptualized as a thermodynamic driving force. A driving force
determines the relative backward flux and thus affects reaction
kinetics and enzymatic efficiency \cite{liuk:10,nflb:13}.  With
smaller forces, the relative backward flux becomes larger, enzyme
usage becomes less efficient, and enzyme demand increases
\cite{tnah:13,fnbl:13} -- a situation that, in models, can be avoided
by applying the MDF method.  In fact, a cost increase due to backward
fluxes can be included in the principle of minimal fluxes in FBA
\cite{holz:04}. However, metabolites do not just affect the
thermodynamic forces, as acknowledged in thermodynamic FBA, but affect
kinetics as reactants and allosteric effectors.  While the relative
backward fluxes depend on thermodynamic forces, the forward flux
depends on the availability of substrate molecules. At sub-saturating
substrate levels, enzyme molecules spend some time waiting for
substrate molecules, thus reducing their average catalyzed
flux. Likewise, enzyme saturation with product can reduce the fraction
of enzyme molecules available for catalysis.

\myparagraph{In enzyme cost, enzyme and metabolite profiles must be
  considered together} Thus, converting metabolic fluxes into enzyme
demand can be difficult because enzymes may not realize their maximal
capacity. Since the decrease in enzyme efficiency depends mostly on
metabolite concentrations, enzyme and metabolite profiles must be
considered together. However, this quickly becomes a cyclic inference
problem because steady-state metabolite levels depend again on enzyme
profiles.  Since many metabolites (e.g., co-factors such as ATP)
participate in several pathways, enzyme demands may be coupled across
the entire metabolic network. Moreover, there may be many possible
enzyme and metabolite profiles that realize the same flux
distribution.  
To determine a single solution, one can make the assumption that the
most reasonable enzyme profile for realizing a given flux is the one
with the minimal associated cost.  This assumption may be justified if
we focus on biological systems shaped by evolution, or on engineered
pathways that should be efficient. A direct optimisation of enzyme
levels can be difficult, but there is a tractable approach in which
metabolite levels are treated as free variables, which determine the
enzyme levels, and therefore enzyme cost.  This approach, together with a minimization of
metabolite concentrations \cite{sche:91}, has been previously applied
to predict enzyme and metabolite levels in metabolic systems
\cite{tnah:13} and to compare structural variants of glycolysis by
their enzyme cost of ATP production \cite{fnbl:13}.

However, to make such optimization schemes generally applicable, some
open problems need to be addressed.  First, our knowledge of the kinetic rate
laws and parameters contains large gaps for the vast majority of
enzymes \cite{bnsl:11}, and combining rate constants from different
sources may lead to inconsistent models
\cite{eune:10,smal:13}. 
Second, the optimization problem may be numerically hard for large
networks and realistic rate laws.  To turn enzyme cost minimization
 into a generally applicable method, we address a number
of questions: (i) When setting up models for enzyme cost prediction,
how can we deal with missing, uncertain, or conflicting data on rate
constants?  Are there approximations, for instance based on
thermodynamics, that yield good predictions with fewer input
parameters?  (ii) How do factors such as $\kcat$ value, driving force,
or rate law affect enzyme demand, and how do they shape the optimal
metabolic state?  (iii) How can enzyme optimisation be formulated as a
numerically tractable optimality problem?  Existing approaches for
flux and enzyme prediction have focused on different aspects
(stationary state, energetics, kinetics, enzyme or flux costs,
molecular crowding).  The new approach, which uses kinetics to
translate fluxes into enzyme demand, shows how these approaches are
logically related, and how heuristic assumptions by other methods,
e.g.~an avoidance of small driving forces, follow from enzyme economy
as a general principle.  We show that enzyme cost minimization is closely related to cost-benefit approaches, which
treat cell fitness as a function of enzyme levels
\cite{reic:83,klhh:02,lksh:04,mbrt:09,zabl:13,laht:13}.  Some general
results of these approaches, e.g., relationships between enzyme costs
and metabolic control coefficients, can be recovered.

\section{Results}

\subsection{Enzyme cost landscape of a metabolic pathway}

\begin{figure}[t!]
  \begin{center}
    \begin{tabular}{l}
      (a)\\[1mm]
      \includegraphics[width=2.5cm]{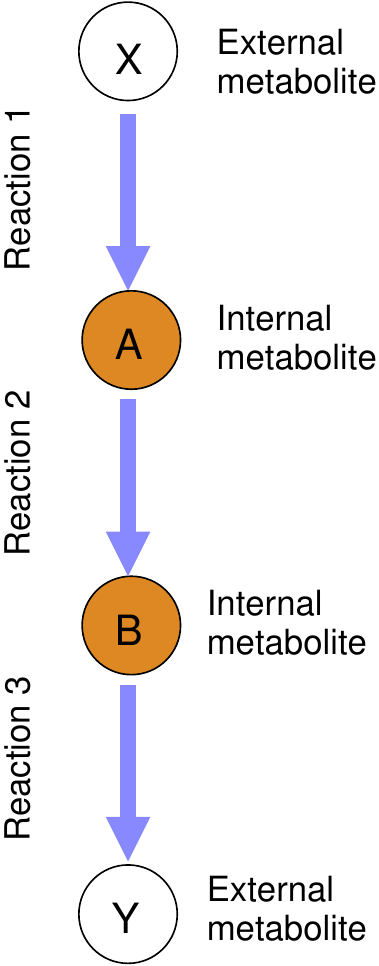}
    \end{tabular}
    \begin{tabular}{llll}
      (b)& (c)& (d) & \\[1mm]
      \includegraphics[height=3cm]{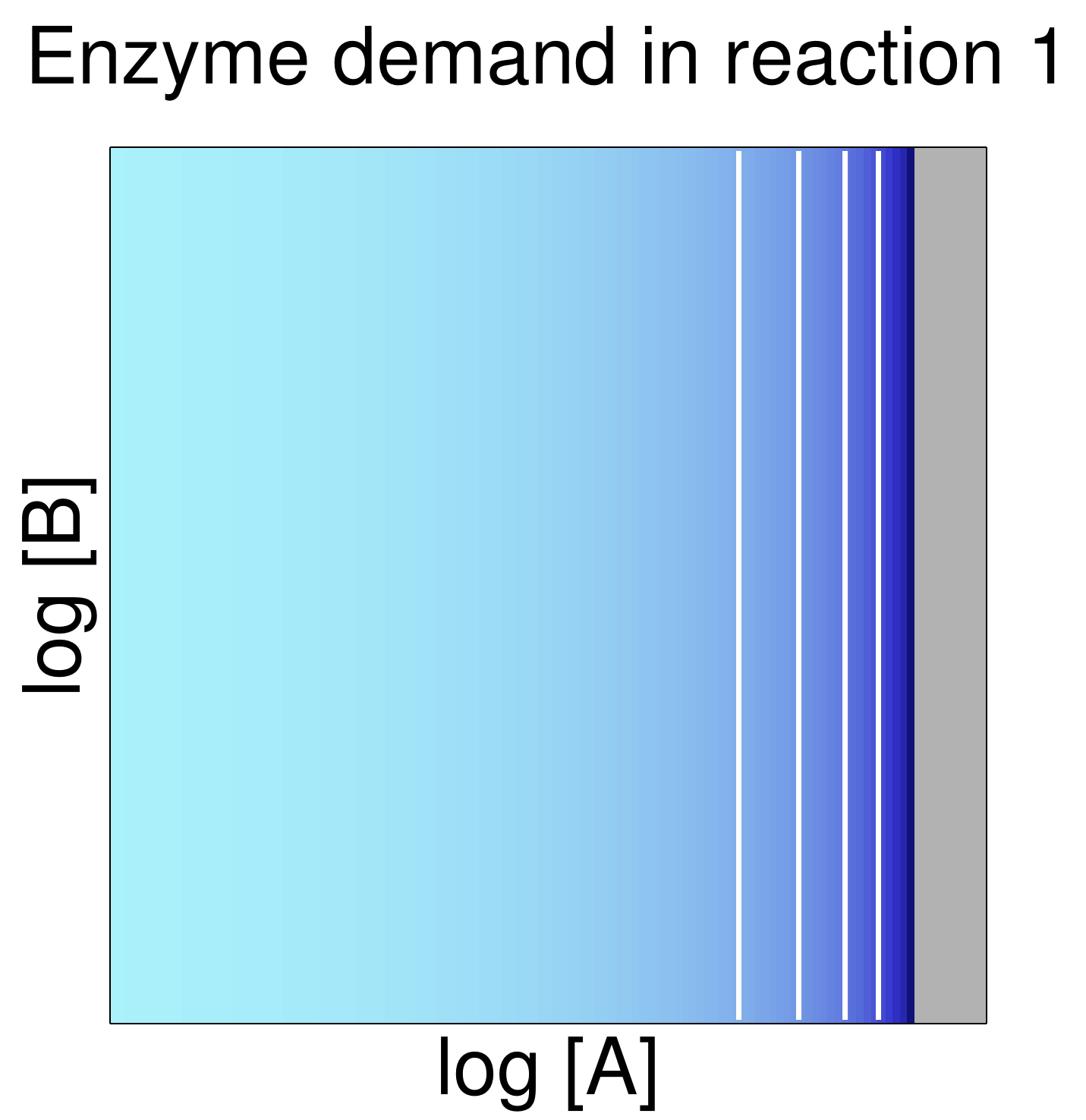} & 
      \includegraphics[height=3cm]{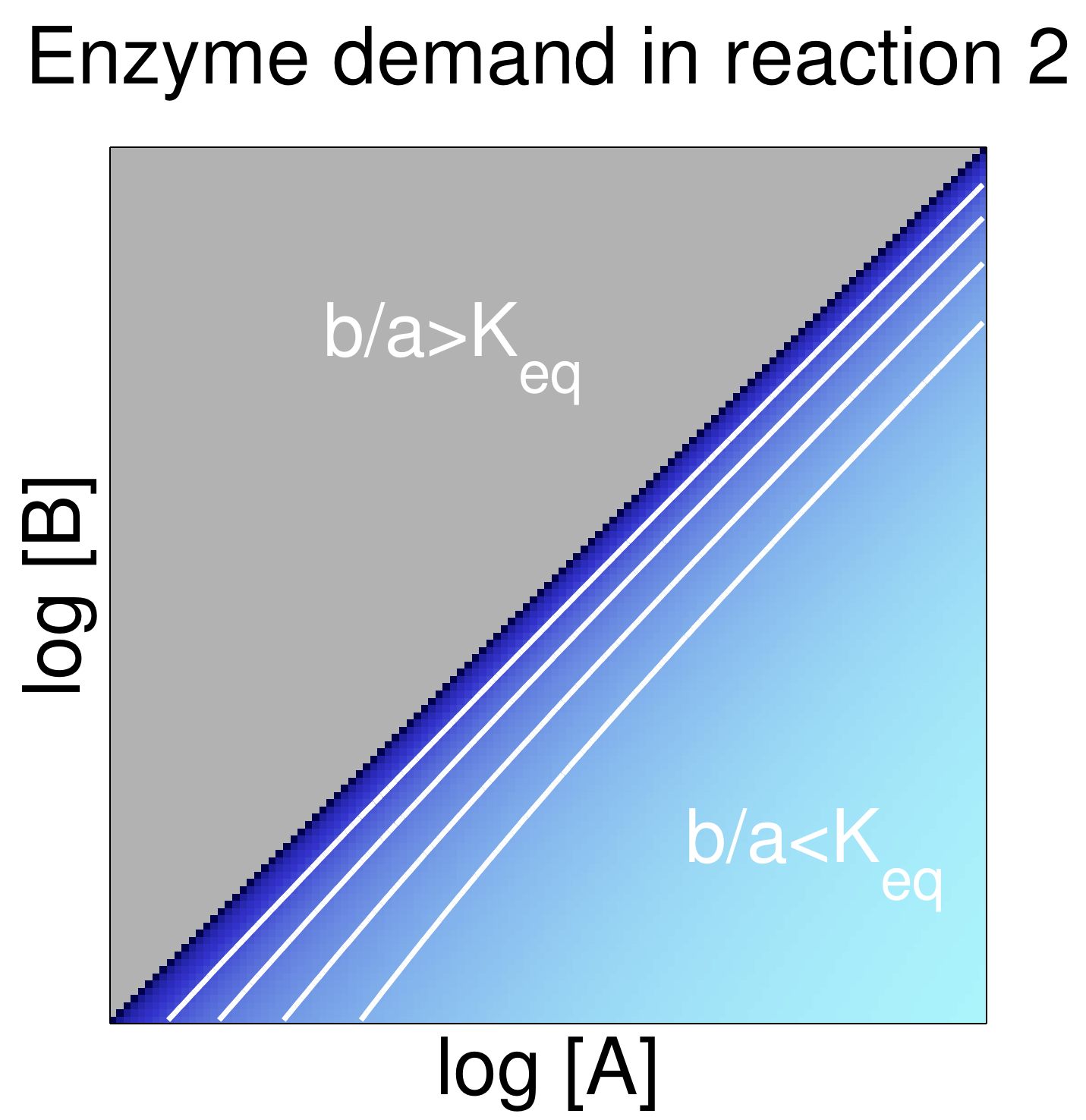}&
      \includegraphics[height=3cm]{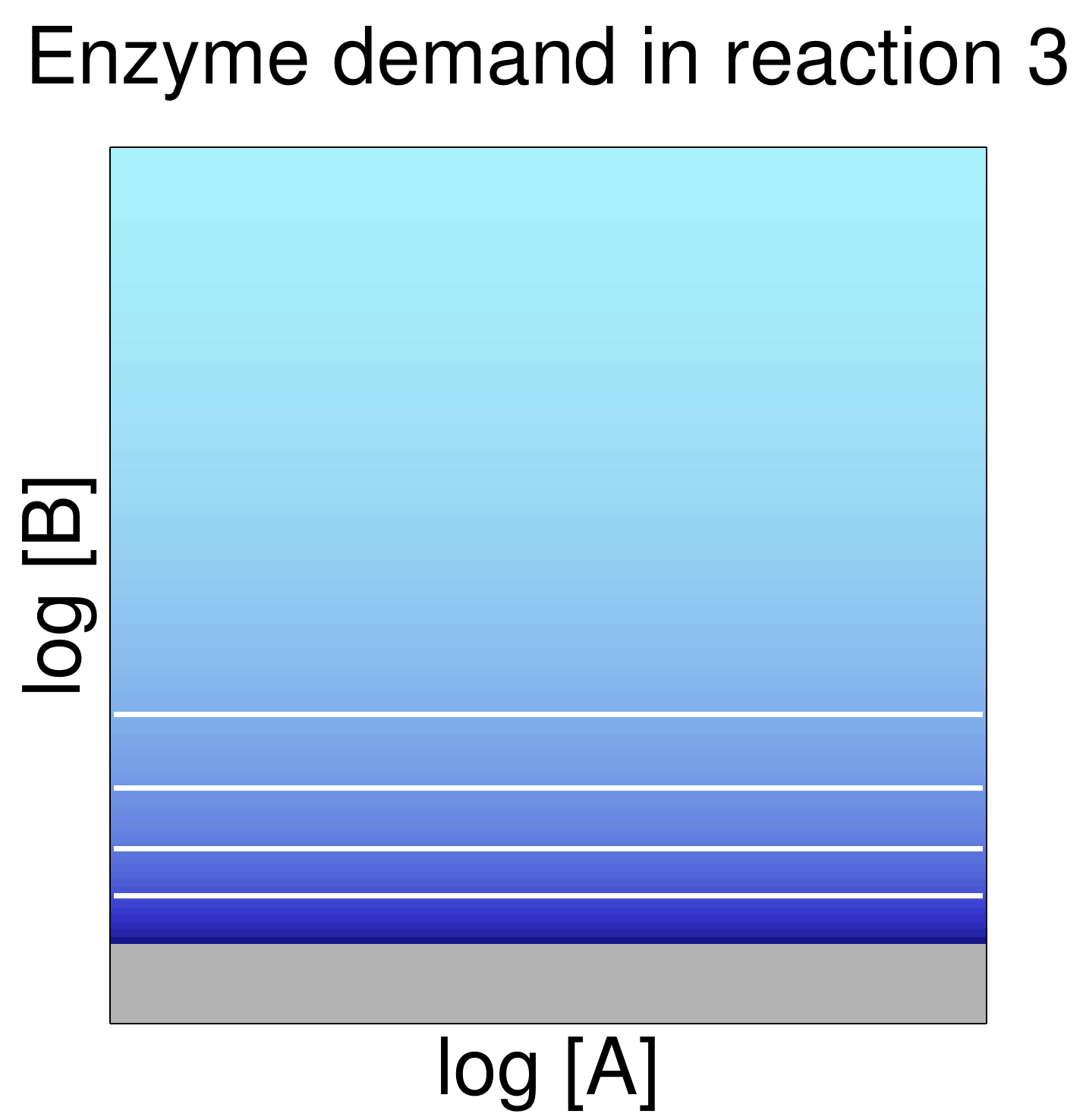}\\
(e) & (f) & (g) & \\[1mm]
      \includegraphics[height=3cm]{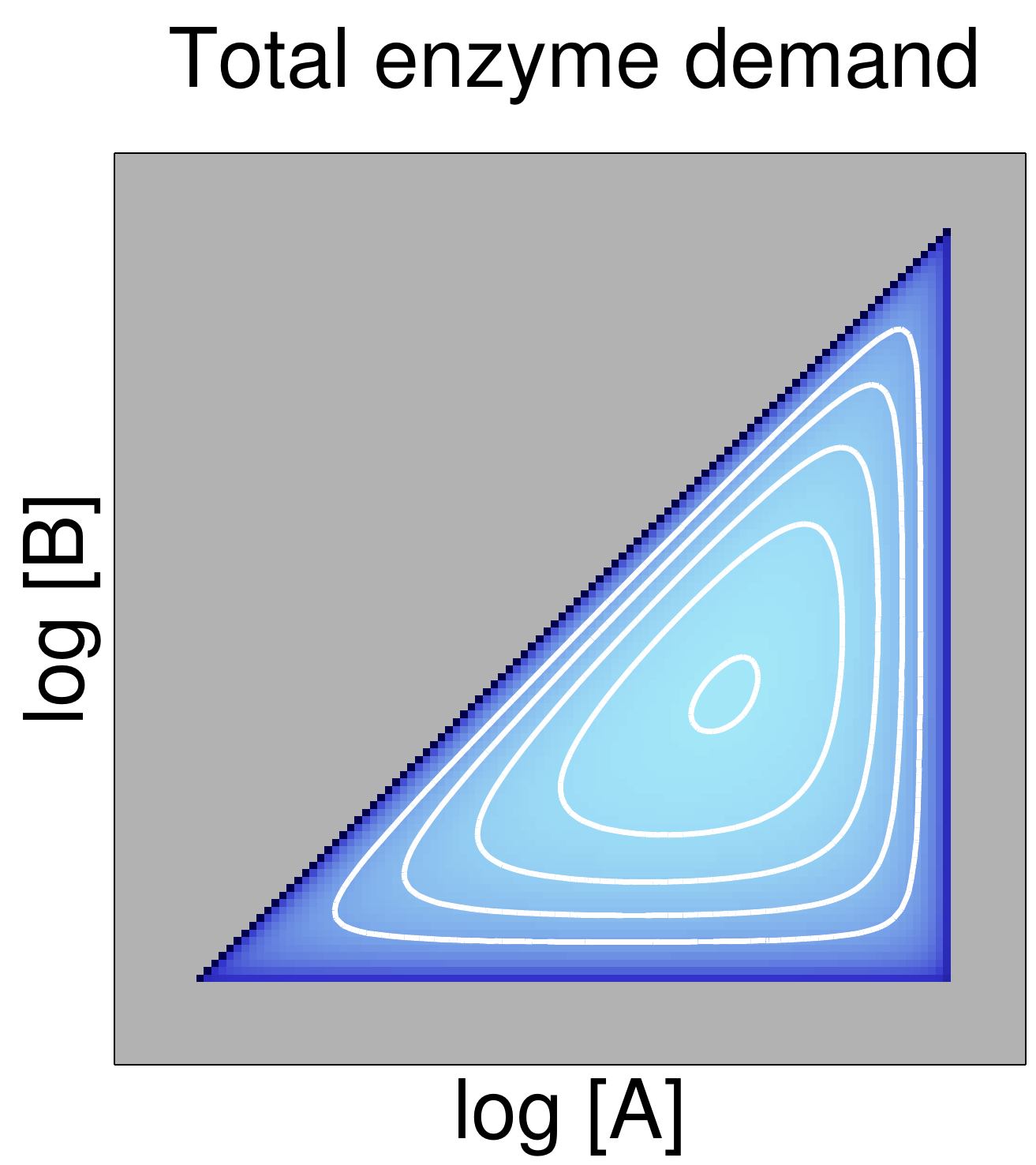}&
      \includegraphics[height=3cm]{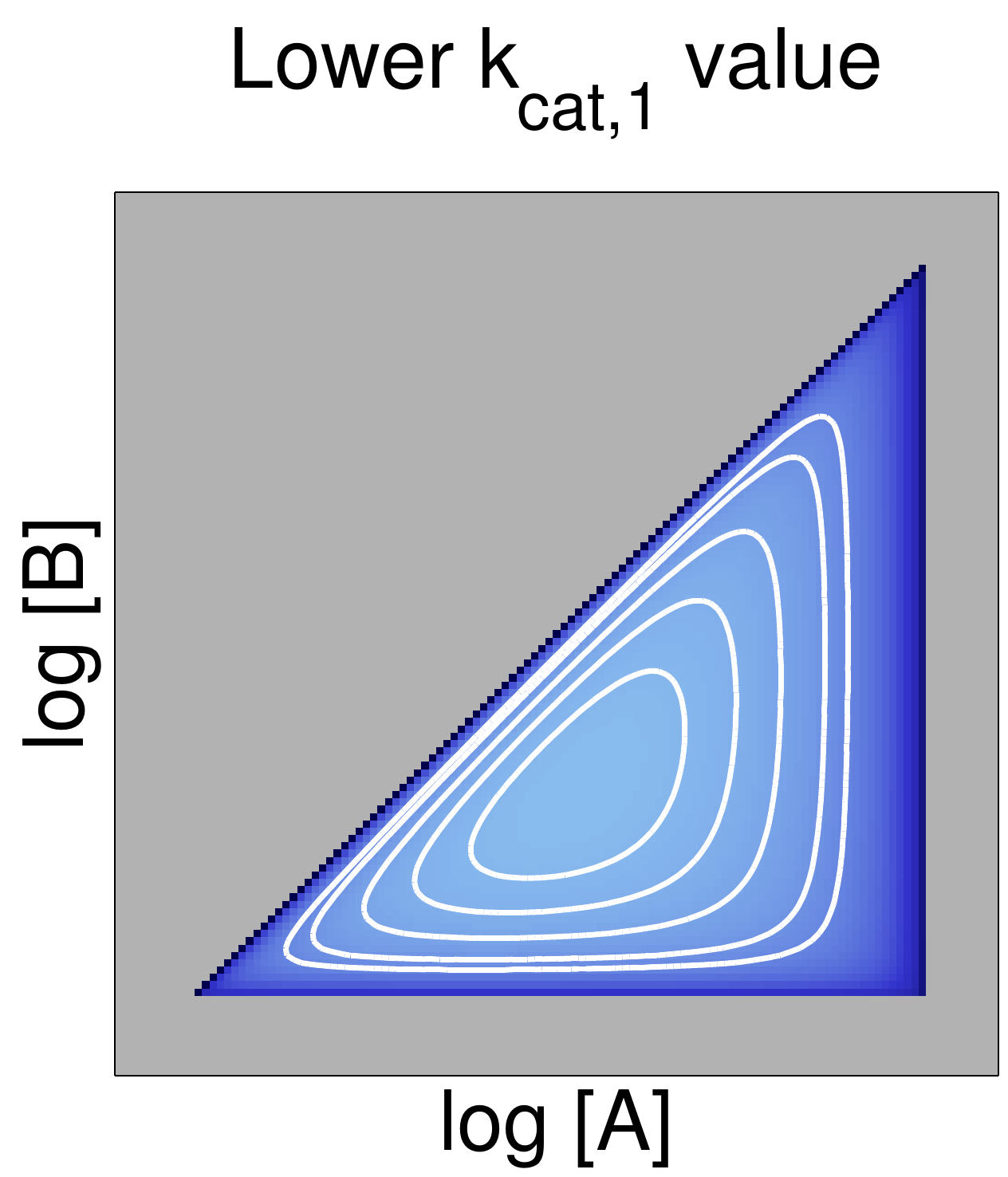} & 
      \includegraphics[height=3cm]{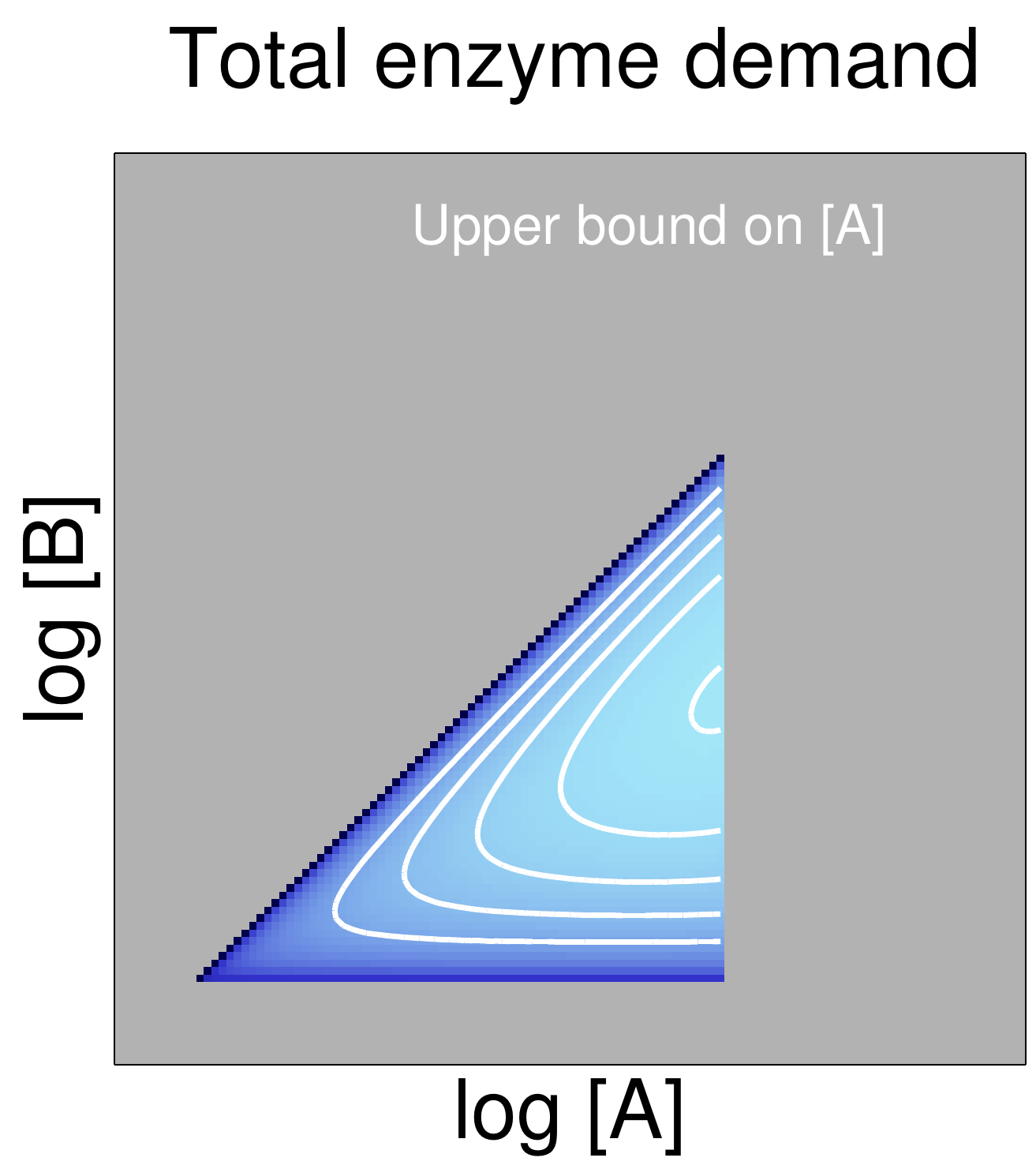} & 
      \includegraphics[height=2.8cm]{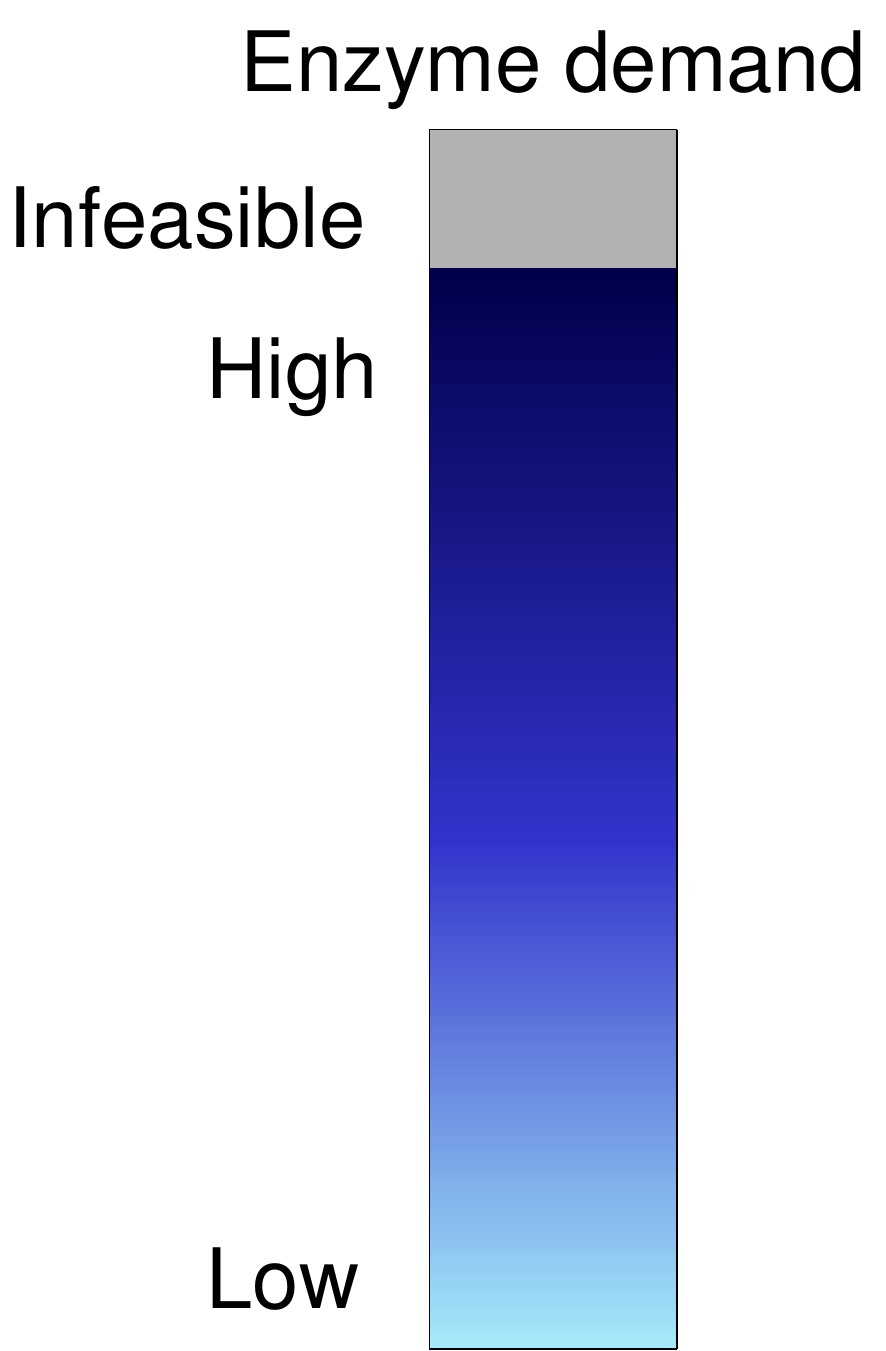}
 \end{tabular}
 \caption{Enzyme demand in a metabolic pathway.  (a) Pathway with
   reversible Michaelis-Menten kinetics (equilibrium constants,
   catalytic constants, and $\km$ values are set to values of 1, [A]
   and [B] denote the variable concentrations of intermediates A and B
   in mM).  The external metabolite levels [X] and [Y] are fixed.
   Plots (b)-(d) show the enzyme demand of reactions 1, 2, and 3 at
   given flux $v=1$ according to Eq.~(\ref{eq:mmratelawdemand}).  Grey
   regions represent infeasible metabolite profiles.  At the edges of
   the feasible region (where A and B are close to chemical
   equilibrium), the thermodynamic driving force goes to zero.  Since
   small forces must be compensated by high enzyme levels, edges of
   the feasible region are always dark blue.  For example, in reaction
   1 (panel (b)), enzyme demand increases with the level of A (x-axis)
   and goes to infinity as the mass-action ratio $[A]/[X]$ approaches
   the equilibrium constant (where the driving force vanishes). (e)
   Total enzyme demand, obtained by summing all enzyme levels.  The
   metabolite polytope -- the intersection of feasible regions for all
   reactions -- is a triangle, and enzyme demand is a cup-shaped
   function on this triangle.  The minimum point defines the optimal
   metabolite levels and optimal enzyme levels. (f) As the $\kcat$
   value of the first reaction is lowered by a factor of 5, states
   close to the triangle edge of reaction 1 become more expensive and
   the optimum point is shifted away from the edge. (g) The same model
   with a physiological upper bound on the concentration [A]. The
   bound defines a new triangle edge.  Since this edge is not caused
   by thermodynamics, it can contain an optimum point, in which driving forces
   are far from zero and enzyme costs are kept low.}
  \label{fig:fourchain}
  \end{center}
\end{figure}

\myparagraph{\ \\Enzyme cost minimization} Given a
pathway flux profile and a kinetic model of the pathway, one can
predict the enzyme demand by assuming that cells minimize the enzyme
cost in that pathway. A reaction rate $v = \enzyme \cdot
\ratelaw(\cv)$ depends on enzyme level $\enzyme$ and metabolite
concentrations $c_i$ through the enzymatic rate law, $\ratelaw(\cv)$.
If the metabolite levels were known, we could directly compute enzyme
demands $\enzyme = v / \ratelaw(\cv)$ from fluxes, and similarly
calculate the flux-specific enzyme demand $\enzyme/v =
1/\ratelaw(\cv)$.  However, metabolite levels are often unknown and
vary between experimental conditions. Therefore, there can be many
solutions for $\enzyme$ and $\cv$ realizing one flux distribution.  To
select one of them, we employ an optimality principle: we define an
enzyme cost function (for instance, total enzyme mass) and choose the
enzyme profile with the lowest cost while restricting the metabolite
levels to physiological ranges and imposing some thermodynamic
constraints. As we shall see below, the solution is in many cases unique. Let us
demonstrate this procedure with a simple example (Figure
\ref{fig:fourchain} (a)).  In the pathway $X \rightleftharpoons A
\rightleftharpoons B \rightleftharpoons Y$, the external metabolite
levels [X] and [Y] are fixed and given, while the intermediate levels
[A] and [B] need to be found.  As rate laws for all three reactions,
we use reversible Michaelis-Menten (MM) kinetics
\begin{eqnarray}
 \label{eq:mmratelaw}
  v = \enzyme\, \frac{\kcatplus\, \concS/\kmS - \kcatminus \, \concP /\kmP}
{1+ \concS/\kmS+\concP/\kmP}
\end{eqnarray}
with enzyme level $\enzyme$, substrate and product levels $\concS$ and
$\concP$, turnover rates $\kcatplus$ and $\kcatminus$, and Michaelis
constants $\kmS$ and $\kmP$.  In kinetic modeling, steady-state
concentrations would usually be obtained from given enzyme levels and
initial conditions through numerical integration. Here, instead, we
fix a desired pathway flux $v$ and compute the enzyme
demand as a function of metabolite levels:
\begin{eqnarray}
 \label{eq:mmratelawdemand}
  \enzyme(\concS,\concP,v) = v\, \frac
{1+ \concS/\kmS+\concP/\kmP}
{\kcatplus\, \concS/\kmS - \kcatminus \, \concP /\kmP}.
\end{eqnarray}
Figure \ref{fig:fourchain} shows how the enzyme demand in each
reaction depends on the logarithmic reactant concentrations.  To
obtain a positive flux, substrate levels $\concS$ and product levels
$\concP$ must be restricted: for instance, to allow for a positive
flux in reaction 2, the rate law numerator $\kcatplus\, [A]/\kmS -
\kcatminus \, [B] /\kmP$ must be positive. This implies that $[B]/[A]
< \keq$ where the reaction's equilibrium constant $\keq$ is determined
by the Haldane relationship, $\keq = (\kcatplus / \kcatminus) \cdot
(\kmP / \kmS)$.  With all model parameters set to 1, we obtain the
constraint $[B]/[A] < 1$, i.e., $\ln [B]-\ln [A]<0$, putting a
straight boundary on the feasible region (Figure \ref{fig:fourchain}
(c)). Close to chemical equilibrium ($[B]/[A] \approx \keq$), the
enzyme demand $\enzyme_2$ approaches infinity. Beyond that ratio
($[B]/[A] > \keq$) no positive flux can be achieved (grey region).
Such a threshold exists for each reaction (see Figure
\ref{fig:fourchain} (b)-(d)).  The remaining feasible metabolite
profiles form a triangle in log-concentration space, which we call
\emph{metabolite polytope} $\metabolitepolytope$ (Figure
\ref{fig:fourchain} (e)), and Eq.~(\ref{eq:mmratelawdemand}) yields
the total enzyme demand $\enzyme_{\rm tot} =
\enzyme_1+\enzyme_2+\enzyme_3$, as a function on the metabolite
polytope.  The demand increases steeply towards the edges and becomes
minimal in the center.  The minimum point marks the optimal metabolite
profile, and via Eq.~(\ref{eq:mmratelawdemand}) we obtain the
resulting optimal enzyme profile.

\myparagraph{Influence of driving forces} The metabolite polytope and
the large enzyme demand at its boundaries follow directly from
thermodynamics.  To see this, we consider the unitless
\emph{thermodynamic driving force} $\Theta = -\Delta_{\rm r} G'/RT$
\cite{beqi:07} derived from the reaction Gibbs free energy
$\Delta_{\rm r} G'$. The thermodynamic force can be written as $\Theta = \ln
\frac{\keq}{[B]/[A]}$, i.e., the driving force is positive whenever
$[B]/[A]$ is smaller than $\keq$, and it vanishes if $[B]/[A] =
\keq$. How is this force related to enzyme cost?  A reaction's net
flux is given by the difference $v=v^{+}-v^{-}$ of forward and
backward fluxes, and the ratio $v^{+}/v^{-}$ depends on the driving
force as $v^{+}/v^{-} = \e^{\Theta}$.  Thus, only a fraction $v/v^{+}
= 1-\e^{-\Theta}$ of the forward flux acts as a net flux, while the
remaining forward flux is cancelled by the backward flux (SI Figure
\ref{fig:fluxfractions}).  Close to chemical equilibrium, where the
mass-action ratio $[B]/[A]$ approaches the equilibrium constant
$\keq$, the driving force goes to zero, the reaction's backward flux
increases, and the flux per enzyme level drops.  This is what happens
at the triangle edges in Figure \ref{fig:fourchain}: a reaction
approaches chemical equilibrium, the driving force $\Theta$ goes to
zero, and large enzyme amounts are needed for compensation.  Exactly
on the edge, the driving force vanishes and no enzyme level, no matter
how large, can support a positive flux.  The quantitative cost depends
on model parameters: for example, by lowering a $\kcat$ value, the
cost for the enzyme increases at the boundary becomes steeper and the
optimum point is shifted away from the boundary (see Figure
\ref{fig:fourchain} (f) and SI Figure \ref{fig:single_metabolite}).

\subsection{Enzyme cost as a function of metabolite profiles}

 The prediction of optimal metabolite and enzyme levels can be
 extended to models with general rate laws and complex network
 structures.  In general, enzyme demand depends not only on driving
 forces and $\kcat$ values, but also on the kinetic rate law, which
 includes $\km$ values and allosteric regulation.  Thus, we need to
 model these factors and approximate them when kinetic information is
 missing.  \myparagraph{Separable rate laws} The rate of a reaction
 depends on enzyme level $\enzyme$, forward catalytic constant
 $\kcatplus$ (i.e. the maximal possible forward rate per unit of
 enzyme, in s$\inv$), driving force (i.e., the ratio of forward and
 backward fluxes), and on kinetic effects such as substrate saturation
 or allosteric regulation. 
 If all active fluxes are positive, reversible rate laws like the
 Michaelis-Menten kinetics in Eq.~(\ref{eq:mmratelaw}) can be
 factorized as \cite{nflb:13}
\begin{eqnarray}
\label{eq:factorised}
  v &=& \enzyme \cdot \kcatplus \cdot \eta^{\rm th} \cdot \eta^{\rm kin}.
\end{eqnarray}
With some rate laws, $\eta^{\rm kin}$ can be further subdivided into
$\eta^{\rm kin} = \eta^{\rm sat} \cdot \eta^{\rm reg}$, where
$\eta^{\rm reg}$ refers to certain types of allosteric regulation (see
example in Box 1). Negative
fluxes, which would complicate our formulae, can be avoided by
orienting the reactions in the direction of fluxes.  The reversible
Michaelis-Menten rate law Eq.~(\ref{eq:mmratelaw}), for instance, can
be written in this separable form \cite{nflb:13}: 

\begin{eqnarray}
\label{eq:factorisedMM1}
 v = \enzyme \,\kcatplus \,\frac{\concS/\kmS\,\left(1-\frac{\kcatminus}{\kcatplus}\frac{\concP /\kmP}{\concS/\kmS}\right)}{1+ \concS/\kmS+\concP/\kmP} = 
\enzyme \,\kcatplus \,
\underbrace{\left(1-\frac{\kcatminus}{\kcatplus}\frac{\concP /\kmP}{\concS/\kmS}\right)}_{\eta^{\rm th}}\,
\underbrace{\frac{\concS/\kmS}{1+ \concS/\kmS+\concP/\kmP}}_{\eta^{\rm kin} },
\end{eqnarray} 
and similar factorizations exist for reactions of any stoichiometry
(see SI \ref{sec:SIratelaws}). The term $\enzyme \cdot \kcatplus$
describes the maximal reaction velocity, which is reduced, depending
on metabolite levels, by condition-specific factors $\eta^{\rm th}$,
$\eta^{\rm sat}$ and $\eta^{\rm reg}$ (see Fig
\ref{fig:efficiencies}b), accounting for backward fluxes, incomplete
substrate saturation, saturation with product, or allosteric
regulation. The thermodynamic factor $\eta^{\rm th}$ can be expressed
in terms of the driving force $\Theta \equiv -\Delta_{\rm r} G'/RT$ by
the general formula $\eta^{\rm th} = 1 - \e^{-\Theta}$, which also
applies to reactions with multiple substrates and products
\cite{nflb:13}. The  factors $\eta^{\rm kin}$  depends
 on the rate law and thus on the enzyme mechanism
considered (see SI \ref{sec:SIratelaws}). 
\myparagraph{Separable enzyme cost functions}
Enzyme demand can be quantified as a concentration (e.g., enzyme
molecules per volume) or mass concentration (where enzyme molecules
are weighted by their molecular weights).  If rate laws, fluxes, and
metabolite levels are known, the enzyme demand of a single reaction
$l$ follows from Eq.~(\ref{eq:factorised}) as
\begin{eqnarray}
\label{eq:factorisedU}
\enzyme_l(\cv,v_l) &=& v_l \cdot \frac{1}{\kcatplusl} \cdot
\frac{1}{\eta^{\rm th}_l(\Theta(\cv))}
 \cdot \frac{1}{\eta_l^{\rm sat}(\cv)} \cdot \frac{1}{\eta_l^{\rm reg}(\cv)}.
\end{eqnarray}
To determine the enzyme demand of an entire pathway, we sum over all
reactions: $\enzyme_{\rm tot} = \sum_l \enzyme_l$.  Based on its
enzyme demands $\enzyme_l$, we can associate each metabolic flux with
an enzyme cost $\enzymemetcost = \sum_l \hel\,\enzyme_l$, describing
the effort of maintaining the enzymes. The burdens $\hel$ of different
enzymes represent, e.g., differences in molecular mass,
post-translation modifications, enzyme maintenance, overhead costs for
ribosomes, as well as effects of misfolding and non-specific
catalysis. The enzyme burdens $\hel$ can be chosen heuristically, for
instance, depending on enzyme sizes, amino acid composition, and
lifetimes (see SI \ref{sec:SIcostFactors}). Setting $\hel=m_l$
(protein mass in Daltons), $\enzymemetcost$ will be in gram protein
per gram cell dry weight. Considering the specific amino acid
composition of enzymes, we can also assign specific costs to the
different amino acids. Alternatively, an empirical cost per protein
molecule can be established by the level of growth impairment that an
artificial induction of protein would cause \cite{deal:05, szad:10}.
Thus, each reaction flux $v_l$ is associated with an enzyme cost
$\enzymemetcostl$, which can be written as a function
$\enzymemetcostl(v_l,\cv) \equiv \hel\,\enzyme_l(\cv, v_l)$ of flux
and metabolite concentrations.  From now on, we refer to log-scale
metabolite concentrations $s_i = \ln c_i$ to obtain simple optimality
problems below.  From the separable rate law
Eq.~(\ref{eq:factorisedU}), we obtain the enzyme cost function
\begin{eqnarray}
\label{eq:TotalEnzymeDemand}
\enzymemetcost(\sv, \vv) &\equiv&  \sum_{l} \hel\,\el(v_{l},\sv) 
= \sum_{l} 
  \hel  \cdot v_{l} \cdot \frac{1}{\kcatplusl} \cdot
  \frac{1}{\eta^{\rm th}_l(\sv)}
  \cdot \frac{1}{\eta_l^{\rm sat}(\sv)}
  \cdot \frac{1}{\eta^{\rm reg}(\sv)}
\end{eqnarray}
for a given pathway flux $\vv$. If the fluxes are fixed and given, our
enzyme cost becomes, at least formally, a function of the metabolite
levels. We call it \emph{enzyme-based metabolic cost} (EMC) to
emphasize this fact.  The cost function is defined on the metabolite
polytope $\metabolitepolytope$, a convex polytope in log-concentration
space containing the feasible metabolite profiles.  Like the triangle
in Figure \ref{fig:fourchain}, the polytope is defined by
physiological and thermodynamic constraints.  It can be bounded by two
types of faces: On ``E-faces'', one reaction is in equilibrium, and
enzyme cost goes to infinity; ``P-faces'' stem from physiological
metabolite bounds. The shape of the cost
function depends on rate laws, rate constants, and enzyme burdens, and
its minimum points can be inside the polytope or on a P-face (see
Figure \ref{fig:fourchain} (f)).

\begin{figure}[ht!]
\hspace{-5mm}
\fcolorbox{black}{lightyellow}{\begin{minipage}{17cm}

\ \\

\hspace{.5cm}\parbox{16cm}{\small
\textbf{Box 1: Separable rate laws and enzyme cost function}
\vspace{1.5mm} 

According to Eq.~(\ref{eq:factorised}), reversible rate laws can be 
factorized into
five terms that depend on metabolite levels in different ways
\cite{nflb:13}.  For a reaction S $\leftrightharpoons$ P with
reversible Michaelis-Menten kinetics Eq.~(\ref{eq:mmratelaw}), a
driving force $\theta = -\Delta_{\rm r} G'/RT$, and a prefactor for
non-competitive allosteric inhibition, the rate law can be written as
\[v = {\color{c1}\enzyme}
\cdot {\color{c2}\kcatplus}
\cdot \underbrace{{\color{c3}[1-\mbox{e}^{-\theta}]}}_{\eta^{\rm th}}
\cdot \underbrace{{\color{c4}\frac{\concS/\kmS}{1+\concS/\kmS+\concP/\kmP}}
\cdot {\color{c4}\frac{1}{1+x/K_{\rm I}}}}_{\eta^{\rm kin}}\]
\vspace{-5mm}{\small
\begin{center}
\begin{tabular}[ht!]{ccccccccc}
Rate & = & {\color{c1} enzyme} &$\cdot$& {\color{c2}forward catalytic} &$\cdot$& {\color{c3}thermodynamic} &$\cdot$& {\color{c4}kinetic } \\
&& {\color{c1} level} && {\color{c2}constant} && {\color{c3}factor} & & {\color{c4}factor}
\end{tabular}
\end{center}
}
with inhibitor concentration $x$. In the example, with non-competitive
allosteric inhibition, the kinetic factor $\eta^{\rm kin}$ could even
be split into a product $\eta^{\rm sat} \cdot \eta^{\rm reg}$.  The
first two terms in our example, $\enzyme \cdot \kcatplus$, represent
the maximal velocity (the rate at full substrate-saturation, no
backward flux, full allosteric activation), while the following
factors decrease this velocity for different reasons: the factor
$\eta^{\rm th}$ describes a decrease due to backward fluxes (see SI
Figure \ref{fig:fluxfractions}) and the factor $\eta^{\rm kin}$
describes a further decrease due to incomplete substrate saturation
and allosteric regulation (see Figure
\ref{fig:efficiencies} b).  While $\kcatplus$ is an enzyme-specific
constant (yet, dependent on conditions such as pH, ionic strength, or
molecular crowding in cells; unit 1/s), the
efficiency factors are concentration-dependent, unitless, and can vary
between 0 and 1. The thermodynamic factor $\eta^{\rm th}$ depends on
the driving force (and thus, indirectly, on metabolite levels), and
the equilibrium constant is required for its calculation. The
saturation factor $\eta^{\rm sat}$ depends directly on metabolite
levels and contains the $\km$ values as parameters. Allosteric
regulation yields additive or multiplicative terms in the rate law
denominator, which in our example and can be captured by a separate
factor $\eta^{\rm reg}$.  The enzyme cost for a flux $v$, with a
enzyme burden $\he$, can be written as

\[
\enzymemetcost =  
          {\color{c5}\he} \cdot {\color{c1}\enzyme} = 
          {\color{c5}\he} \cdot v
          \cdot {\color{c2}\frac{1}{\kcatplus}}
          \cdot \underbrace{{\color{c3}\frac{1}{[1-\mbox{e}^{-\theta}]}}}_{1/\eta^{\rm th}}
          \cdot \underbrace{{\color{c4}\frac{1+\concS/\kmS+\concP/\kmP}{\concS/\kmS}}     \cdot {\color{c4} [1+x/K_{\rm I}]}}_{1/\eta^{\rm kin}}
          \]

and  contains the terms from the rate law in inverse
form. The first factors, $\he \,v/\kcatplus$, define a minimum
 enzyme cost, which is then increased by the following efficiency factors.
Again, $1/\eta^{\rm kin}$ can be split into $1/\eta^{\rm sat} \cdot 1/\eta^{\rm reg}$.
  By omitting some of these factors, one can construct simplified enzyme
cost functions with higher specific rates, or lower enzyme demands
(compare Figure \ref{fig:efficiencies}b). For a closer approximation,
the factors may be substituted with constant numbers between $0$ and
$1$.
The conversion between fluxes and enzyme levels, in both directions,  is shown below.

\begin{center} 
\includegraphics[width=10cm]{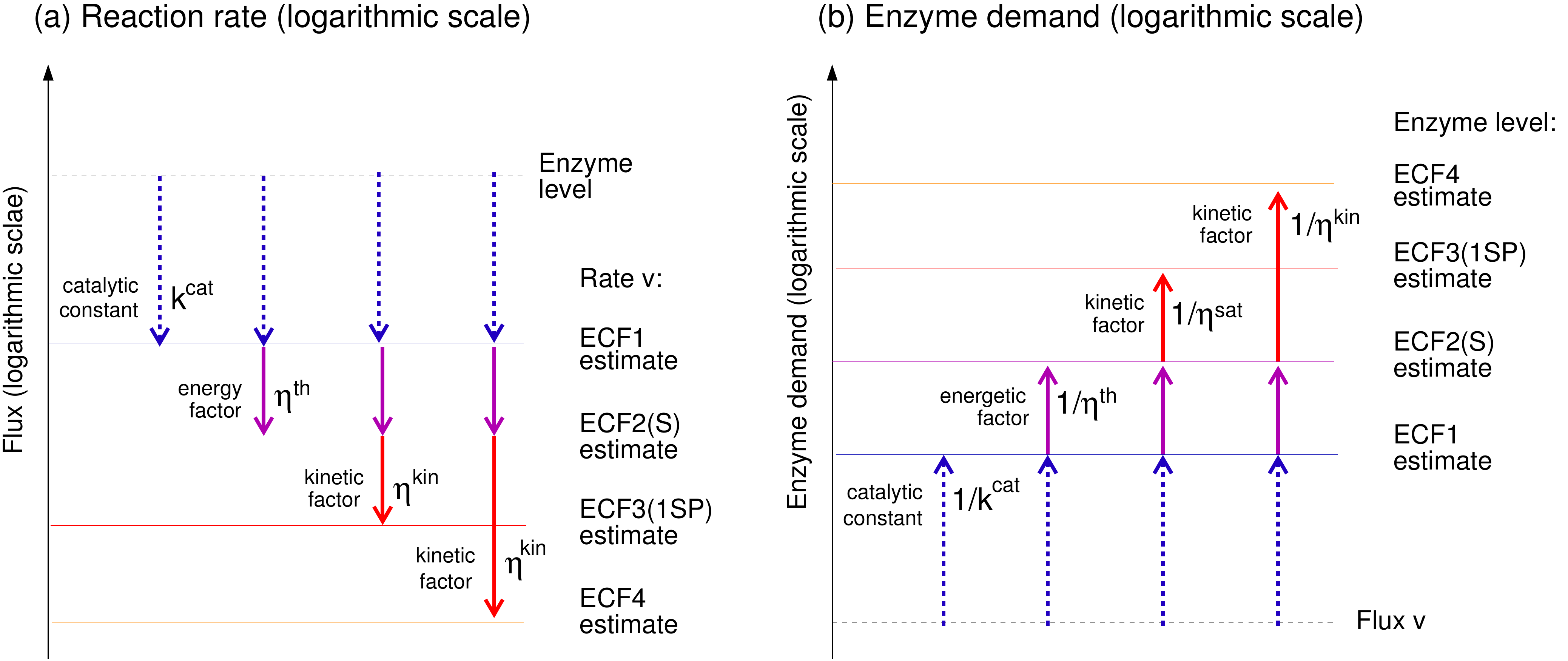} 
\end{center}
On a logarithmic scale, rates and enzyme
cost can be split into sums of efficiency terms.
(a) Starting from the logarithmic
enzyme level (dashed line on top), we add the terms $\log \kcatplus$,
$\log \eta^{\rm th}$, and $\log \eta^{\rm kin}$, and obtain better and
better approximation of the rate. In the example shown, $\kcatplus$
has a numerical value smaller than 1. The more precise approximations
(with more terms) yield smaller rates. The ECF4 arrows refer to other
possible rate laws with additional terms in the denominator.  (b)
Enzyme demand is shaped by the same factors (see
Eq.~(\ref{eq:factorisedU})).  Starting from a desired flux (bottom
line), the predicted demand increases as more terms are considered. \ \\
}
\end{minipage}
}
  \label{fig:pscscores}
\end{figure}

\subsection{Enzyme cost minimization}

\myparagraph{\ \\Enzyme cost minimization} The cost
function $\enzymemetcost(\sv,\vv)$ reflects a trade-off between fluxes
to be realized and enzyme expression to be minimized, where the
relation between fluxes and enzyme levels is not fixed, but depends on
metabolite log-concentrations $\sv$.  Wherever trade-offs exists in
biology, it is common to assume that evolution converges to
Pareto-optimal solutions \cite{szzh:12}, e.g., metabolic states for
which there are no other solutions that have the same flux but with a
lower cost $\enzymemetcost$, or the same cost $\enzymemetcost$ but
with a higher flux. Therefore, for a given measured flux $\vv$, we may
expect to find profiles of metabolite and enzyme concentrations that
minimize $\enzymemetcost$.  We
can now use this principle to predict metabolite and enzyme
concentrations in cells. As with our simple model in
Figure \ref{fig:efficiencies}, minimizing the enzyme cost on the
metabolite polytope yields an optimal metabolite profile, from which
the optimal enzyme profile can be computed using
Eq.~(\ref{eq:factorisedU}).

The resulting method, which we call enzyme cost minimization (ECM), is a convex optimization problem and can be
solved with local optimizers. Enzyme demand and enzyme cost functions,
for single reactions or pathways, are differentiable, convex functions
on the metabolite polytope. This convexity holds for a variety of rate
laws, including rate laws describing polymerization reactions
\cite{hogr:13}, and even for the more complicated problem of
preemptive enzyme expression, i.e., a cost-optimal choice of enzyme
levels that allows the cell to deal with a number of future conditions
(see SI \ref{eq:preemptiveExpression}). If a model contains
non-enzymatic reactions, this changes the shape of the metabolite
polytope, but not the enzyme cost function, and the polytope remains
convex, e.g., if the non-enzymatic reactions are irreversible with
mass-action rate laws (see Methods). \arrenalt{a very important
  point.}  \myparagraph{Sub-optimal solutions and tolerance ranges.}
Obviously, metabolite and enzyme levels may be under various other
constraints that are not reflected in our pathway model. To assess how
easily the metabolic state can be adapted to external requirements, we
can study the cost of deviations from the optimal metabolite levels.
If the cost function $\enzymemetcost(\sv)$ has a broad optimum as in
Figure \ref{fig:fourchain}, cells may flexibly realize metabolite
profiles around the optimal point, and the choice of metabolite levels
may vary from cell to cell.  We can quantify the tolerable variations
by relaxing the optimality assumptions and a computing tolerance range
for each metabolite level.  \myparagraph{Workflow for model building,
  and example model} To apply ECM in practice, we developed
a workflow in which a kinetic model is constructed, a consistent set
of kinetic constants is determined by parameter balancing
\cite{lskl:10,slsk:13}, and optimal metabolite and enzyme levels are
predicted along with their tolerance ranges.  Different types of EMC
function and constraints (e.g., allowed ranges for metabolite levels)
can be chosen. Missing data (e.g., $\km$ values), can thus be handled
in two ways: either, by using a simplified EMC function that does not
require this parameter, or by relying on parameter values that are
chosen by the workflow based on parameter balancing.

\subsection{Which factors shape the optimal enzyme profile and how?}
\label{sec:rules}

\myparagraph{\ \\Simplified cost functions} What determines the demand
for specific enzymes?  If the metabolite levels are known,
we can easily see by analyzing the efficiency factors
in Eq.~(\ref{eq:TotalEnzymeDemand}) (see Box 1).  By omitting some
factors or replacing them by constant numbers $0<\eta\le 1$,
simplified enzyme cost functions with fewer parameters can be
obtained.  For example, $\eta^{\rm th}=1$ would imply an infinite
driving force $\Theta \rightarrow \infty$ and a vanishing backward
flux, $\eta^{\rm kin}=1$ implies full substrate saturation, as well as
full allosteric activation and no allosteric inhibition (or no
allosteric regulation at all).  In these limiting cases, enzyme
activity will not be reduced, and enzyme demand will be given by the
capacity-based estimate $v/\kcatplus$, a lower estimate of the actual
demand.  Instead of omitting an efficiency factor, it can also be set
to a constant value between 0 and 1.  Such simplifications and the
resulting enzyme cost functions with fewer parameters can be practical
if rate constants are unknown.

Depending on the type of data available (e.g., $\kcat$ values,
equilibrium constants, or even $\km$ values), one may choose between
different types of cost functions with different data requirements:
EMC0 (``sum-of-fluxes-based'' same prefactors for all enzymes), EMC1
(``capacity-based'', setting all $\eta = 1$ and thus replacing
reaction rates by the maximal velocities), EMC2 (``energy-based'';
considering driving forces, and setting $\eta^{\rm kin}=1$), EMC3 (``saturation-based'', assuming simple rate laws
depending on products of substrate or product concentrations, and
including the driving forces), and EMC4 functions
(``kinetics-based''; with dependence on individual metabolite levels).
Details of the simplified EMC functions are given in Table
\ref{tab:EnzymeCostScores} and SI \ref{sec:SIListOfScores}.  Each EMC
function is a lower bound on the following functions; i.e., even if
only a simplified cost function can be used, it will always yield a
lower bound on the actual enzyme cost. 

\begin{table}[t!]
\begin{center}
{\small 
\begin{tabular}{|>{\columncolor{lightyellow}}l|>{\columncolor{lightyellow}}l|>{\columncolor{lightyellow}}l|>{\columncolor{lightyellow}}l|>{\columncolor{lightyellow}}l|>{\columncolor{lightyellow}}l|}
  \hline
  \rowcolor{brown} EMC function & 
$\eta^{\rm th}(\Theta(\cv))$ &     
$\eta^{\rm kin}(\cv)$ &     
Parameters  &     
Denominators &
Depends on  \\ \hline & & & & & \\[-2mm]
EMC0 (``Sum of fluxes'')  & - & - & - & &  \\
EMC1 (``Capacity-based'') & - & - & $\he$, $\kcatplus$  & &  \\
EMC2 (``Energy-based'')   & $\checkmark$ & - & $\he$, $\kcatplus$, $\keq$ & $D^{\rm S}$, $D^{\rm SP}$ & Driving force\\
EMC3 (``Saturation-based'') & $\checkmark$ & $\checkmark$ & $\he$, $\kcatplus$, $\keq$, $\km$ & $D^{\rm 1S}$, $D^{\rm 1SP}$ & Metabolite levels  \\
EMC4 (``Kinetics-based'')   & $\checkmark$ & $\checkmark$ & $\he$, $\kcatplus$, $\keq$, $\km$ & general & Metabolite levels  \\
  \hline
\end{tabular}
}
\end{center}
\caption{Simplified enzyme cost functions. By omitting some terms in
  Eq.~(\ref{eq:factorisedU}), we obtain a number of cost functions
  with simple dependencies on enzyme parameters and metabolite levels.
  Terms marked by $\checkmark$ appear explicitly in the rate and cost
  formulae, while other terms are omitted or set to constant values.
  The EMC0 function yields the sum of fluxes, EMC1 functions contain
  enzyme-specific flux burdens based on $\kcat$ and $h$ values (i.e.,
  replacing reaction rates by their maximal velocities).  EMC2 depends
  on metabolite levels only via the driving forces.  EMC3 functions
  are based on simplified rate laws, and EMC4 functions capture all
  rate laws, possibly including allosteric regulation.  The rate law
  denominators $D^{\rm S}, D^{\rm SP}, D^{\rm 1S}$, and $D^{\rm 1SP}$
  are described in SI \ref{sec:SIratelaws}, the EMC functions
  themselves in SI \ref{sec:SIfscScores}.}
\label{tab:EnzymeCostScores}
\end{table}

\myparagraph{Simple estimates of enzyme costs (EMC0 and EMC1
  functions)} Let us consider the different simplifications in more
detail.  As long as fluxes are the only data available, we may assign
identical catalytic constants and enzyme burdens to all enzymes and
assume that all reactions run at their maximal velocities.  Then,
enzyme levels and fluxes will be proportional across the network and
the cost function Eq.~(\ref{eq:TotalEnzymeDemand}) will be of type
EMC0 and proportional to the sum of fluxes. However, catalytic constants
span many orders of magnitude \cite{bnsl:11} and enzyme molecular masses
are quite variable as well, suggesting that EMC0 is a strong oversimplification.
In contrast, if individual $\kcatplus$ and $\hel$ values are known, we obtain an EMC1
cost function, which is still independent of metabolite levels.  In
the flux cost function $\sum_l \rrlmin\,v_l$, each enzyme
has an individual flux burden $\rrlmin = \hel/{\kcatplus}_l$, and the
same ratios have been used as cost weights in FBA with flux
minimization \cite{holz:04} or molecular crowding \cite{bvem:07}.  If
$\kcat$ values are unknown, they may be replaced by ``typical'' values
(see \cite{bnsl:11}).  The enzyme burdens $\he$ can subsume factors
like protein size, protein lifetime, covalent modifications, or space
restrictions (see \cite{horh:11} and SI \ref{sec:SIcostFactors}); if
these are unknown, one may assume that all enzymes are equally costly,
setting their burdens to $\hel=1$.  While the specific costs $\he$ are
relatively uniform, the $\kcat$ values vary within five orders of
magnitude \cite{bnsl:11}, and are thus the major determinant of
$\rrlmin$.  

\myparagraph{Small driving forces lead to larger enzyme costs (EMC2)}

However, by setting $\eta^{\rm th} = \eta^{\rm kin} = 1$, we may
obtain unrealistic results.  First, the simplifying assumption
$\eta^{\rm th}=\eta^{\rm kin}=1$ implies uncontrollable metabolic
states. In a kinetic model with completely irreversible and
substrate-saturated enzymes, the reaction rates would be
\emph{independent} of metabolite levels and the steady state would
depend on finely tuned enzyme levels. Any random variation of
the enzyme levels would lead to non-steady states, with fast
accumulation or depletion of intermediate metabolites. Such states are
extremely fragile and thus uncontrollable. \arrenalt{very good!}  When
assuming efficiencies $\eta^{\rm th}$ or $\eta^{\rm kin}$ smaller than
1, we accept an increased cost and thereby acknowledge that
controllability must be paid by enzyme investments.  Second, EMC1
functions underestimate all enzyme costs, and in reactions close to
chemical equilibrium the errors may become large. At a reaction Gibbs
energy of $\Delta_{\rm r} G' = -0.1\, RT$, the efficiency of the
catalyzing enzyme decreases by a factor of $\eta^{\rm th} = 1-\e^{0.1}
\approx 0.1$, and the demand for enzyme increases by a factor of
$1/\eta^{\rm th} \approx 10$. To account for this effect, we can use
EMC2 functions, considering the thermodynamic factor $\eta^{\rm th}_l
= 1-\e^{-\Theta_l(\sv)}$. The driving forces are expressed in terms of
metabolite log-concentrations $\Theta_l(\sv)$ and equilibrium
constants, which need to be known.  This factor approaches infinity as
reactions reach equilibrium (i.e. where $\Theta_l \rightarrow 0$),
which is what keeps reactions away from equilibrium during cost
minimization (see, for example, Figure \ref{fig:fourchain}).

Compared to the following EMC3 and EMC4 functions, the advantage of
EMC2 functions is that they are based on equilibrium constants only,
i.e., on a physical property of the reacting compounds that does not
depend on the enzyme at all.  Several in silico methods exist to
estimate $\keq$ for virtually any biochemical reaction \cite{jhbh:08,
  nhmf:13} and the values can be easily obtained at
\url{http://equilibrator.weizmann.ac.il/} \cite{fnbm:12}.  Methods
like MDF \cite{nbfr:14} and mTOW \cite{tnah:13} have been developed to
address exactly this situation, where detailed kinetic information is
hard to obtain.  We discuss the relation between EMC2 and MDF in SI
\ref{sec:SItighterConstraints}. Of course, in all of these cases,
$\kcat$ values still have to be guesstimated or set to identical
values for all enzymes. Aside from the EMC2 function, there are other
energy-based estimates of the enzyme cost.  For instance, the enzyme
demand in Figure \ref{fig:fourchain} (an EMC3 function with kinetic
constants, fluxes, and enzyme burdens set to 1) has the energy-based
cost $\rrpw = \sum_{l} [1-\e^{-\Theta(\cv)}]\inv$ as a lower estimate.
Since $1 - e^{-x} \leq x$ for all positive $x$, an even lower estimate
is $\sum_{l} \Theta(\cv)\inv$ (Figures \ref{fig:fourchainOther} and
\ref{fig:single_metabolite2} in SI).  Some variants of FBA relate
fluxes to metabolite profiles, which are then required to be
thermodynamically feasible, i.e., within the metabolite polytope.
ECM constrains the metabolite profiles even further: as shown in
Figure \ref{fig:fourchain}, profiles close to an E-face are very
costly and can never be optimal. This holds for EMC2 functions and for
the more realistic enzyme costs, which will even be higher. Thus,
regions close to E-faces can be excluded from the polytope. At
P-faces, defined by physiological bounds, there will be no such
increase, so the optimum may lie on a P-face (see Figure
\ref{fig:fourchain} (f)).  To do so, we simply define lower bounds for
all driving forces (see SI \ref{sec:SItighterConstraints}): these
bounds can be used both in ECM or in thermodynamic FBA to reduce the
search space for metabolite profiles.

\myparagraph{Incomplete substrate and product saturation increase the
  enzyme cost (EMC3 and EMC4)} The next logical step is to relax the
assumption that $\eta^{\rm kin} = 1$.  Just like the thermodynamic
factor $\eta^{\rm th}$, the kinetic factors $\eta^{\rm sat}$ and
$\eta^{\rm reg}$ can be used to define tighter constraints on
metabolite levels.  However, unlike $\eta^{\rm th}$, the kinetic terms
can take various forms and contain many kinetic
parameters. To obtain
simple, but reasonable formulae, we first consider rate laws in which
enzyme molecules exist only in three possible states: unbound, bound
to all substrate molecules, or bound to all product molecules.
Metabolites affect the rate only through the mass-action terms
$S=\prod_i (s_i/{\km}_i)$ (for substrates) and $P=\prod_j p_i/{\km}_j$
(for products), and the degree of saturation is determined by
$\eta^{\rm kin} = S / (1 + S + P)$, where the formula contains only
one Michaelis-Menten constant for all substrates and (optionally) one
for all products.  Since EMC3 requires both $\kcat$ and $\km$ values
for every enzyme, and $\kcat$ values are more likely to be known than
$\km$ values, there is no real reason to consider cases where the
$\kcat$ value is not known.  EMC3 represents a good balance between
complexity and requirement for kinetic parameters, and is a practical
cost function if simple, realistic rate laws are desired. The EMC4
functions, finally, represent general rate laws and $\eta^{\rm kin}$
can take many different forms depending on mechanism and order of
enzyme-substrate binding. Again, for simplicity, we resort to
analyzing only a small set of relatively general templates for EMC4,
known as convenience kinetics \cite{likl:06a} or modular rate laws
\cite{liuk:10}. Nevertheless, our formalism allows a much wider
range of rate laws, and we consider EMC4 a wild-card cost function
that covers almost any well-behaved metabolic rate law (see SI
\ref{sec:SIfscScores} for more details).

\subsection{Enzyme and metabolite levels in \emph{E.~coli} central metabolism} 
\label{sec:ecoli}

\begin{figure}[t!]
  \begin{center}
    \includegraphics[width=0.37\textwidth]{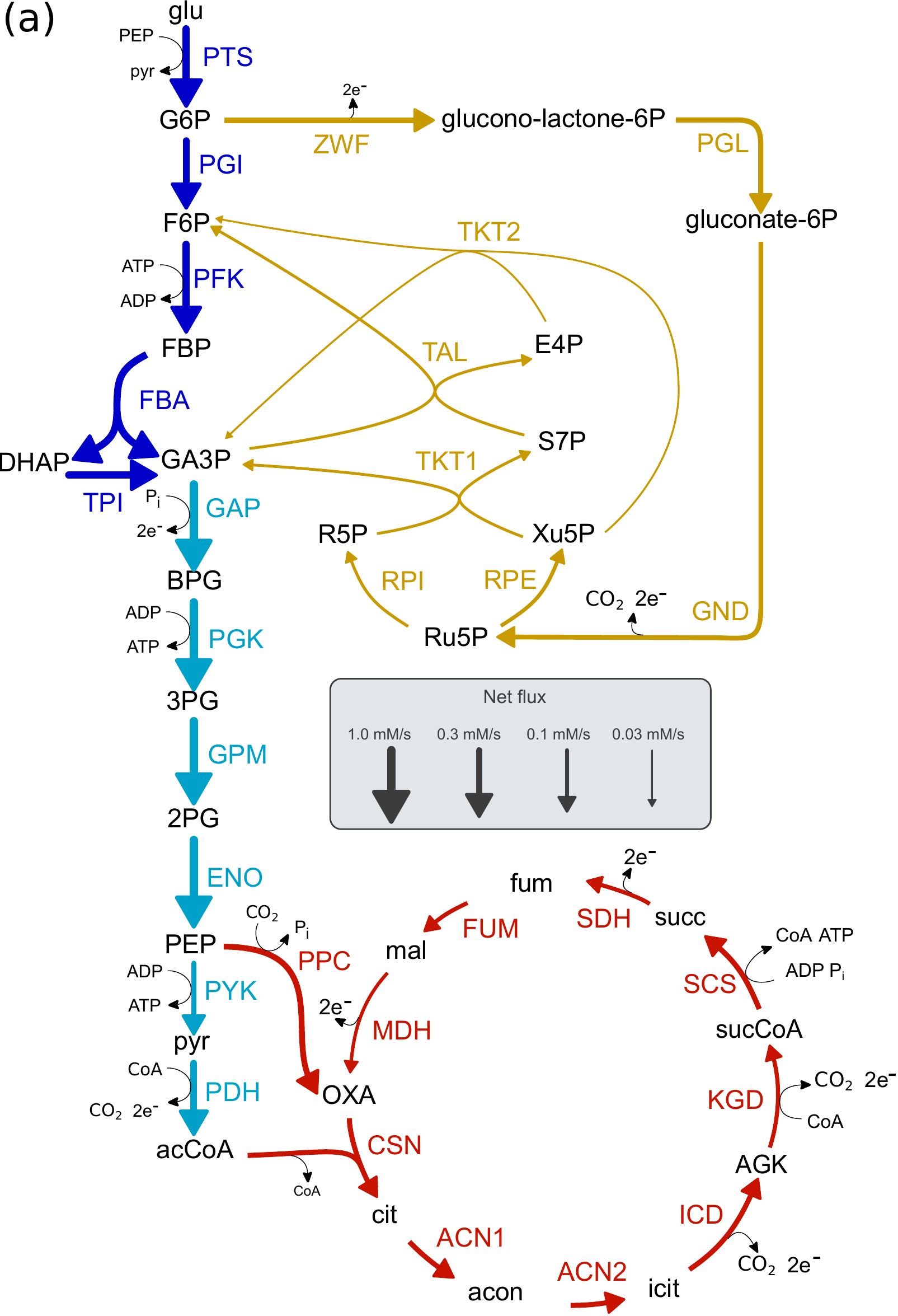} 
    \hspace{3mm}
    \includegraphics[width=0.55\textwidth]{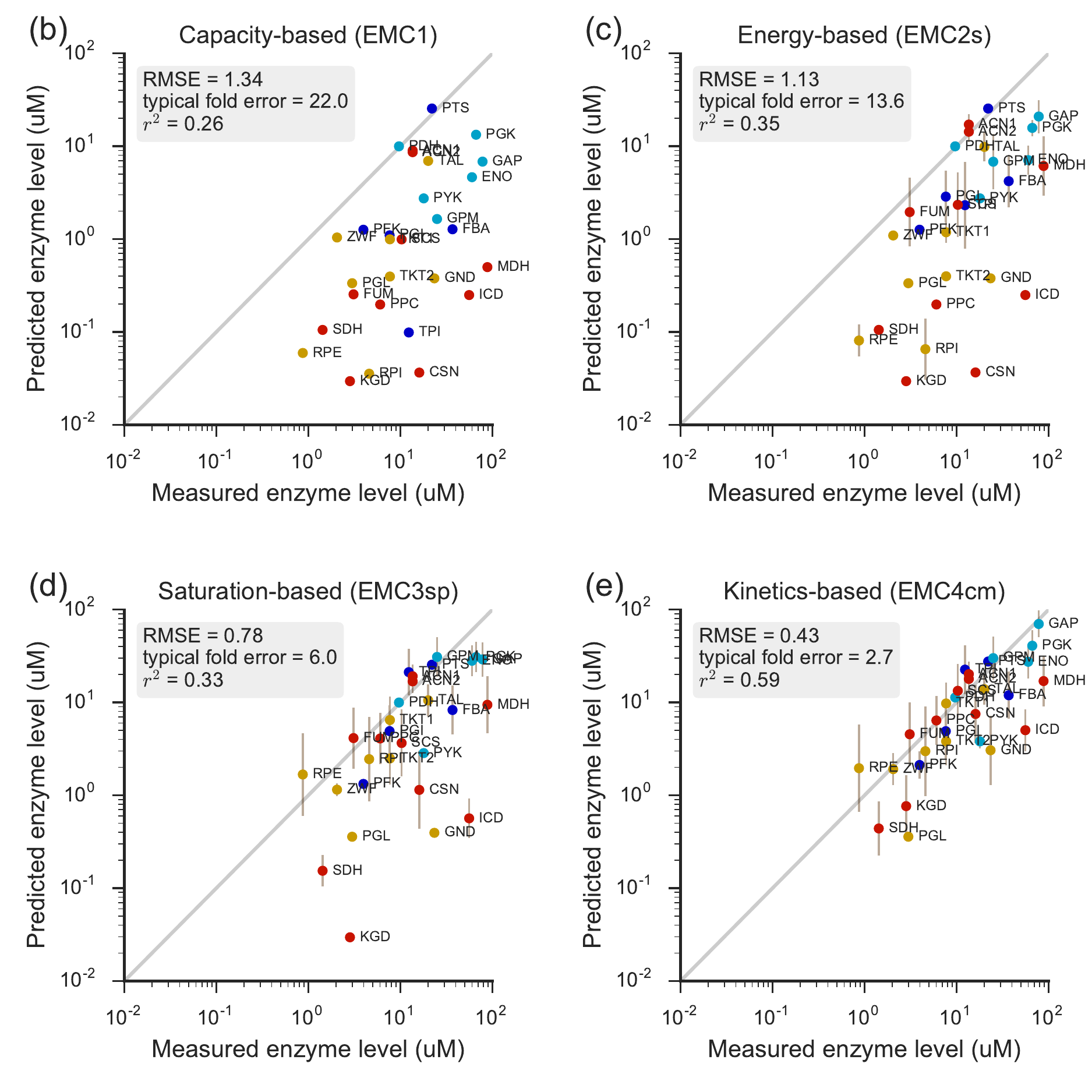}
    \caption{Predicted enzyme levels in
      \emph{E.~coli} central metabolism.  (a) Network model with
      pathways marked by colors.  Flux magnitudes are represented by
      the arrows' thickness.  (b) The ratio flux/$\kcatplus$ (EMC1)
      as a predictor for enzyme levels. Points on the dashed line
      would represent precise predictions.  (c) Enzyme levels
      predicted by the energy-based EMC2(S) function.  Vertical bars
      indicate tolerance ranges obtained from a relaxed optimality
      condition (allowing for a one percent increase in total enzyme
      cost).  (d) Enzyme levels predicted with EMC3 function
      representing fast substrate or product binding. (e) Enzyme
      levels predicted with EMC4 function based on the common modular
      rate law \cite{liuk:10}. In all sub-figures (b-e), RMSE is the 
      root mean squared error (in log$_{10}$-scale) of our predictions compared to the
      measured enzyme levels, and $r$ stands for the Pearson correlation
      coefficient.
      Predictions are based on fluxes from \cite{hann:11}, $\kcatplus$
      and $\km$ values from BRENDA \cite{sceg:04}, and compared to
      protein data from \cite{skva:15}.  For metabolite predictions,
      see SI Figure \ref{fig:EcoliPredictionsSI}.}
    \label{fig:EcoliPredictions}
  \end{center}
\end{figure}

\myparagraph{\ \\Validation using data from \emph{E. coli} central
  metabolism} To benchmark our prediction of metabolite and enzyme
levels and to see whether more complex EMC functions improve the
predictions, we applied ECM to a model of \emph{E. coli} central
metabolism, containing three major pathways: glycolysis, the pentose
phosphate pathway, and the TCA cycle (see Figure
\ref{fig:EcoliPredictions} (a), and Methods for modeling details).
Figure \ref{fig:EcoliPredictions} (b-d) compares predicted enzyme
profiles to measured protein levels \cite{skva:15}.  The absolute
values of predicted enzyme levels arise directly from the model, using
the fluxes reported in \cite{hann:11} (e.g., glucose uptake rate
$8.13$ mmol/gCDW/h), while cellular protein concentrations were
obtained from proteomics data (measured in similar conditions
\cite{skva:15}) and assuming an average cell volume of $\sim\!1$ fL
($10^{-15}$ liters) \cite{milo:13}. EMC4 predicts values that are in
the right order of magnitude and reflect differences in enzyme levels
along the pathways. The prediction error of $0.43$ for enzyme levels
(RMSE: root mean square error on a log$_{10}$ scale) corresponds to a
typical fold error of $10^{\rm RMSE} = 2.7$.  In line with the
measured protein levels, the predicted enzyme levels tend to be larger
in glycolysis than in TCA and pentose phosphate pathway, reflecting
the larger fluxes.  Predicted metabolite concentrations (RMSE $0.58$,
corresponding to a typical fold error of $3.8$), thermodynamic forces
and $c/\km$ ratios are shown in a supplementary file.

We note that the predicted enzyme levels become more accurate when
stepping up to more complex cost functions, with a
prediction error decreasing monotonously from $1.34$ to $0.43$.  The capacity-based
enzyme cost (EMC1) assumes that enzymes operate at full capacity ($v =
\enzyme\, \kcatplus$) and therefore underestimates all enzyme levels (Figure
\ref{fig:EcoliPredictions} (c)). In reality, many reactions in
central metabolism are reversible and many substrates do not reach
saturating concentrations. When taking these effects into account, the predictions
come closer to measured enzyme levels (Figure
\ref{fig:EcoliPredictions} (c-d)). For instance, FUM (fumarase, fumA)
and MDH (malate dehydrogenase) have a much higher predicted level in
EMC2-4 than in EMC1 since the thermodynamics-based costs account for
their low driving force. Similarly, the predicted levels of two
pentose-phosphate enzymes (Ribulose-5-phosphate epimerase RPE and
ribose phosphate isomerase RPI) are much higher in EMC3 and EMC4
because of their low affinity to the substrate ribulose-5-phosphate
(Ru5P).  In some cases, however, the more complex EMC4 fails to
improve the prediction over the simpler methods and can actually
make them worse. For instance, the 6-phosphogluconolactonase (PGL) 
and pyruvate kinase (PYK) reactions are underestimated in all cases and
do not improve significantly in EMC4.  Glucose 6-phosphote dehydrogenase
(ZWF) is predicted quite well by EMC2-3, but its level is overestimated
in EMC4.  Overall, the EMC4 function 
performs substantially better on
average than the simpler cost functions even though it relies on a
much larger set of parameters, many of which are known with low
certainty.  To test the sensitivity of our results to the choice of
proteomic data, we repeated the entire analysis using measured enzyme
concentrations from \cite{avpn:12} and reached essentially the same
findings.

\begin{figure}[t!]
  \begin{center}
{\small
    \begin{tabular}{cc}
      (a) Enzyme demand (energy-based EMC2s function) &
      (b) Enzyme demand (kinetics-based EMC4cm function)\\
      \includegraphics[width=7.5cm]{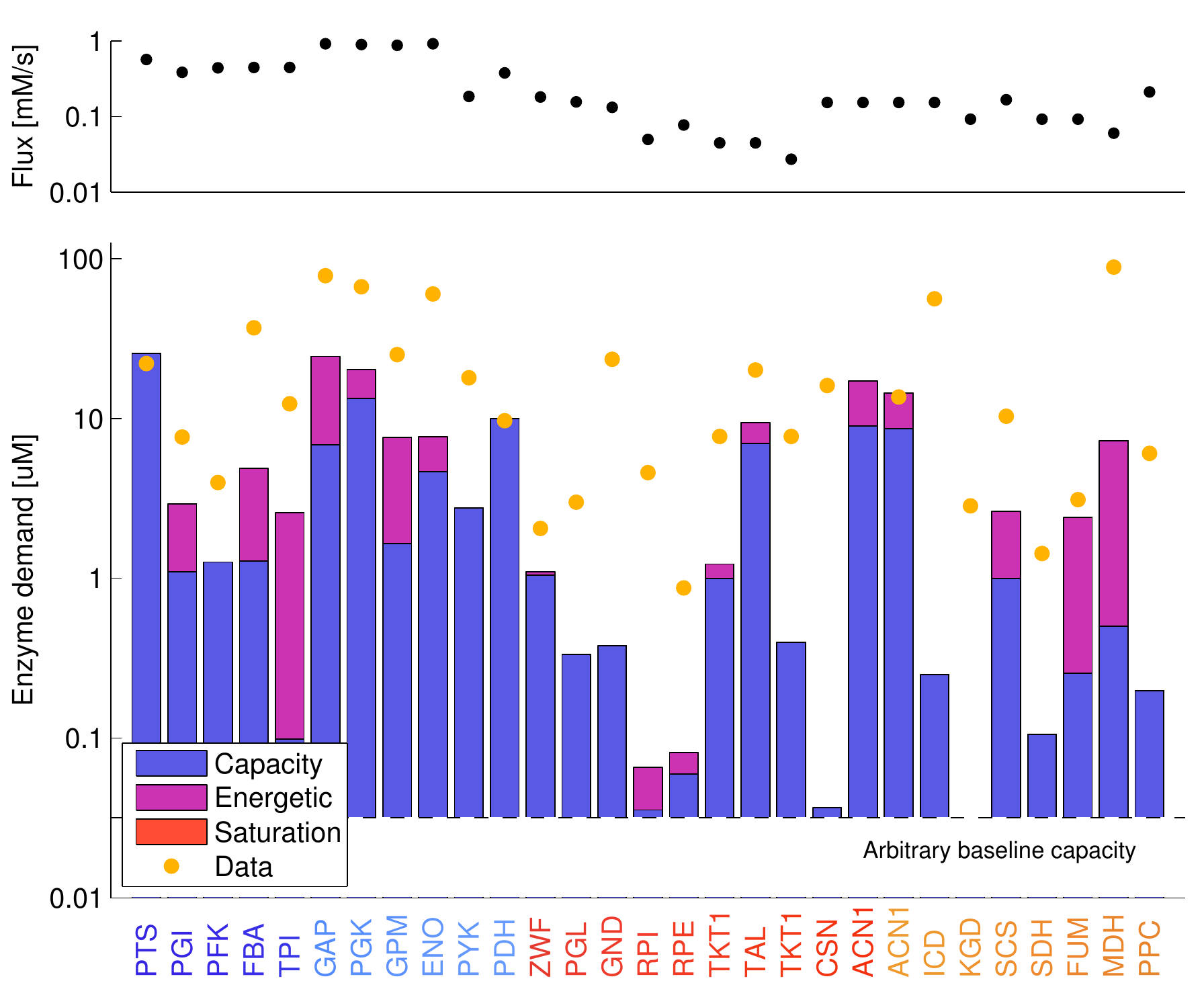}&
      \includegraphics[width=7.5cm]{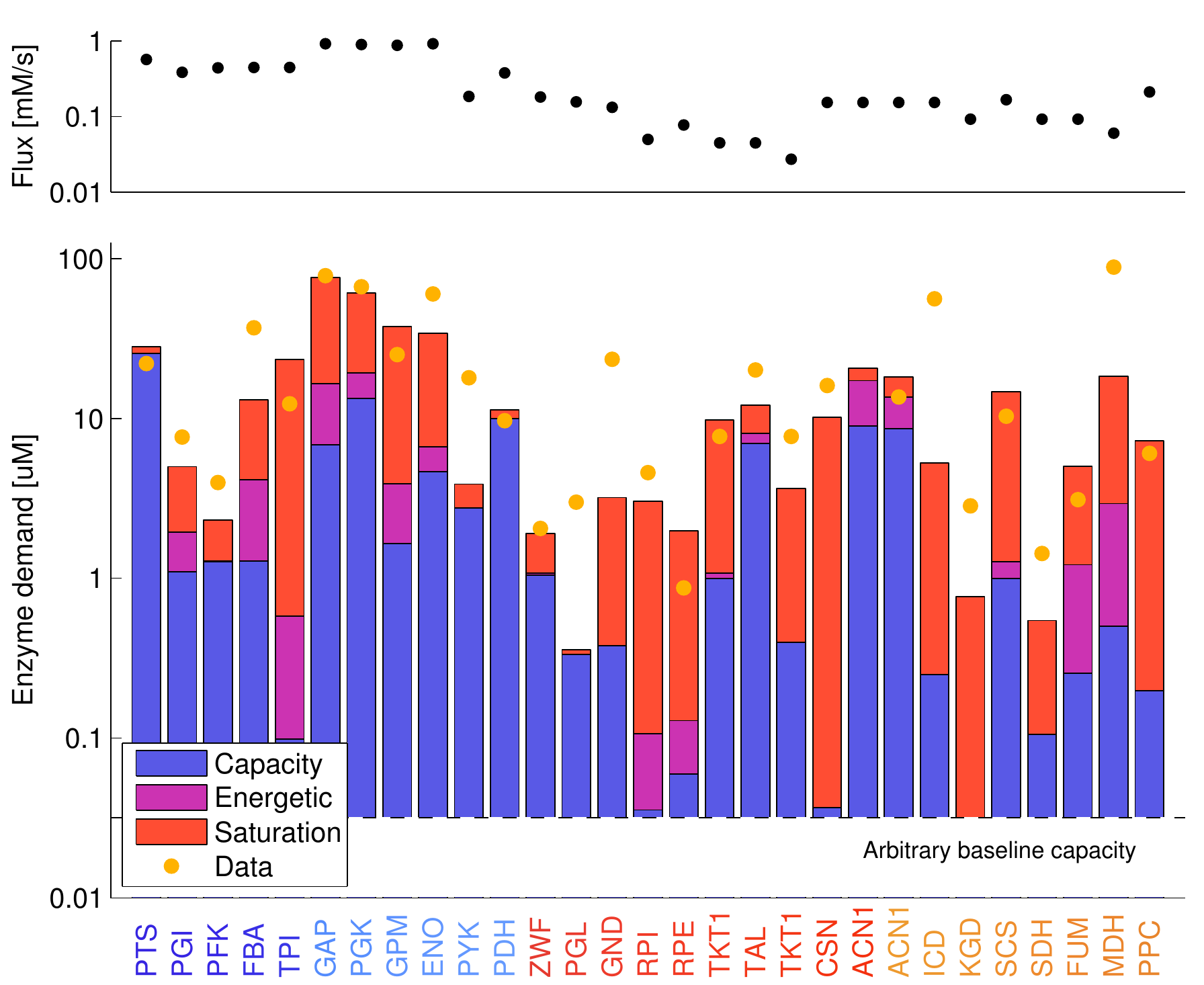}\\
    \end{tabular}
}
    \caption{Enzyme demand in central
      metabolism.  (a) Measured fluxes for all reactions (black dots
      on top) lead to an enzyme demand (bottom). The enzyme demand,
      predicted by using the energy-based EMC2s cost function, can be
      split into factors representing enzyme capacity and
      thermodynamics (see Box 1).  Bars show predicted enzyme levels
      in mM for individual enzymes on logarithmic scale (compare
      Figure \ref{fig:efficiencies} b). Yellow dots denote measured
      enzyme levels (in mM).  Note that the bars do not represent
      additive costs, but multiplicative cost terms on logarithmic
      scale; therefore, the relevant feature of the blue bars is not
      their absolute lengths, but their differences between enzymes.
      (b) The kinetics-based EMC4cm cost function includes saturation
      terms and yields more accurate predictions. Starting from the
      capacity cost (in blue), the thermodynamic (purple) and
      saturation (red) terms increase the enzyme demands and make them
      less variable between enzymes (on log-scale).  Note that flux
      data (circle) and protein data (yellow dots) are identical in
      both plots.}
    \label{fig:EcoliPredictionsTerms}
  \end{center}
\end{figure}

\myparagraph{Validation of predicted metabolite levels} Although ECM
puts enzymes on a pedestal due to their relatively high cost, the
metabolite concentrations are key to minimizing that cost. One would
thus expect to find a good correspondence between the predicted
metabolite profile and concentrations measured \emph{in vivo},
especially when the predictions of the enzyme levels are good.  Since
some of the EMC functions leave metabolite levels underdetermined, we
penalized very high or low metabolite concentrations by adding a
second, concentration-dependent objective to the optimization
problem. In particular for EMC0 and EMC1, this regularization term is
the only term -- aside from global constraints -- that determines the
metabolite concentrations since metabolites have no effect on enzyme
cost whatsoever. In all other cases, the term mostly influences
metabolites that have a minimal effect on the cost.  Comparing the EMC
metabolite prediction with in vivo experimental data, as shown in SI
Figure \ref{fig:EcoliPredictionsSI}, the predicted metabolite levels
are in the correct scale.  Similarly to enzyme level predictions, the
error decreases from EMC1 to EMC4cm, where we find a prediction error
of about $0.58$ (corresponding to a typical fold error of $3.8$),
slightly higher than the prediction error for enzymes (see Figure
\ref{fig:EcoliPredictions} (e)).

Can we now explain the cellular enzyme levels?  Figure
\ref{fig:EcoliPredictionsTerms}, just like the scheme in Figure
\ref{fig:efficiencies} (b), shows the specific contributions to enzyme
demand in each reaction.  A simple (EMC2) prediction based on driving
forces predicts the thermodynamic cost terms already quite reliably
and improves the enzyme prediction. However, accounting for incomplete
substrate saturation, described by the saturation cost term, has an
even larger effect on cost in most enzymes. For practical cost
estimates, for instance when computing flux burdens for FBA, we can
conclude that multiplying the experimentally determined $\kcat$ values
by energetic factors tends to improve the results. However, total
enzyme demand will still be underestimated.

\section{Discussion}

\myparagraph{\ \\Summary} When applying mathematical models to learn
about biology, one typically faces a conflict between desired model
accuracy and the amount of available data. Metabolic systems are known
to abide to several physical and physiological considerations, all of
which are mathematically well-described (e.g.~flux balance,
thermodynamics, kinetics, and cost-benefit optimality). Taking all of
these aspects into account would create very detailed models but at
the price of considerably increasing the demand for data.  Here, we
obtain a flexible modelling method by combining the two main modelling
approaches, constraint-based and kinetic modelling, in a new way: with
fixed metabolic fluxes, kinetic models are used to determine a
cost-optimal state.  The tiered approach in ECM allows for different
levels of detail, which can easily be matched to the amount of
existing data. The minimal requirement for running ECM is to have a
metabolic network with given steady-state fluxes, while the maximal
requirement would be a fully parameterized kinetic model.  Although
similar approaches exist in dynamic modeling \cite{sgsb:06,slsk:13}
and enzyme optimization \cite{tnah:13,fnbl:13,nbfr:14}, ECM extends
these ideas to the most general kinetic rate laws and cost functions,
while proving that the emerging optimization problem is convex and
thus easily (albeit numerically) solvable. We discuss the
advances made by ECM in detail by listing these five points:

 \textbf{1. Solving the enzyme optimality problem in metabolite space}
 One way of modelling the cost and benefit of enzymes is to study
 kinetic models and to treat enzyme levels as free variables to be
 optimized. However, this calculation can be hard because enzyme
 profiles may lead to one, several, or no steady states, and the
 resulting optimality problem can be non-convex. By using fluxes, and
 then metabolite concentrations, as our primary variables, we
 drastically simplify this task.  In thermodynamic FBA, known flux
 directions are used to determine a set of feasible metabolite
 profiles, the metabolite polytope.  Here, the same set is used as a
 space for screening, sampling, and optimization of metabolic states;
 accordingly bounds on metabolite concentrations or driving forces can
 be easily formulated as linear constraints.  Using log-concentrations
 as free variables, we can enumerate all possible metabolite profiles,
 solve for the enzyme profiles, and obtain a systematic
 parametrization of all (steady and non-steady) states (see SI Figure
 \ref{sec:SIworkflow}) -- which renders a screening of enzyme space
 obsolete. 

\textbf{2. Convexity} The metabolite polytope does not only provide a
good search space, but it also facilitates optimization because enzyme
cost is a convex function of the metabolite log-concentrations (see SI
\ref{sec:SIConvexCost}).  Convexity makes the optimization tractable
and scalable (see SI
\ref{sec:SIConvexCost}) -- unlike a direct optimization in enzyme
space. Simple convexity holds for a wide range of rate laws and for
extended versions of the problem, e.g., including bounds on the sum of
(non-logarithmic) metabolite levels or bounds on weighted sums of
enzyme fractions.  By adding a regularization term, representing
biological side objectives, we can even ensure strict convexity, and
thus the existence of a unique optimum that can be efficiently found.

\textbf{3. Separable rate laws disentangle individual enzyme cost
  effects} To assess how different physical factors shape metabolic
states, we focused on separable rate laws, which lead to a series of easily
interpretable, convex cost functions. The terms in these functions
represent specific physical factors and require different kinetic and
thermodynamic data for their calculation.  By neglecting some of the
terms, one obtains different approximations of the true enzyme cost.
The more terms are considered, the more precise our predictions about
metabolic states becomes (see Figure \ref{fig:pscscores} and SI
\ref{sec:SIListOfScores}).  Of course, it is often important to keep
models simple and the number of parameters small, and therefore the
stripped-down versions of ECM can be useful as well. For instance, in
some conditions such as batch-fed \emph{E. coli}, a simple enzyme
economy might still be a realistic approximation. Our results for EMC4
(see Figure \ref{fig:EcoliPredictions}) indicate that indeed one can
predict enzyme levels quite well even with this relatively simple
objective. Finally, in conditions where ECM's predictions are far
from the measured enzyme levels, we can use this information to focus
on specific enzymes or pathways that deviate the most, and that may
therefore display optimization or adaptations beyond simple
resource-optimality.

 \textbf{4. Relationship to other optimality approaches} Beyond the
 practical advantages of using factorized enzyme cost functions, they
 also allow us to easily compare our method to earlier modeling and
 optimization approaches. These approaches typically focused on only
 one or two of the factors that are taken into account in ECM, and
 many of them can be reformulated as approximations of ECM (as we
 have shown for MDF \cite{nbfr:14} and, by proxy, earlier
 thermodynamic profiling methods \cite{hlih:05, fsbh:09}).  For
 instance, the optimization performed by FBA with flux minimization is
 equivalent to using EMC0, while EMC1 is based on the same principles
 as FBA with molecular crowding \cite{bvem:07} and pathway specific
 activities \cite{bnlm:10}.  Thermodynamic profiling methods
 \cite{hlih:05, fsbh:09, nbfr:14} which use driving forces as a proxy
 for the cost, can be compared to EMC2 (where all $\kcat$ are assumed
 to be equal, see SI \ref{sec:SItighterConstraints}). To our
 knowledge, ECM is the first method that accounts for substrate and
 product saturation (as well as allosteric) effects in the
 optimization process and guarantees a convex, i.e., relatively
 tractable optimality problem. Moreover, ECM highlights how different
 aspects of metabolism are linked: most importantly, thermodynamic
 feasibility \cite{nbfr:14} is generalised by the quantitative notion
 of thermodynamic efficiency, which then turns out to be a natural
 precondition for enzyme economy. 

   \textbf{5. Improved parameters for flux analysis} Accordingly,
   results from ECM can be used to improve flux analysis
   \cite{hebh:07, tnah:13}. First, ECM can be used to define more
   realistic flux cost functions for FBA. In practice, the cost
   weights used so far (typically, defined by $\kcat$ values and
   enzyme sizes) could be adjusted by dividing them by efficiency
   factors obtained from our workflow. 
   Furthermore, ECM could be ``embedded'' into FBA by screening
   possible flux distributions (e.g., elementary flux modes) and
   characterising each of them by quantitative cost. Then the most
   cost-favorable mode could be picked.  This could be seen as a
   version of minimal-flux FBA, but one that uses kinetic knowledge
   instead of the various heuristic assumptions that go into FBA.
   Second, we can derive realistic bounds on thermodynamic forces
   based on kinetics and enzyme cost, or lower/upper bounds on
   substrates/products concentrations to avoid extreme saturation
   effects. All these constraints follow systematically from setting
   upper limits on the individual efficiency factors (see SI
   \ref{sec:SItighterConstraints}).  By applying them in
   thermodynamics-based flux analysis, we shrink the metabolite
   polytope and exclude stripes at its boundary where costs would be
   too high to allow for an optimal state. Similarly, by giving
   individually weights to thermodynamic driving forces, MDF could be
   used as a method to optimize some lower bound on the system's
   enzyme cost (see SI \ref{sec:SItighterConstraints}).

\myparagraph{Including other objectives into ECM} The assumption that
enzyme levels are continuously cost-optimized is of course debatable.
There is ample evidence that cells assume apparently sub-optimal
states in order to maintain robust homeostasis or to gain metabolic
flexibility for addressing future challenges \cite{szzh:12}. Moreover,
a random drift by mutations may affect the cell states as long as the
impact on fitness is not very high.  For example, an allosterically
regulated enzyme will often not reach its maximal possible activity,
so investment in enzyme production appears to be wasted. Nevertheless,
cells pay this price in order to gain the ability to adjust quickly to
changes (i.e. within seconds rather than the minutes required for
altering gene expression).  A simple principle of cost optimality, as
in ECM can be justify in several ways.  First, some alternative
objectives can be integrated into ECM by adding them to the objective
function. We have tried to keep our method as general as possible to
facilitate such objectives, e.g. by allowing for non-linear, convex
enzyme costs ($h(\enzyme)$).  In particular, metabolite levels may be
under additional constraints or optimality pressures because they
appear in pathways outside our model, which may favor high or low
levels of the metabolites. Also chemical molecule properties, such as
hydrophobicity or charge, may affect the preferable metabolite levels
in cells \cite{bnfb:11}. For instance, if our model captures an
ATP-producing pathway, low ATP levels will be energetically favorable,
whereas other ATP-consuming pathways would favor higher ATP levels. To
account for this trade-off, a requirement for sufficiently high ATP
levels can be included in our ECM model by constraints or additional
objectives $\metbene(\cv)$ that penalize low ATP levels.  If
metabolite levels are kept far from their upper or lower physiological
bounds, this will allow for more flexible adjustments in case of
perturbation.

If enzyme profiles were shaped by optimal resource allocation, as
assumed in ECM, this would have consequences for the shapes of enzyme
and metabolite profiles.  Enzyme cost, thermodynamic forces, and an
avoidance of low substrate levels would be tightly entangled, and the
shapes of enzyme profiles would reflect the role of enzymes in
metabolism, i.e., the way in which they control metabolic
concentrations and fluxes. Among other things, this would imply  three 
general properties of enzyme profiles:

   \textbf{1. Enzyme cost is related to thermodynamics} In FBA,
   thermodynamic constraints and flux costs appear as completely
   unrelated aspects of metabolism. Thermodynamics is used to restrict
   flux directions, and to relate them to metabolite bounds, while
   flux costs are used to suppress unnecessary fluxes.  In ECM,
   thermodynamics and flux cost appear as two sides of a coin. Like in
   FBA, flux profiles are thermodynamically \emph{feasible} if they
   lead to a finite-sized metabolite polytope, allowing for positive
   forces in all reactions. However, the values of these forces also
   play a role in shaping the enzyme cost function on that
   polytope. Together, metabolite polytope and enzyme cost function
   (as in Figure \ref{fig:fourchain}) summarize all relevant
   information about flux cost. 

  \textbf{2. Enzyme profiles reflect local metabolic necessities} What
  are the factors that determines the levels of specific enzymes? High
  levels are required whenever catalytic constants, driving forces, or
  substrate concentrations are low.  Accordingly, an efficient use of
  enzymes requires metabolite profiles with sufficient driving forces
  (for energetic efficiency) and sufficient substrate levels (for
  saturation efficiency). Trade-offs between these requirements,
  together with predefined bounds, will shape the optimal metabolite
  profiles \cite{tnah:13}: in a linear pathway, a need for energetic
  efficiency will push substrate concentrations up and product
  concentrations down; the need for saturation efficiency has the same
  effect.  However, since the product of one reaction is the substrate
  of another reaction, there will be trade-offs between efficiencies
  in different reactions. Therefore, where enzymes are costly or show
  low $\kcat$ values, we may expect a strong pressure on sufficient
  driving forces and substrate levels.

  \textbf{3. Enzyme profiles reflect global effects of enzyme usage}
  If enzyme profiles follow a cost-benefit principle, costly enzymes
  should provide large benefits.  Such a correspondence has been
  predicted, for example, from kinetic models in which flux is
  maximised at a fixed total enzyme investment \cite{hekl:1996}: in
  optimal states, high-abundance enzymes exert a strong control on the
  flux, and enzymes with strong flux control are highly abundant.  If
  this applies in reality, then highly investment (e.g., large enzyme
  levels shown in Figure \ref{fig:efficiencies} a) could be seen as a
  sign of large benefit, in terms of flux control. Here, we studied a
  different optimality problem (fixing the fluxes and optimizing
  enzyme levels under constraints on metabolite levels), and obtain a
   more general result. The optimal enzyme cost profile obtained by
  ECM is a linear combination of flux control coefficients and,
  possibly, control coefficients on metabolites that hit upper or
  lower bounds (see SI \ref{sec:SIcostAndControlCoefficients}).  In
  simple cases (e.g., the example in Figure \ref{fig:fourchain}),
  where there is only one flux mode and no metabolite hits a bound,
  and enzyme demands and flux control coefficients will be directly
  proportional.

\myparagraph{Practical application} Beyond the analysis of central
metabolism, ECM can be applied to select candidate pathways in
metabolic engineering projects. A prediction of enzyme demands or
specific activities (SI \ref{sec:SIoptimization}) can be helpful at
different stages of pathway design.  The optimal expression
profile for a pathway can be determined, critical steps in a pathway
can be detected (i.e., steps where lowering the enzyme's flux-specific
cost $\rrl$ would be most important), and enzyme demand and cost can
be compared between pathway structures. This type of application is
not unique to ECM, and although several of the methods that we
mention throughout this manuscript \cite{hmkw:97, mwhm:97, klhe:99,
  bnlm:10, tnah:13, fnbl:13} have been used for this purpose in the
past, we believe that ECM manages to bring them all under one
umbrella. 
  
\section*{Acknowledgements}
The authors thank Hermann-Georg Holzh\"utter, Andreas Hoppe, Uwe
Sauer, Tomer Shlomi, and Naama Tepper for inspiring discussions.  WL
is supported by the German Research Foundation (Ll 1676/2-1). EN is
supported by a SystemsX.ch TPdF fellowship.  RM is supported by
ERC:novcarbfix. 

\bibliographystyle{unsrt}
\bibliography{biology}

\begin{thebibliography}{10}

\bibitem{szzh:12}
R.~Schuetz, N.~Zamboni, M.~Zampieri, M.~Heinemann, and U.~Sauer.
\newblock Multidimensional optimality of microbial metabolism.
\newblock {\em Science}, 336(6081):601--604, 2012.

\bibitem{bnlm:10}
A.~Bar-Even, E.~Noor, N.E. Lewis, and R.~Milo.
\newblock Design and analysis of synthetic carbon fixation pathways.
\newblock {\em PNAS}, 107(19):8889--8894, 2010.

\bibitem{wapn:24}
O.~Warburg, K.~Posener, and E.~Negelein.
\newblock {Ueber den Stoffwechsel der Tumoren}.
\newblock {\em Biochemische Zeitschrift}, 152:319--344, 1924.

\bibitem{fnbl:13}
A.~Flamholz, E.~Noor, A.~Bar-Even, W.~Liebermeister, and R.~Milo.
\newblock Glycolytic strategy as a tradeoff between energy yield and protein
  cost.
\newblock {\em PNAS}, 110(24):10039--10044, 2013.

\bibitem{bhoz:15}
M.~Basan, S.~Hui, H.~Okano, Z.~Zhang, Y~Shen, J.R. Williamson, and T.~Hwa.
\newblock Overflow metabolism in {Escherichia coli} results from efficient
  proteome allocation.
\newblock {\em Nature}, 528:99, 2015.

\bibitem{lnfd:14}
W.~Liebermeister, E.~Noor, A.~Flamholz, D.~Davidi, J.~Bernhardt, and R.~Milo.
\newblock Visual account of protein investment in cellular functions.
\newblock {\em PNAS}, 111(23):8488--8493, 2014.

\bibitem{scks:07}
R.~Schuetz, L.~Kuepfer, and U.~Sauer.
\newblock Systematic evaluation of objective functions for predicting
  intracellular fluxes in {Escherichia} coli.
\newblock {\em Molecular Systems Biology}, 3:119, 2007.

\bibitem{belq:02}
D.~A. Beard, S.~Liang, and H.~Qian.
\newblock Energy balance for analysis of complex metabolic networks.
\newblock {\em Biophysical Journal}, 83(1):79--86, 2002.

\bibitem{bbcq:04}
D.A. Beard, E.~Babson, E.~Curtis, and H.~Qian.
\newblock Thermodynamic constraints for biochemical networks.
\newblock {\em J. Theor. Biol.}, 228(3):327--333, 2004.

\bibitem{faqb:05}
F.~Yang, H.~Qian, and Daniel~A. Beard.
\newblock Ab initio prediction of thermodynamically feasible reaction
  directions from biochemical network stoichiometry.
\newblock {\em Metabolic Engineering}, 7(4):251--259, 2005.

\bibitem{fmsy:11}
R.M.T. Fleming, C.M. Maes, M.A. Saunders, Y.~Ye, and B.\O. Palsson.
\newblock A variational principle for computing nonequilibrium fluxes and
  potentials in genome-scale biochemical networks.
\newblock {\em J. Theor. Biol.}, 292:71--77, 2012.

\bibitem{hjbh:06}
C.S. Henry, M.D. Jankowski, L.J. Broadbelt, and V.~Hatzimanikatis.
\newblock Genome-scale thermodynamic analysis of {E.~coli} metabolism.
\newblock {\em Biophys. J.}, 90:1453--1461, 2006.

\bibitem{hebh:07}
C.S. Henry, L.J. Broadbelt, and V.~Hatzimanikatis.
\newblock Thermodynamics-based metabolic flux analysis.
\newblock {\em Biophys J.}, 92(5):1792--1805, 2007.

\bibitem{hohh:07}
A.~Hoppe, S.~Hoffmann, and H.-G. Holzh\"utter.
\newblock Including metabolite concentrations into flux-balance analysis:
  Thermodynamic realizability as a constraint on flux distributions in
  metabolic networks.
\newblock {\em BMC Syst. Biol}, 1(1):23, 2007.

\bibitem{nbfr:14}
E.~Noor, A.~Bar-Even, A.~Flamholz, E.~Reznik, W.~Liebermeister, and R.~Milo.
\newblock Pathway thermodynamics uncovers kinetic obstacles in central
  metabolism.
\newblock {\em PLoS Comp. Biol.}, 10:e100348, 2014.

\bibitem{schp:08}
T.~Shlomi, M.~Cabili, M.~Herrgard, B.\O. Palsson, and E.~Ruppin.
\newblock Network-based prediction of human tissue-specific metabolism.
\newblock {\em Nature Biotechnology}, 26(9):1003, 2008.

\bibitem{holz:04}
H.-G. Holzh\"utter.
\newblock The principle of flux minimization and its application to estimate
  stationary fluxes in metabolic networks.
\newblock {\em Eur. J. Biochem.}, 271(14):2905--2922, 2004.

\bibitem{lhcl:10}
N.E. Lewis, K.K. Hixson, T.M. Conrad, J.A. Lerman, P.~Charusanti, A.D.
  Polpitiya, J.N. Adkins, G.~Schramm, S.O. Purvine, D.~Lopez-Ferrer, K.K.
  Weitz, R.~Eils, R.~K\"onig, R.D. Smith, , and B~\O Palsson.
\newblock Omic data from evolved e. coli are consistent with computed optimal
  growth from genome-scale models.
\newblock {\em Mol. Syst Biol.}, 6:390, 2010.

\bibitem{bvem:07}
Q.K. Beg, A.~Vazquez, J.~Ernst, M.A. de~Menezes, Z.~Bar-Joseph, A.-L.
  Barab\'asi, and Z.N. Oltvai.
\newblock Intracellular crowding defines the mode and sequence of substrate
  uptake by {Escherichia coli} and constrains its metabolic activity.
\newblock {\em PNAS}, 104(31):12663--12668, 2007.

\bibitem{horh:11}
A.~Hoppe, C.~Richter, and H.-G. Holzh\"utter.
\newblock Enzyme maintenance effort as criterion for the characterization of
  alternative pathways and length distribution of isofunctional enzymes.
\newblock {\em Biosystems}, 105(2):122--129, 2011.

\bibitem{liuk:10}
W.~Liebermeister, J.~Uhlendorf, and E.~Klipp.
\newblock Modular rate laws for enzymatic reactions: thermodynamics,
  elasticities, and implementation.
\newblock {\em Bioinformatics}, 26(12):1528--1534, 2010.

\bibitem{nflb:13}
E.~Noor, A.~Flamholz, W.~Liebermeister, A.~Bar-Even, and R.~Milo.
\newblock A note on the kinetics of enzyme action: {a} decomposition that
  highlights thermodynamic effects.
\newblock {\em FEBS Letters}, 587(17):2772--2777, 2013.

\bibitem{tnah:13}
N.~Tepper, E.~Noor, D.~Amador-Noguez, H.S. Haraldsd\'ottir, R.~Milo,
  J.~Rabinowitz, W.~Liebermeister, and T.~Shlomi.
\newblock Steady-state metabolite concentrations reflect a balance between
  maximizing enzyme efficiency and minimizing total metabolite load.
\newblock {\em PLoS ONE}, 8(9):e75370, 2013.

\bibitem{sche:91}
S.~Schuster and R.~Heinrich.
\newblock Minimization of intermediate concentrations as a suggested optimality
  principle for biochemical networks.
\newblock {\em Journal of Mathematical Biology}, 29(5):425--442, 1991.

\bibitem{bnsl:11}
A.~Bar-Even, E.~Noor, Y.~Savir, W.~Liebermeister, D.~Davidi, D.S. Tawfik, and
  R.~Milo.
\newblock The moderately efficient enzyme: evolutionary and physicochemical
  trends shaping enzyme parameters.
\newblock {\em Biochemistry}, 21:4402--4410, 2011.

\bibitem{eune:10}
K.~van Eunen, J.~Bouwman, P.~Daran-Lapujade, J.~Postmus, A.B. Canelas, F.I.
  Mensonides, R.~Orij, I.~Tuzun, J.~van~den Brink, G.J. Smits, W.M. van Gulik,
  S.~Brul, J.J. Heijnen, J.H. de~Winde, M.J. de~Mattos, C.~Kettner, J.~Nielsen,
  H.V. Westerhoff, and B.M. Bakker.
\newblock Measuring enzyme activities under standardized in vivo-like
  conditions for systems biology.
\newblock {\em FEBS Journal}, 277:749–760, 2010.

\bibitem{smal:13}
K.~Smallbone, H.L. Messiha, K.M. Carroll, C.L. Winder, N.~Malys, W.B. Dunn,
  E.~Murabito, N.~Swainston, J.O. Dada, F.~Khan, P.~Pir, E.~Simeonidis,
  I~Spasi\'c, J.~Wishart, D.~Weichart, N.W. Hayes, D.~Jameson, D.S. Broomhead,
  S.G. Oliver, S.J. Gaskell, J.E. McCarthy, N.W. Paton, H.V. Westerhoff, D.B.
  Kell, and P.~Mendes.
\newblock A model of yeast glycolysis based on a consistent kinetic
  characterisation of all its enzymes.
\newblock {\em FEBS Letters}, 587:2832–2841, 2013.

\bibitem{reic:83}
J.G. Reich.
\newblock Zur {{\"O}}konomie im {P}roteinhaushalt der lebenden {Z}elle.
\newblock {\em Biomed.~Biochim.~Acta}, 42(7/8):839--848, 1983.

\bibitem{klhh:02}
E.~Klipp, R.~Heinrich, and H.-G. Holzh\"utter.
\newblock Prediction of temporal gene expression. {M}etabolic optimization by
  re-distribution of enzyme activities.
\newblock {\em Eur. J. Biochem.}, 269:1--8, 2002.

\bibitem{lksh:04}
W.~Liebermeister, E.~Klipp, S.~Schuster, and R.~Heinrich.
\newblock A theory of optimal differential gene expression.
\newblock {\em BioSystems}, 76:261--278, 2004.

\bibitem{mbrt:09}
D.~Molenaar, R.~van Berlo, D.~de~Ridder, and B.~Teusink.
\newblock Shifts in growth strategies reflect tradeoffs in cellular economics.
\newblock {\em Molecular Systems Biology}, 5:323, 2009.

\bibitem{zabl:13}
L.~Zelcbuch, N.~Antonovsky, A.~Bar-Even, A~Levin-Karp, U.~Barenholz, M.~Dayagi,
  W.~Liebermeister, A.~Flamholz, E.~Noor, S.~Amram, A.~Brandis, T.~Bareia,
  I.~Yofe, H.~Jubran, and R.~Milo.
\newblock Spanning high-dimensional expression space using ribosome-binding
  site combinatorics.
\newblock {\em Nucleic Acids Research}, 41(9):e98, 2013.

\bibitem{laht:13}
M.E. Lee, A.~Aswani~A.S. Han, C.J. Tomlin, and J.E. Dueber.
\newblock Expression-level optimization of a multi-enzyme pathway in the
  absence of a high-throughput assay.
\newblock {\em Nucleic Acids Res.}, 41(22):10668--10678, 2013.

\bibitem{beqi:07}
D.A. Beard and H.~Qian.
\newblock Relationship between thermodynamic driving force and one-way fluxes
  in reversible processes.
\newblock {\em PLoS ONE}, 2(1):e144, 2007.

\bibitem{deal:05}
E.~Dekel and U.~Alon.
\newblock Optimality and evolutionary tuning of the expression level of a
  protein.
\newblock {\em Nature}, 436:588--692, 2005.

\bibitem{szad:10}
I.~Shachrai, A.~Zaslaver, U.~Alon, and E.~Dekel.
\newblock Cost of unneeded proteins in {E. coli} is reduced after several
  generations in exponential growth.
\newblock {\em Molecular Cell}, 38:1--10, 2010.

\bibitem{hogr:13}
J.S. Hofmeyr, O.P.C. Gqwaka, and J.M. Rohwer.
\newblock A generic rate equation for catalysed, template-directed
  polymerisation.
\newblock {\em FEBS Letters}, 587:2868--2875, 2013.

\bibitem{lskl:10}
T.~Lubitz, M.~Schulz, E.~Klipp, and W.~Liebermeister.
\newblock Parameter balancing for kinetic models of cell metabolism.
\newblock {\em J. Phys. Chem. B}, 114(49):16298--16303, 2010.

\bibitem{slsk:13}
N.J. Stanford, T.~Lubitz, K.~Smallbone, E.~Klipp, P.~Mendes, and
  W.~Liebermeister.
\newblock Systematic construction of kinetic models from genome-scale metabolic
  networks.
\newblock {\em PLoS ONE}, 8(11):e79195, 2013.

\bibitem{jhbh:08}
M.D. Jankowski, C.S. Henry, L.J. Broadbelt, and V.~Hatzimanikatis.
\newblock Group contribution method for thermodynamic analysis of complex
  metabolic networks.
\newblock {\em Biophys. J.}, 95(3):1487--1499, 2008.

\bibitem{nhmf:13}
E.~Noor, H.S. Haraldsdottir, R.~Milo, and R.M.T. Fleming.
\newblock Consistent estimation of {Gibbs} energy using component
  contributions.
\newblock {\em PLOS Comp. Biol.}, 9:e1003098, 2013.

\bibitem{fnbm:12}
A.~Flamholz, E.~Noor, A.~Bar-Even, and R.~Milo.
\newblock equilibrator -- the biochemical thermodynamics calculator.
\newblock {\em Nucleic Acids Research}, 40(D1):D770--D775, 2012.

\bibitem{likl:06a}
W.~Liebermeister and E.~Klipp.
\newblock Bringing metabolic networks to life: convenience rate law and
  thermodynamic constraints.
\newblock {\em Theor.~Biol.~Med.~Mod.}, 3:41, 2006.

\bibitem{hann:11}
B.R.B.H. van Rijsewijk, A.~Nanchen, S.~Nallet, R.J. Kleijn, and U.~Sauer.
\newblock Large-scale 13c-flux analysis reveals distinct transcriptional
  control of respiratory and fermentative metabolism in escherichia coli.
\newblock {\em Mol. Syst. Biol.}, 7(477):477, 2011.

\bibitem{sceg:04}
I.~Schomburg, A.~Chang, C.~Ebeling, M.~Gremse, C.~Heldt, G.~Huhn, and
  D.~Schomburg.
\newblock {BRENDA, the} enzyme database: updates and major new developments.
\newblock {\em Nucleic Acids Research}, 32:Database issue:D431--433, 2004.

\bibitem{skva:15}
A.~Schmidt, K.~Kochanowski, S.~Vedelaar, E.~Ahrn\'e, B.~Volkmer, L.~Callipo,
  K.~Knoops, M.~Bauer, R.~Aebersold, and M.~Heinemann.
\newblock The quantitative and condition-dependent escherichia coli proteome.
\newblock {\em Nature Biotechnology}, page doi:10.1038/nbt.3418, 2015.

\bibitem{milo:13}
Ron Milo.
\newblock What is the total number of protein molecules per cell volume? {A}
  call to rethink some published values.
\newblock {\em BioEssays}, 35(12):1050–1055, 2013.

\bibitem{avpn:12}
L.~Arike, K.~Valgepea, L.~Peil, R.~Nahku, K.~Adamberg, and R.~Vilu.
\newblock Comparison and applications of label-free absolute proteome
  quantification methods on {Escherichia} coli.
\newblock {\em J Proteomics}, 75(17):5437--5448, 2012.

\bibitem{sgsb:06}
R.~Steuer, T.~Gross, J.~Selbig, and B.~Blasius.
\newblock Structural kinetic modeling of metabolic networks.
\newblock {\em Proc Natl Acad Sci USA}, 103(32):11868--11873, 2006.

\bibitem{hlih:05}
V.~Hatzimanikatis, C.~Li, J.A. Ionita, C.S. Henry, M.D. Jankowski, and L.J.
  Broadbelt.
\newblock Exploring the diversity of complex metabolic networks.
\newblock {\em Bioinformatics}, 21(8):1603--1609, 2005.

\bibitem{fsbh:09}
S.D. Finley, L.J. Broadbelt, and V.Hatzimanikatis.
\newblock Computational framework for predictive biodegradation.
\newblock {\em Biotechnol. Bioeng.}, 104(6):1086--1097, 2009.

\bibitem{bnfb:11}
A.~Bar-Even, E.~Noor, A.~Flamholz, J.M. Buescher, and R.~Milo.
\newblock Hydrophobicity and charge shape cellular metabolite concentrations.
\newblock {\em PLoS Computational Biology}, 7(10):e1002166, 2011.

\bibitem{hekl:1996}
R.~Heinrich and E.~Klipp.
\newblock Control analysis of unbranched enzymatic chains in states of maximal
  activity.
\newblock {\em J. Theor. Biol.}, 182(3):243--252, 1996.

\bibitem{hmkw:97}
R.~Heinrich, F.~Montero, E.~Klipp, T.G. Waddell, and E.~Mel\'{e}ndez-Hevia.
\newblock Theoretical approaches to the evolutionary optimization of glycolysis
  -- thermodynamic and kinetic constraints.
\newblock {\em Eur. J. Biochem.}, 243:191--201, 1997.

\bibitem{mwhm:97}
E.~Mel\'endez-Hevia, T.G. Waddell, R.~Heinrich, and F.~Montero.
\newblock Theoretical approaches to the evolutionary optimization of glycolysis
  -- chemical analysis.
\newblock {\em Eur. J. Biochem.}, 244:527--543, 1997.

\bibitem{klhe:99}
E.~Klipp and R.~Heinrich.
\newblock Competition for enzymes in metabolic pathways: implications for
  optimal distributions of enzyme concentrations and for the distribution of
  flux control.
\newblock {\em BioSystems}, 54:1--14, 1999.

\bibitem{wegs:02}
R.~Wegscheider.
\newblock {\"Uber simultane Gleichgewichte und die Beziehungen zwischen
  Thermodynamik und Reactionskinetik homogener Systeme}.
\newblock {\em Z. Phys. Chem.}, 39:257--303, 1902.

\bibitem{hald:30}
J.B.S. Haldane.
\newblock {\em Enzymes}.
\newblock Longmans, Green and Co., London. (republished in 1965 by MIT Press,
  Cambridge, MA), 1930.

\bibitem{KEGG}
M.~Kanehisa, S.~Goto, S.~Kawashima S, and A.~Nakaya.
\newblock The {KEGG} databases at genomenet.
\newblock {\em Nucleic Acids Research}, 30:42--46, 2002.

\bibitem{gerc:15}
L.~Gerosa, B.R.B.H. van Rijsewijk, D.~Christodoulou, K.~Kochanowski, T.S.B.
  Schmidt, E.~Noor, and U.~Sauer.
\newblock Pseudo-transition analysis identifies the governing regulation of
  microbial nutrient adaptations from steady state data.
\newblock {\em Cell Systems}, 1:270--282, 2015.

\bibitem{zhvm:11}
K.~Zhuang, G.N. Vemuri, and R.~Mahadevan.
\newblock Economics of membrane occupancy and respiro-fermentation.
\newblock {\em MSB}, 7:500, 2011.

\bibitem{digs:11}
K.A. Dill, K.~Ghosh, and J.D. Schmit.
\newblock Physical limits of cells and proteomes.
\newblock {\em PNAS}, 108(44):17876--17882, 2011.

\bibitem{eako:12}
M.~Eames and T.~Kortemme.
\newblock Cost-benefit tradeoffs in engineered lac operons.
\newblock {\em Science}, 336:911--915, 2012.

\bibitem{akgo:02}
H.~Akashi and T.~Gojobori.
\newblock Metabolic efficiency and amino acid composition in the proteomes of
  {Escherichia coli} and {Bacillus subtilis}.
\newblock {\em PNAS}, 99(6):3695--3700, 2002.

\end{thebibliography}

\section{Methods}

\paragraph{Metabolite polytope and enzyme cost functions}
A metabolic network with given flux directions, equilibrium constants,
and metabolite bounds defines the \emph{metabolite polytope}. This
convex polytope $\metabolitepolytope$ in the space of
log-concentrations $s_i = \ln c_i$ represents the set of feasible
metabolite profiles.  The flux profile used can be stationary
(e.g. determined by FBA or $^{13}$C MFA) or non-stationary (like
experimentally measured fluxes, directly inserted into a model). If
the provided flux directions are thermodynamically infeasible, the
metabolite polytope will be an empty set, $\metabolitepolytope =
\emptyset$.  The faces of the metabolite polytope arise from two types
of inequality constraints. First, the physical ranges $s^{\rm min}_{i}
\le s_{i} \le s^{\rm max}_{i}$ of metabolite levels define a
box-shaped polytope (bounded by P-faces). Some metabolite levels may
even be constrained to fixed values. Second, each reaction must
dissipate Gibbs free energy, and to make this possible, driving forces
and fluxes must have the same signs ($\Theta_l \cdot v_l>0$), and thus
$\sign(v_l) = \sign(\Delta_{\rm r} {G'}^{\circ}_{l}/RT + \sum_{i}
n_{il} s_{i})$.  The resulting constraints define E-faces of the
metabolite polytope (representing equilibrium states, $\Theta_l =
0$). Close to these faces, enzyme cost goes to infinity.

\paragraph{Enzyme cost minimization can be formulated as a convex
 optimality problem for metabolite levels} Enzyme cost minimization (ECM) uses a metabolic network, a flux profile $\vv$,
kinetic rate laws, enzyme burdens, and bounds on metabolite
levels to predict optimal metabolite and enzyme concentrations. The
enzyme cost of reactions or pathways is a convex function on the
metabolite polytope (proof in SI \ref{sec:SIConvexCost}), that is, a
metabolite vector $\sv$, linearly interpolated between vectors
$\sv_{\rm a}$ and $\sv_{\rm b}$, cannot have a higher cost than the
interpolated cost of $\sv_{\rm a}$ and $\sv_{\rm b}$. Convexity also
holds for cost functions $\hminus(\enzymev)$ that are non-linear, but
convex over $\enzymev$.  Some EMC functions are even strictly convex
(i.e., Eq.~(\ref{eq:ExplainConvexity}) holds with a $<$ sign instead
of $\leq$). In contrast, simplified EMC functions can be constant (as
in EMC0 and EMC1), or constant in certain directions in the metabolite
polytope (as in EMC2, under combined metabolite variations that do not
affect the driving forces) (see SI \ref{sec:SIproofUniqueMapping}). To
find an optimal state, we choose an EMC function and minimize the
total enzyme cost within the metabolite polytope.  Optimal metabolite
profiles, enzyme profiles, and enzyme costs are obtained by solving
the enzyme cost minimization (ECM) problem
\begin{eqnarray}
\label{eq:ECMproblem}
\sv^{\rm opt}(\vv)  &=& \mbox{argmin}_{\sv \in \metabolitepolytope}\, \enzymemetcost(\sv, \vv)\nonumber \\
\enzymev^{\rm opt}(\vv) &=& \enzymev(\sv^{\rm opt}(\vv),\vv)\nonumber \\
\enzymemetcost^{\rm opt}(\vv) &=& \enzymemetcost(\sv^{\rm opt}, \vv).
\end{eqnarray}
The total cost $\enzymemetcost(\sv,\vv)$ (defined in
Eq.~(\ref{eq:TotalEnzymeDemand})) is the sum of enzyme costs given by
EMC functions. Since $\enzymemetcost(\sv)$ and the metabolite polytope
itself are convex, ECM is a convex optimization problem.  The optimal
enzyme levels depend on external conditions and have to be
recalculated after any change in external metabolite levels.
There are cases where
$\enzymemetcost(\sv)$ is convex, but not strictly convex, and
therefore Eq. (\ref{eq:ECMproblem}) will have a continuum of optimal
solutions. To enforce a unique solution, one may add strictly convex
side objectives that score the log-metabolite levels, e.g., a
quadratic function favoring metabolite levels close to some typical
concentration vector $\hat{\sv}$: $\mbox{min}_{\sv \in
  \metabolitepolytope}\, \left(\enzymemetcost(\sv, \vv) + ||\sv -
\hat{\sv}||\right)$. Such extra objectives can be justified
biologically, e.g.~by assuming that intermediate metabolite levels
give cells more flexibility to adapt to perturbations.  Convexity does
not only simplify numerical calculations, but it also shows that the
evolutionary optimality problem  has a unique
solution.  In fact, metabolite polytope and cost functions remain
convex even under various modifications of the problem.  The shape of
the feasible set (usually, the metabolite polytope) remains convex if
we add constraints on the total metabolite level, on weighted sums of
metabolite levels, or on weighted sums of enzyme levels (see SI
\ref{sec:SIworkflowDetails}).  Finally, we can consider the more
complicated problem of preemptive enzyme expression, where a fixed
enzyme profile and  allosteric inhibition must allow a cell to
realise different flux distributions under different conditions (see
SI \ref{eq:preemptiveExpression}). Also this problem is convex. If a
model contains non-enzymatic reactions (or non-enzymatic processes
such as metabolite diffusion out of the cell or dilution in growing
cells), each such reaction leads to an extra constraint on the
metabolite polytope (for details, see SI \ref{sec:SInonEnzymatic}).  A
known flux in an irreversible diffusion or dilution reaction fixes the
concentration of one metabolite. In the presence of irreversible
non-enzymatic reactions with mass-action rate laws, the polytope is
intersected by a subspace. In both cases, the resulting sub-polytope
may be empty, i.e., the given flux distribution will not be
realisable.

\paragraph{Tolerance ranges for nearly optimal solutions} 
Evolution could tolerate non-optimal enzyme
costs; this tolerance depends on population dynamics and can sometimes
be quite significant, e.g. in small compartmentalized communities.  To
compute realistic tolerance ranges for the ECM problem, we start from
the optimum (total cost $\enzymemetcost$) and choose a tolerable cost
$\enzymemetcost^{\rm tol}$ (e.g., one percent higher than the optimal
cost). This defines a tolerable region in $\metabolitepolytope$:
$\metabolitepolytope^{\rm tol} \equiv \{\sv \in \metabolitepolytope
\,|\, \enzymemetcost(\sv) \le \enzymemetcost^{\rm tol}\}$.  A
tolerance range for each metabolite is defined by the minimal and
maximal values the metabolite can show within
$\metabolitepolytope^{\rm tol}$.  Tolerance ranges for enzyme levels
are defined in a similar way. Alternatively, tolerance ranges and
nearly optimal solutions can be estimated from the Hessian matrix (see
SI \ref{sec:SIHessianApproximation}).

\paragraph{Enzyme-based flux cost function} In FBA (e.g., in FBA variants with
flux minimization or molecular crowding), flux cost or enzyme demand
are linear functions of the fluxes. ECM yields plausible prefactors
for this formula: by rearranging Eq.~(\ref{eq:TotalEnzymeDemand}), we
can write the enzyme cost as a linear function $\enzymemetcost =
\sum_l \rrl\cdot v_l$ with flux burdens $\rrl(\cv) = \hel\cdot
\frac{1}{\kcatplusl} \cdot \frac{1}{\eta^{\rm th}_l(\cv)} \cdot
\frac{1}{\eta_l^{\rm sat}(\cv)} \cdot \frac{1}{\eta^{\rm reg}(\cv)}$.
The flux burden has a lower bound $\rrlmin = \hel / \kcatplusl$,
denoting the cost per flux under ideal conditions. Ignoring all
dependencies on metabolite levels, $\rrlmin$ could be used as a cost
weight to define flux cost functions for FBA.  However, these values
are further increased by the inverse enzyme efficiencies. A
flux-specific enzyme cost (or, inversely, a flux per enzyme invested)
can also be defined for entire pathways.  The Pathway Specific
Activity \cite{bnlm:10} is defined as the flux per enzyme mass
concentration (with flux in mM/s and enzyme mass concentration in
$\mu$g enzyme per gram of cell dry weight) and can be computed by
treating enzyme mass as a cost function.  Assuming that $\eta^{\rm
  enr} = \eta^{\rm kin} = 1$ and that cost is
expressed in terms of protein mass ($\hel=m_l$), we obtain the Pathway
Specific Activity by dividing the pathway flux $\vPW$ by
$\enzymemetcost$ (see SI \ref{sec:SIpathwayCost} and
\ref{sec:SIpsa}). Using protein masses in Daltons as specific cost
weights $\hel$, we obtain a formula for the enzyme
mass per flux (for reactions or pathways). The reciprocal value
$\activity_{\rm pw} = \vPW/\enzymemetcost$ is the pathway-specific
activity.

\paragraph{Workflow for model building and enzyme prediction.}  
To predict enzyme and metabolite levels in metabolic pathways (Figure
\ref{fig:workflowData}) we developed an automated workflow.  In a
consistent model, all parameters must satisfy Wegscheider conditions
for equilibrium constants \cite{wegs:02} and Haldane relationships
between equilibrium constants and rate constants \cite{hald:30}.  The
kinetic constants used in rate laws should represent \emph{effective}
parameters, which may differ from ``ideal'' parameters, e.g., by
crowding effects. However, since measured parameter values are usually
incomplete and inconsistent, parameter balancing \cite{lskl:10} is
used to translate measured kinetic constants into consistent model
parameters.  Based on a network and given fluxes, the software
extracts relevant data from a database (thermodynamic constants, rate
constants, fluxes, and protein sizes; metabolite and protein levels
for validation), determines a consistent set of model parameters,
builds a kinetic model, and optimizes enzyme and metabolite profiles
for the EMC function chosen. To assess the effects of parameter
variation, parameter sets can be sampled from the posterior
distribution, provided by parameter balancing. Sampled parameters lead
to different predicted enzyme levels, but the resulting variation in
enzyme levels can be explained, to a large extent, as a direct
compensation for the varying $\kcat$ values.  The workflow has been
implemented in MATLAB and python.

\paragraph{\emph{E.~coli} model} 
The model shown in Figure \ref{fig:EcoliPredictions} was built automatically
from a list of chemical reactions in \emph{E.~coli} central
metabolism. KEGG reaction identifiers \cite{KEGG} were automatically
translated into a kinetic model
(for details, see SI \ref{sec:SIecolimodel} and SI Table
\ref{tab:data}). The cofactors ATP, ADP, phosphate, NADH, NAD$^+$,
NADPH, and NADP$^+$ are included in the model.  Equilibrium constants
were estimated using the component contribution method \cite{nhmf:13},
kinetic constants ($\kcatplus$ and $\km$ values) were obtained from
the BRENDA database (after which each value was curated manually), and
a complete, globally consistent parameter set was determined by
parameter balancing.   State-dependent data were obtained from
publications using batch fed \emph{E. coli} BW25113 grown on minimal
media (M9) with glucose as the carbon source.  Our source
for metabolic fluxes \cite{hann:11} used $^{13}$C metabolic flux
analysis, metabolite concentrations \cite{gerc:15} were obtained using
LC-MS/SM, and enzyme concentrations \cite{skva:15} using SWATH-MS.
For a summary of data provenance, see SI Table \ref{tab:data}.  During
ECM, all metabolite levels were limited to predefined ranges, and the
levels of cofactors and some other metabolites were fixed at
experimentally known values.  To compute tolerances for predicted
metabolite and enzyme levels, we  defined an acceptable
enzyme cost, one percent higher than the minimal value, and determined
ranges for metabolite levels that agree with this cost limit.  Data,
model, and matlab code for ECM can be obtained from 
\url{www.metabolic-economics.de/enzyme-cost-minimization/}..

\begin{table}[t!]
\begin{center}
  \begin{tabular}{|>{\columncolor{lightyellow}}l>{\columncolor{lightyellow}}l>{\columncolor{lightyellow}}l|}
\hline
    \rowcolor{brown} \textbf{Name} &
  \textbf{Symbol}&
 \textbf{Unit}\\ \hline & & \\[-2mm]
    Flux & $v_{l}$& mM/s \\
    Metabolite level & $c_{i}$ & mM     \\
    Logarithmic metabolite level & $s_i = \ln (c_{i}/c_{\sigma})$ & unitless  \\
    Enzyme level & $\el$ & mM\\
    Reaction rate & $v_{l}(\el, \cv) = \el \cdot \ratelaw_{l}(\cv)$ & mM/s \\
    Catalytic rate & $\ratelaw_{l} = v_{l}/\el$& 1/s\\
    Scaled reactant elasticity & $\Esc_{li}$ & unitless \\
    Gibbs energy of formation (std.~chemical potential)  & ${G'}^{\circ}_{i}$ & kJ/mol\\
    Reaction Gibbs energy & $\Delta_{\rm r} G'_{l} = \Delta_{\rm r} {G'_{l}}^{\circ}+ RT \sum_i n_{il}\,\ln c_{i}]$ & kJ/mol\\
    Driving force & $\Theta_l = - \Delta_{\rm r} G'_{l}/RT$ & unitless \\
    Forward/backward catalytic constant  & $\kcatplus, \kcatminus$      & 1/s    \\
    Michaelis-Menten constant     & $\kmli$ & mM     \\
    Protein mass & $m_l$ & Da \\
    Enzyme cost & $\hminus(\enzymev) = \sum_{l } \hel \,\el$ & D \\
    Enzyme burden & $\hel$ & D/mM \\
    Enzyme-induced metabolite cost & $\enzymemetcost(\sv,\vv) = \hminus(\enzymev(\sv,\vv))$ & D \\
    Flux-specific cost & $\rrl$ & D/(mM/s) \\
    Baseline flux-specific cost & $\rrlmin$ & D/(mM/s) \\
\hline
  \end{tabular}
\end{center}
\caption{Mathematical symbols used.  The fitness unit Darwin (D) is
  a proxy for the different fitness  units used in cell
  models. Reaction must be orientated in such a way that all
  fluxes are positive.  To define metabolite log-concentrations, we
  use the standard concentration $c_\sigma=$ 1 mM.}
\label{tab:symbolsshort}
  \end{table}

\clearpage

\renewcommand{\thesection} {S\arabic{section}}
\renewcommand{\thefigure}  {S\arabic{figure}}
\renewcommand{\thetable}   {S\arabic{table}}
\renewcommand{\theequation}{S\arabic{equation}}

\setcounter{figure}{0}
\setcounter{equation}{0}
\setcounter{section}{0}

\ \\
\begin{center}
{\LARGE{Supplementary information}}
\end{center}
\ \\[-1cm]

\section{Kinetic rate laws}

\subsection{Rate laws for general enzymatic reactions}
\label{sec:SIratelaws}

\myparagraph{Reversible rate laws}
Reversible rate laws for reactions with multiple substrates
(concentrations $s_{i}$) and products (concentrations $p_{j}$) have
the form
\begin{eqnarray}
  \label{eq:GeneralRateLawRate}
  v &=&  \enzyme\, \frac{k^+_{\rm cat} \prod_{i} (\frac{s_{i}}{\kmi})^{m^{\rm S}_{i}}
    - k^-_{\rm cat} \prod_{i} (\frac{p_{i}}{\kmi})^{m^{\rm P}_{i}}}{D(s_1, s_2, ..,p_1, p_2, ..)}.
\end{eqnarray}
By default, we assume that an enzyme molecule contains a single
catalytic site.  If an enzyme is a protein complex with $N_{\rm sub}$
subunits and $N_{\rm cat}$ catalytic sites, we can use effective
values ${\kcatplus}' = \frac{N_{\rm cat}}{N_{\rm sub}}\,\kcatplus$
referring to single enzyme subunits, whose concentrations $\enzyme$
are recorded in proteomics data.  The molecularities $m^{\rm S}_{li}$
or $m^{\rm P}_{li}$ describe in what numbers reactants participate in
the enzyme mechanism.  Molecularities can differ from the (nominal)
stoichiometric coefficients by a reaction-specific scaling factor
$\gamma$ because the stoichiometric coefficients in the sum formula
may be arbitrarily scaled.  For example, in a reaction 2 A + 4 B
$\rightarrow$ 2 C (stoichiometric coefficients -2, -4, 2) with th rate
law $\kcatplus\,[A]\,[B]^{2} -\kcatminus\,[C]$ (with molecularities 1,
2, 1), this factor would be $\gamma = 1/2$.  If we assume, as it is
usually done, that the exponents in the rate law represent
stoichiometric coefficients, the factor $\gamma_l$ will appear like an
effective Hill coefficient. In turn, a reactant with stoichiometric
coefficient $n_{il}$ and an effective Hill coefficient $\gamma=2$ will
have a molecularity of $m^{\rm S}_{li} = 2 |n_{il}|$.  For reasons of
thermodynamic consistency (existence of a consistent equilibrium
state), all substrates and products in a reaction must show the same
$\gamma$ factors \cite{liuk:10}.

\myparagraph{Thermodynamic driving force and thermodynamic
  constraints} The signs and magnitudes of metabolic fluxes depend on
thermodynamic driving forces (see Figure \ref{fig:fluxfractions}).  We
define the thermodynamic driving force as the negative reaction Gibbs
energy $-\Delta_{\rm r} G'$, measured in units of $RT$. The symbol G'
denotes transformed Gibbs free energies, suitable variables for
systems at given or buffered pH value.  To obtain the correct
relationship between fluxes and driving forces, the driving forces
$\Theta_l$ must be defined based on molecularities, not on
stoichiometric coefficients\footnote{If stoichiometric coefficients
  and molecularities differ, the Hill-like coefficient $\gamma_l$ must
  appear in the definition of driving forces $\Theta_l = -
  \frac{1}{RT}\, \sum_i \gamma_l\,n_{il}\, G'_l = - \gamma_l\,
  \Delta_{\rm r} G'_l$. The difference along a reaction is not defined
  based on nominal stoichiometric coefficients, but on actual
  molecularities.}.  According to thermodynamics, all reaction rates
must vanish in chemical equilibrium; to ensure this in kinetic models,
equilibrium constants and rate constants must satisfy the Haldane
relationship \cite{hald:30}
\begin{eqnarray}
  \label{eq:HaldaneRelationShip}
 \keq =  \frac{\prod_{i} (s^{\rm eq}_{i})^{m^{\rm S}_{i}}}
{\prod_{i} (p^{\rm eq}_{i})^{m^{\rm P}_{i}}}
=  \frac{  k^+_{\rm cat} \prod_{i} (\kmi)^{m^{\rm P}_{i}}}
{k^-_{\rm cat}  \prod_{i} (\kmi)^{m^{\rm S}_{i}}},
\end{eqnarray}
where $s_i$ and $p_j$ denote substrate and product levels,
respectively.  Moreover, the equilibrium constants follow from Gibbs
energies of formation as $\keq = \e^{-\Delta_{\rm r} {G^{\circ}}'/RT}$. This implies
 Wegscheider conditions \cite{wegs:02}: the
vector of equilibrium constants satisfies $\ln \keq = {\Ntot}\trans
\,{\muv^\circ}'$, with the stoichiometric matrix $\Ntot$ for all
metabolites and the vector ${\muv^\circ}'$ of transformed Gibbs free
energies of formation. Accordingly, the equilibrium constants must
satisfy a Wegscheider condition $\ln \keq \cdot \kv = 0$ for any
thermodynamic cycle $\kv$, i.e., any nullspace vector of
${\Ntot}\trans$.

\begin{figure}[t!]
  \begin{center} \begin{tabular}{ll} (a) &
    (b) \\[2mm]
      \includegraphics[width=6cm]{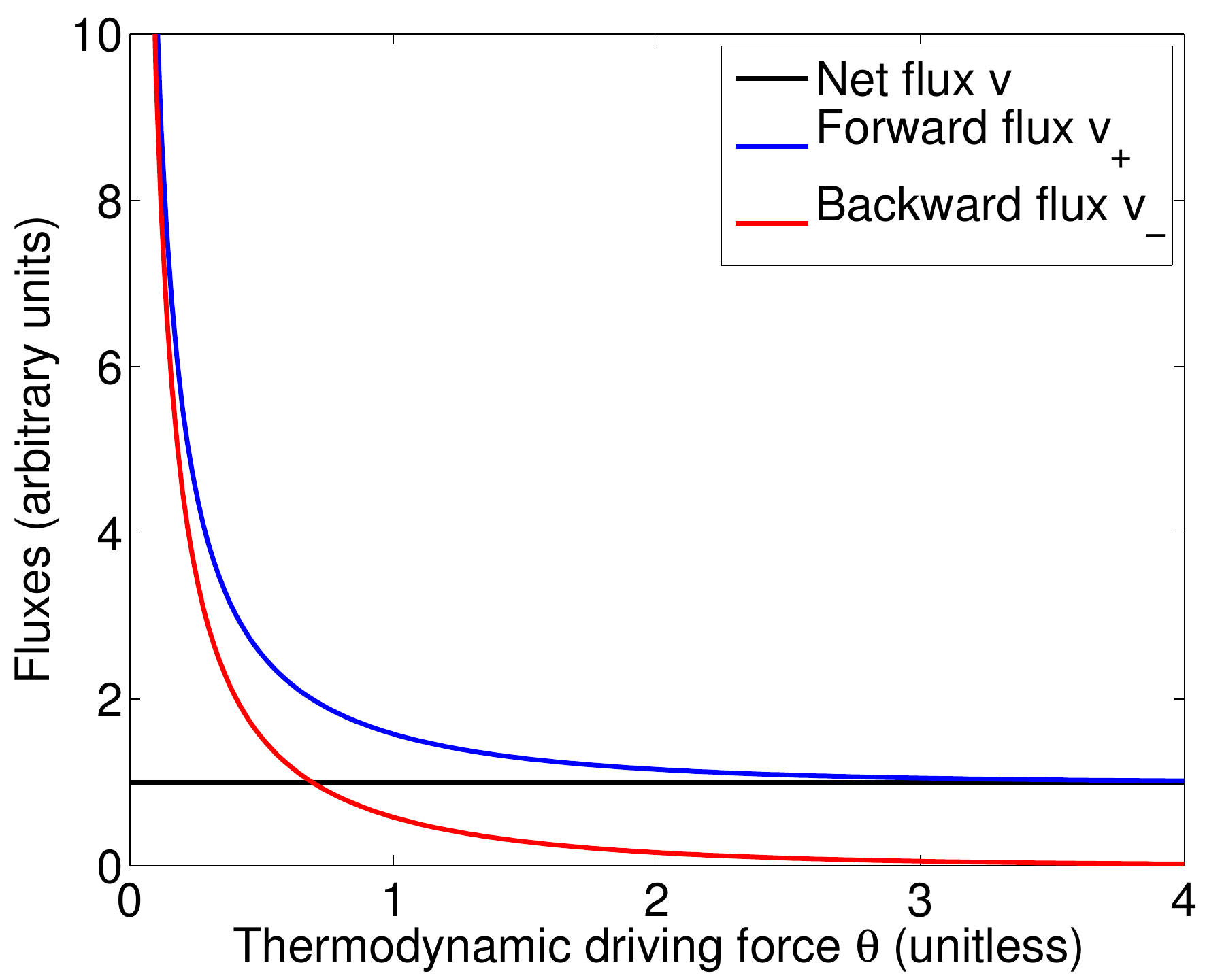} &
      \includegraphics[width=6cm]{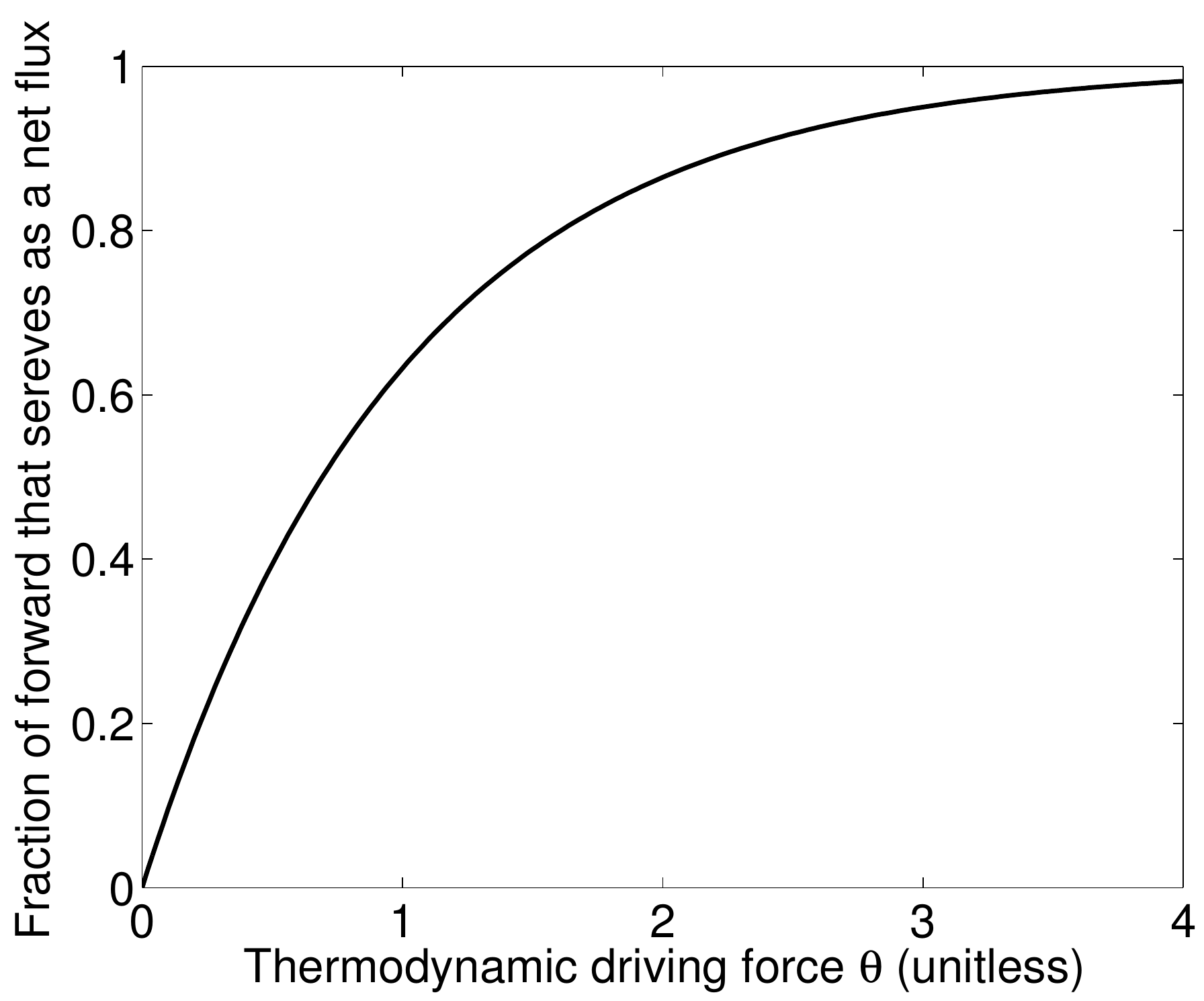} \end{tabular} \caption{Enzyme
      efficiency depends on thermodynamics. The thermodynamic driving
      force $\Theta=-\Delta G'/RT$ in a reaction determines the ratio
      between forward and backward fluxes: $v^+/v^-=\e^\Theta$ .  If
      the net flux $v=1$ is kept fixed, forward and backward fluxes
      strongly increase as $\Theta$ approaches 0 (chemical
      equilibrium). (a) Forward (blue) and backward (red) flux as
      functions of the thermodynamic  force. In each point,
      their difference yields the predefined net flux $v=1$. (b) Only
      a fraction of the forward flux $v^+$ acts as a net flux, while
      the rest is canceled by the backward flux (see Figure
      \ref{fig:efficiencies}). This fraction  varies between 0 (no
      thermodynamic force, chemical equilibrium) and 1 (high
      thermodynamic force, strongly driven
      reaction).}  \label{fig:fluxfractions} \end{center}
\end{figure}

\myparagraph{Simplified rate law denominators} 
The formula for the denominator $D$ in Eq.~(\ref{eq:GeneralRateLawRate})
depends on the enzyme mechanism assumed.  A general, biochemically plausible
choice is polynomials of the form
\begin{eqnarray}
\label{eq:DEMC4}
D(\cv) = 1 + \sum_{k} M_{lk} \prod_i c_{i}^{m_{lik}}
\end{eqnarray}
with positive coefficients $M_{lk}$ and exponents $m_{lik}$.  Each sum
term (index $k$) represents a binding state of the enzyme.  The
exponents $m_{lik}$ encode the numbers of bound reactant molecules and
the prefactors encode the binding energies. The sum term 1 represents
the unbound enzyme.  The highest-order substrate term reads $\prod_{i}
(s_{i}/\kmi)^{m^{\rm S}_{i}}$ and the highest-order product term reads
$\prod_{i} (p_{i}/\kmi)^{m^{\rm P}_{i}}$.  The denominator may also
contain additive or multiplicative terms for allosteric activation and
inhibition.  While the exponents $m_{lik}$ are usually positive
integers, allosteric regulation can imply denominator terms
$\kmS/s$. A special case of Eq.~(\ref{eq:DEMC4}) are rate laws for
polymerization reactions \cite{hogr:13}, which can also be used as
simplified rate laws for biomass-producing reactions. In this case, it
will be the ``template'' molecules rather than the enzyme that is
scored by a cost.  By focusing on simple enzyme mechanisms with few
binding states, we obtain general rate laws that are valid for all
reaction stoichiometries. Their denominators have simple structures
(containing only few sum terms and a few Michaelis-Menten constants as
parameters) \cite{liuk:10}.  Since these rate laws containing fewer
denominator terms than more complex rate laws, the rates become higher
and enzyme demand and costs tend to be underestimated. The
energy-based EMC2 functions are based on rate laws with the
denominators
\begin{eqnarray}
  \label{eq:EMC2ratelaws}
 D^{\rm S} &=& \prod_{i} (s_{i}/\kmi)^{m^{\rm S}_{i}} \nonumber \\
 D^{\rm SP} &=& 
 \prod_{i} (s_{i}/\kmi)^{m^{\rm S}_{i}} + \prod_{j} (p_{j}/\kmj)^{m^{\rm P}_{i}}.
\end{eqnarray}
The mathematical products are called mass-action
terms. In the first formula, we  assume that substrate levels are high and
product levels are low; and in the second one, that both substrate and
product levels are  high. The saturation-based EMC3 functions
are based on rate laws with the 
denominators
\begin{eqnarray}
  \label{eq:EMC3ratelaws}
 D^{\rm 1S} &=& 
1 +  \prod_{i} (s_{i}/\kmi)^{m^{\rm S}_{i}}\nonumber \\
 D^{\rm 1SP} &=& 
1 + \prod_{i} (s_{i}/\kmi)^{m^{\rm S}_{i}}+ \prod_{j} (p_{j}/\kmj)^{m^{\rm P}_{i}}.
\end{eqnarray}
These denominators contain only the term 1 and the substrate and
product mass-action terms. To justify these rate laws, we assume a
strongly cooperative binding between substrates and between products
and consider an enzyme mechanism with only three states: enzyme bound
with all substrates, enzyme bound with all products, and unbound
enzyme.  The first formula assumes low product concentrations, and The
second formula describes the direct-binding modular (DM) rate law
\cite{liuk:10}. The direct-binding modular rate law generalizes the
reversible MM kinetics.  Furthermore, we consider the common modular
rate (CM) law \cite{likl:06a, liuk:10}, a generalized form of
reversible MM kinetics with the denominator
\begin{eqnarray}
  \label{eq:EMC3ratelawsgeneral}
 D^{\rm CM} &=& 
 \prod_{i} (1+s_{i}/\kmi)^{m^{\rm S}_{i}}+ \prod_{j} (1+p_{j}/\kmj)^{m^{\rm P}_{i}} -1.
\end{eqnarray}
In the  enzyme mechanism, substrate molecules bind
independently, product molecules bind independently, and substrate and
product binding exclude each other.  Multiplying out the denominator
(\ref{eq:EMC3ratelawsgeneral}), we obtain many more terms than in the
direct-binding modular rate law. Realistic rate laws will
contain more denominator terms than the DM rate law, but possibly
fewer than the CM rate law. To interpolate between  the two
extremes, we may take
their arithmic or geometric mean
\begin{eqnarray}
D^{\rm geom} = \sqrt{ D^{\rm DM} \, D^{\rm CM}}, \qquad
D^{\rm arith} = \half \,D^{\rm DM} + \half \, D^{\rm CM}.
\end{eqnarray}
If the denominator values $D^{\rm DM}$ and $D^{\rm CM}$ are not too
different, the two mean values will be similar\footnote{If $a\approx
  b$, we can approximate $\sqrt{a\,b} = \sqrt{a\,(a+b-a)} = a \sqrt{1
    + \frac{b-a}{a}} \approx a (1 + \half \frac{b-a}{a}) = \half
  [a+b]$.}. In the second formula (arithmetic mean), the mass-action
terms appear as in DM and CM rate laws, and all other terms from the
CM law appear with prefactors of $\half$. If we define rate laws by
taking a geometric (or arithmetic) mean of rate laws denominators, the
corresponding enzyme costs will be given by geometric (or harmonic)
mean values of enzyme costs.  \myparagraph{Allosteric regulation} If
an enzyme is allosterically regulated, this can be described by
additive or multiplicative regulation terms in the rate law
denominator \cite{liuk:10}.  Additive terms arise from competitive
regulation. Multiplicative terms (for non-competitive regulation) can
be split from the denominator and become prefactors of the rate
law. Typical choices are $\frac{x}{x+k^{\rm A}_{X}}$ for
non-competitive activation and $\frac{k^{\rm I}_{X}}{x+k^{\rm I}_{X}}$
for non-competitive inhibition, with rate constants $k^{\rm A}$ and
$k^{\rm I}$ and a regulator concentration $x$ \cite{liuk:10}.  Thus,
in the factorized EMC formulae, allosteric effects can either be
listed by a separate efficiency factor or be included in the
saturation factor. For instance, the saturation factor for
Michaelis-Menten kinetics with non-competitive inhibition can be split
into
\begin{eqnarray} 
 \eta^{\rm sat} = \frac{s/\kmS}{(1+\frac{x}{K_{\rm I}}) (1+\frac{s}{\kmS}+\frac{p}{\kmP})}  = 
\frac{1}{1+\frac{s}{\kmS}+\frac{p}{\kmP}} \frac{1}{(1+x/K_{\rm I})} = 
\eta^{\rm sat*}\, \eta^{\rm reg}.
\end{eqnarray} 

\subsection{How the efficiency factors are derived} 

The general formula \ref{eq:GeneralRateLawRate} covers a wide range of
possible rate laws.  To demonstrate how it can be factorized into the
capacity and efficiency factors, we consider a bimolecular reaction $A
+ B \rightleftharpoons P + Q$ and an enzyme with a common modular (CM)
rate law (Equation \ref{eq:EMC3ratelawsgeneral}), i.e.
\begin{eqnarray}
\label{eq:convenienceKinetics}
v &=& \enzyme\, \frac{\kcatplus\, \frac{[A][B]}{K_A\,K_B} - \kcatminus\, \frac{[P][Q]}{K_P\,K_Q}}
{(1+\frac{[A]}{K_A})(1+\frac{[B]}{K_B})  + (1+ \frac{[P]}{K_P}((1 + \frac{[Q]}{K_Q}) -1} 
\nonumber \\
&=& \enzyme\, \kcatplus\, \frac{\frac{[A][B]}{K_A\,K_B} - \frac{\kcatminus}{\kcatplus}\frac{[P][Q]}{K_P\,K_Q}}
{ 1 + \frac{[A]}{K_A} + \frac{[B]}{K_B} + \frac{[A][B]}{K_{AB}} + \frac{[P]}{K_P} + \frac{[Q]}{K_Q} + \frac{[P][Q]}{K_{PQ}}}\nonumber \\
&=& \enzyme\, \kcatplus\, 
\underbrace{\left(1 - \e^{-\Theta}\right)}_{\eta^{\rm enr}} 
\underbrace{\frac{ \frac{[A][B]}{K_A\,K_B}} {1 + \frac{[A]}{K_A} + \frac{[B]}{K_B} + \frac{[A][B]}{K_{AB}} + \frac{[P]}{K_P} + \frac{[Q]}{K_Q} + \frac{[P][Q]}{K_{PQ}}}}_{\eta^{\rm sat}}
\end{eqnarray}
where for the last step, we used the Haldane relationship $K_{\rm eq}
= \frac{\kcatplus}{\kcatminus} \frac{K_P\,K_Q}{K_A\,K_B}$ and the
identity $\e^{-\Theta} = \frac{[P][Q]}{[A][B]}/K_{\rm eq}$.  This
Haldane relationship and the connection between the thermodynamic
driving force and the ratio between the numerator terms hold in
general:
\begin{eqnarray}
\frac{k^+_{\rm cat} \prod_{i} (\frac{s_{i}}{\kmi})^{m^{\rm S}_{i}}}
     {k^-_{\rm cat} \prod_{i} (\frac{p_{i}}{\kmi})^{m^{\rm P}_{i}}}
= 
\frac{\prod_{i} s_{i}^{m^{\rm S}_{i}}}{\prod_{i} p_{i}^{m^{\rm P}_{i}}}
\cdot
\frac{k^+_{\rm cat} \prod_{i} \kmi^{m^{\rm P}_{i}}}
     {k^-_{\rm cat} \prod_{i} \kmi^{m^{\rm S}_{i}}}
=
\frac{\prod_{i} s_{i}^{m^{\rm S}_{i}}}{\prod_{i} p_{i}^{m^{\rm P}_{i}}}
/ K_{\rm eq}
=
\e^{-\Theta}.
\end{eqnarray}
Thus we can obtain the general factorized rate law:
\begin{eqnarray}
  \label{eq:GeneralRateLawMulti}
  v =  \enzyme\, \frac{k^+_{\rm cat} \prod_{i} (\frac{s_{i}}{\kmi})^{m^{\rm S}_{i}} - 
                         k^-_{\rm cat} \prod_{i} (\frac{p_{i}}{\kmi})^{m^{\rm P}_{i}}}
                        {D(s_1, s_2, ..,p_1, p_2, ..)}
     = \enzyme\, \kcatplus\, 
       \underbrace{\left(1 - \e^{-\Theta}\right)}_{\eta^{\rm enr}} \,
       \underbrace{\frac{\prod_{i} (s_{i}/\kmi)^{m^{\rm S}_{i}}}
                        {D(s_1, s_2, .., p_1, p_2, ..)}}_{\eta^{\rm sat} \cdot \eta^{\rm reg}}.
\end{eqnarray}

\section{Enzyme cost functions}
\label{sec:SIListOfScores}

To quantify enzyme cost, we assume it is proportional to the
concentration of that enzyme. Potentially, each enzyme level can be
weighted by different enzyme-specific costs.  Are such cost weights
biologically justified? We now discuss the relevance of these costs
and show how the linearity assumption, combined with separable rate
laws, yields simple factorized enzyme cost functions.

\subsection{What factors determine the cost per enzyme molecule?}
\label{sec:SIcostFactors}

In ECM, we assume that cells realize their metabolic fluxes at a
minimal enzyme cost and that this cost is a direct function of the
enzyme levels.  We further assume that the cost function is linear,
i.e.  $\hminus(\enzyme_1, \enzyme_2, ..) = \sum_l {\hminus_{\enzyme_l}
  \enzyme_l}$.  The values of the enzyme-specific costs
$\hminus_{\enzyme_l}$ depend on the biological context. For instance,
cast can be defined by a growth deficit caused by enzyme
over-expression.  In microbes, such cost values can be measured using
standard lab techniques for measuring growth rate. A disadvantage of
this approach is the difficulty to disentangle the cost of the
specific over-expressed enzyme from other effects that the enzyme
could have on the metabolic network at large, most importantly the
potential benefit of increasing the flux in the reaction it catalyzes.
Theoretical approaches, on the other hand, can be used to isolate
enzyme cost from global effects, but may not capture the many possible
ways in which growth deficits are caused in reality.  Aside from the
resources required for production and maintenance, enzymes need to
compete with other proteins and macromolecules for the limited space
in the cytoplasm or on membranes \cite{zhvm:11, digs:11}. The
restriction can be related to the volume of the protein, the occupied
membrane surface area, and the effect it has on the osmotic pressure
(which depends on electro-static interactions with the surrounding
water). Therefore, protein cost is a complex function of the
copy-number of the enzyme, its physico-chemical parameters (such as
molecular weight, 3D structure, hydrophobicity, charge, etc.), and its
production or degradation rate. Furthermore, enzymes can have adverse
side effects, e.g., by promiscuous activity \cite{eako:12}, which are
virtually impossible to predict without extensive knowledge about an
organism's full metabolic network and physiology.  Since many of these
features are unknown for most enzymes, and some of these effects
require elaborate 3D models, which are beyond the scope of this work,
we try to define a cost function that is simple to calculate, but
captures many of these biological aspects.  To determine relative
cost weights, we may simply assume that enzyme cost is proportional to
enzyme mass. Since we only use total cost as an optimization goal,
the problem is scale-free and therefore the relative cost weights are
enough for ECM.  We obtain a linear cost function with specific
costs $\hel \sim L_l$, where $L_l$ is a measure of protein size (e.g.,
length in amino acid units or mass in Daltons). This may be relevant,
in particular, for protein complexes or for lumped reactions
representing entire pathways.  This formula for protein cost weights
can be extended by other factors:
\begin{itemize}
\item \textbf{Degradation rate and protein lifetime} To account for
  differences in protein degradation, we can assume that enzyme cost
  is proportional to the enzyme production rate (in units of amino
  acids or Daltons per second). We define the lifetime of the enzyme
  as $\tau_l = (\kappa_l+\lambda)\inv$ (where $\lambda$ is the growth
  rate $\kappa_l$ is the degradation rate\footnote{Protein lifetimes
    may systematically between types of reactions catalyzed.  Enzymes
    catalyzing oxidation reactions are likely to accumulate damages
    faster and can be expected to have shorter lifetimes.}), and
  therefore the cost would be $\hel \sim L_l/\tau_l$.  When the cell
  growth rate is much faster than the degradation rate, $\kappa_l \ll
  \lambda$, all enzymes have approximately the same lifetime and
  therefore the effect of protein degradation would be negligible.
\item \textbf{Individual amino acid costs} Enzymes show different
amino acid
  compositions, and different amino acids
  require different amounts of energy for their production. If enzyme
  cost is mainly due to investments in amino acid production, we can
  quantify the energetic and material costs of individual amino acids
  \cite{akgo:02} and account for them in our choice of enzyme cost
  weights. We did this in our calculations, but other cost functions,
  in which amino acid composition is neglected, lead to similar
  predictions of enzyme levels.
\item \textbf{Enzyme complexes with multiple subunits and catalytic
  sites} An enzyme may consist of several protein subunits and may
  contain several catalytic sites. Therefore, we adopt the
  convention that $\kcat$ values refer to catalytic sites, while
  protein levels refer to protein subunits. The number $N_{\rm sub}$
  of complex subunits and the number $N_{\rm cat}$ of catalytic sites
  per complex must appear in the formulae for reaction rates, enzyme
  demand, and enzyme cost: we replace in all these formulae the
  $\kcatplus$ value (referring to a single catalytic site) by an
  effective value ${\kcatplus}'=\frac{N_{\rm cat}}{N_{\rm sub}}\,\kcatplus$
  (referring to one subunit).
\item \textbf{Covalent modification} Enzyme activity can  be changed
  by phosphorylation or other posttranslational modifications.  So
  far, we assumed that enzymes exist in one form and that the enzyme
  level $\el$ represents their concentration. For modifiable enzymes,
  our variable $\el$ describes the concentration of enzyme molecules
  \emph{in the right modification state}, which is only a fraction
  $\rho_l<1$ of the total concentration. Since, the total enzyme
  concentration is $1/\rho_l$ times as large as the concentration
  $\el$ appearing in the rate law, the enzyme cost weight $\hel$
  must be increased by this factor $1/\rho_l$.
\item \textbf{Constrained enzyme levels} The enzyme amounts in cells
  are restricted by physical constraints (e.g.~space restrictions on
  mitochondrial membranes, which limit the number of respiration
  complexes). In ECM, this could be described by imposing upper
  limits on sums of enzyme levels in the cell, in cell compartments,
  or within membranes. As a heuristics, such constraints can also be
  replaced by cost terms that penalize high levels of these
  enzymes\footnote{To justify such cost terms mathematically, one
    could first consider a model with constraints on some enzyme
    fractions. These constraints could be treated by Lagrange
    multipliers, which lead to effective cost terms in the objective
    function. In these terms, the Lagrange multipliers appear as if
    they were enzyme cost weights.  Replacing these Lagrange
    multipliers by constant numbers, we obtain effective linear cost
    terms which we can add to our cost function.}.
\item \textbf{Lumped reactions} In a model, series of reactions can be
  represented by lumped reactions.  Effective parameters ($h$ and
  $\kcatplus$ values) for lumped reactions can be obtained as
  described in \ref{sec:SIlumped}.
\item \textbf{Absolute scaling} Beyond ECM, some applications require
  an absolute scaling of the cost function, i.e. the cost must be in
  units that are comparable to other factors that affect fitness, such
  as the biomass flux or the growth rate.  This absolute scaling can
  be determined based on experimental data, for instance, by matching
  measured growth deficits for GFP \cite{szad:10}.  Alternatively, we
  can convert the other fitness terms to units of enzyme mass, e.g. by
  quantifying how the biomass flux generates the amino acids that are
  eventually used to synthesize the enzymes.
\item \textbf{Convex non-linear cost functions} Finally, if a
  nonlinear cost function $\hminus(\enzymev)$ is used, the total cost
  of a pathway is not simply a sum over the reactions' enzyme
  costs. Instead, a high cost in one enzyme could increase the cost
  pressure on other enzymes. Nevertheless, if the cost functions is
  convex, the total cost remains a convex function on the metabolite
  polytope, so numerical optimization stays feasible.
\end{itemize}

\subsection{Enzyme cost as a function of metabolite levels}
\label{sec:SIfscScores}

The enzyme cost of a given metabolic flux profile can be cast as a
function $\enzymemetcost(\ln \cv)$ on the metabolite polytope.  To
obtain simple cost functions, we consider the factorized enzyme cost
Eq.~(\ref{eq:TotalEnzymeDemand}) and approximate some of the terms by
constant numbers. Constant values of 1 arise from limiting cases: an
infinite driving force leads to an energy factor of 1, and if enzymes
are fully substrate-saturated and product concentrations are small ($a
\gg k^{\rm M}_{a}$, $b \ll {\km}_{b}$), the saturation factor can be
set to 1. To approximate the true cost function, we can start from the
most simple estimate (EMC1) and subsequently reintroduce the different
efficiency factors.  The enzyme
cost functions can be grouped, according to the data required, into
five levels (see Tables \ref{tab:EnzymeCostScores} and
\ref{tab:formulae}):

\begin{table}[t!]
\parbox{10cm}{\begin{tabular}{lll}
    \cellcolor{brown} \textbf{Function} & \cellcolor{brown} \textbf{Denominator}
    & \cellcolor{brown} \textbf{Rate law $\ratelaw(\cv)$} \\ \hline\\[-2mm]
    EMC0         & 1 & $\const$ \\
    EMC1         & 1 & $\kcatplus$\\
    EMC2s      & $D^{\rm S} = S$    &  $\kcatplus \eta^{\rm enr}$\\
    EMC2sp     & $D^{\rm SP} = S+P$ &  $\kcatplus\frac{\eta^{\rm enr}}{1+\keq\,\e^{- \Theta}}$\\
    EMC3s     & $D^{\rm 1S}=1+S$   &  $\kcatplus\,\frac{\eta^{\rm enr}\,S}{1 + S}$\\
    EMC3sp    & $D^{\rm 1SP}=1+S+P$ &  $\kcatplus\frac{S\,\eta^{\rm enr}}{1 + S + P}$\\
    EMC4cm     & $D^{\rm CM}=S^{\rm CM}+P^{\rm CM}-1$ &  $\kcatplus\frac{S\,\eta^{\rm enr}}{S^++P^+-1}$\\
    EMC4geom   & $D^{\rm geom}$      &  $\kcatplus\frac{S\,\eta^{\rm enr}}{\sqrt{D^{\rm CM}\,D^{\rm 1SP}}}$\\
    EMC4arith  & $D^{\rm arith}$     &  $\kcatplus\frac{S\,\eta^{\rm enr}}{\half(D^{\rm CM} + D^{\rm 1SP})}$\\
    EMC4         & $D(\cv)$         &  $\kcatplus\frac{S\,\eta^{\rm enr}}{D(\cv)}$
  \end{tabular}}
\parbox{6cm}{\begin{tabular}{ll}
    \cellcolor{brown} \textbf{Quantity} & \cellcolor{brown} \textbf{Formula} \\\hline
    \\[2mm] 
    Reaction rate & $v = \enzyme\, \ratelaw$\\[2mm]
    Enzyme demand & $\enzyme = v/\ratelaw$\\[2mm]
    Enzyme cost & $\enzymemetcost = h\,\enzyme =  \frac{h\, v}{\ratelaw}$ \\[2mm]
    Flux-specific cost & $\rrs = \frac{\enzymemetcost}{v} = \frac{h}{\ratelaw}$
  \end{tabular}\ \\[2.2cm]}
  \caption{Rate laws and enzyme-based metabolic cost (EMC) 
functions. First
    table: EMC functions derived from simplified rate
    laws. Abbreviations: Forward catalytic constant
    $\kcatplus$. Energy efficiency $\eta^{\rm enr} =
    1-\e^{-\Theta(\cv)}$. Mass-action denominator terms $S= \prod_i
    (s_i/\kmi)^{m^{\rm S}_i}$, $P= \prod_j (p_j/\kmj)^{m^{\rm P}_j}$;
    convenience denominator terms $S^{\rm CM}= \prod_i
    (1+s_i/\kmi)^{m^{\rm S}_i}$; $P^{\rm CM}= \prod_j
    (1+p_j/\kmj)^{m^{\rm P}_j}$. The denominator $D(\cv)$ in the EMC4
    function is a polynomial with non-negative coefficients as in
    Eq.~(\ref{eq:DEMC4}); it can also contain terms describing
    allosteric regulation. The molecularities $m^{\rm s}$ and $m^{\rm
      p}$ represent stoichiometric coefficients, but they can contain
    reaction-specific Hill coefficients as prefactors.  The second
    table lists some quantities derived from the rate laws.  }
  \label{tab:formulae}
\end{table}

\begin{itemize}
\item \textbf{EMC0 (``sum of fluxes'')} If no enzyme parameters are known
  at all, we can assume the same flux-specific cost $\rrs$ for all
  enzymes. Enzyme levels and enzyme costs are proportional to fluxes
  across the network: $\el \sim \enzymemetcostl \sim v_l$, and enzyme
  cost is proportional to the sum of fluxes (where fluxes are positive
  due to our convention about reaction orientations).
\item \textbf{EMC1 (``capacity-based'')} In the capacity-based (EMC1)
  functions, enzymes have individual specific flux costs ${\rrl} =
  \hel/\kcatplusl$ (based on known $\kcatplus$ and $h$ values) and are
  independent of metabolite levels. This is equivalent to replacing
  reaction rates $v$ by $v_{\rm max}$ values, or dropping the
  efficiency factors in
  Eq.~(\ref{eq:TotalEnzymeDemand}). Alternatively, we can set each
  factor to a constant, enzyme-specific value.
\item \textbf{EMC2 (``energy-based'')} The energy-based (EMC2)
  functions capture the fact that cost increases close to
  equilibrium. They depend on metabolite levels, but only via the
  driving forces only, and equilibrium constants need to be known for the
  calculation. In the EMC2s function, we assume that enzymes are
  strongly substrate-saturated while product saturation is negligible:
  the denominator $D^{\rm S}$ (see Eq.~(\ref{eq:EMC2ratelaws}))
  cancels the numerator term, resulting in a constant saturation
  factor $\eta^{\rm sat}=1$ (i.e., full substrate saturation).  The
  EMC2sp function, another energy-based function, describes enzymes
  with strong substrate and product saturation (denominator $D^{\rm
    SP}$).  Here the saturation factor
\begin{eqnarray} \eta^{\rm sat}
  &=& \frac{\prod_{i} s_{i}/{\kmi}} {\prod_{i} s_{i}/{\kmi}+ \prod_{j}
    p_{j}/{\kmj}} = \frac{1} {1+ \frac{\prod_{j}
      p_{j}/{\kmj}}{\prod_{i} s_{i}/{\kmi}}} = \frac{1} {1+
    \keq\,\e^{-\Theta}} \end{eqnarray} is not a constant, but it
depends on the driving force.  Since the rate can be computed from
driving forces alone, the EMC2sp function is claissified as
``energy-based''.  \item \textbf{EMC3 (``saturation-based'')} The
saturation-based (EMC3) functions represent rate laws with the
denominators $D^{(1S)}$ and $D^{(1SP)}$.  These are rate laws that do
not depend on metabolite levels, but on their mathematical products,
the mass-action terms. To compute them, the $\km$ values (more
precisely, $\km$ values multiplied over all substrates or products)
must be known.  The EMC3sp function follows from the direct-binding
rate law or, for unimolecular reactions, from reversible
Michaelis-Mention kinetics.  The EMC3s function has a similar form,
but contains no product term.  It describes enzymes with incomplete
substrate saturation, but far from equilibrium ($\Theta \rightarrow
\infty$), so the product term can be neglected.  The energy factor
$\eta^{\rm enr}$ can be set to 1, but the factor
\begin{eqnarray}
\eta^{\rm sat} &=& \underbrace{\frac{\prod_i (s_i/\kmi)^{m^{\rm
        S}_i}}{1+\prod_i (s_i/\kmi)^{m^{\rm S}_i}}}_{D^{\rm 1S}}
\end{eqnarray}
remains an explicit term in the rate law.

\item \textbf{EMC4 (``Kinetics-based'')} The kinetics-based cost
  functions (EMC4) capture all thermodynamically feasible rate laws,
  including rate laws with allosteric activation or inhibition terms.
  Their denominators have the form Eq.~(\ref{eq:DEMC4}) and contain
  the terms from the $D^{\rm 1SP}$ denominator, plus others.  Examples
  are rate laws with $D^{\rm CM}$, $D^{\rm geom}$, and $D^{\rm arith}$
  denominators.
\end{itemize}

Here are some additional remarks.
\begin{itemize}
\item To obtain simplified EMC functions, we can apply the following
  simplifications: (i) neglect individual enzyme cost weights $h$;
  (ii) neglect individual catalytic constants $\kcatplus$; (ii) set
  $\eta^{\rm enr}$ to a constant value; (iii) set $\eta^{\rm sat}$ to a
  constant value. (iv) If $\eta^{\rm sat}$ is not set constant, (iv.a)
  use/do not use the term 1 in denominator; (iv.b) use/do not use
  highest-order product term in denominator; (iv.c) with more
  reactants: use/do not use additional terms; (v) if the enzyme is
  allosterically regulated: possibly, set regulation term
  constant. These simplifications can be freely combined. Whether a
  cost function is classified as EMC0, EMC1, EMC2, EMC3, or EMC4
  depends can be determined from its formula.
\item Different EMC functions require different types of input data:
  $\kcatplus$ values for EMC1; additionally equilibrium constants (or
  standard reaction Gibbs energies) for EMC2; additionally, $\km$
  values for EMC3; and possibly, more parameters for EMC4.
\item If efficiency factors are set to 1 (and not to smaller constant
  values), each EMC function is a lower estimate of the following (less
  simplified) ones. This includes EMC0 function if we use the
  \emph{largest} $\rrlmin$ value from the other functions as a prefactor
  in the EMC0 function.
\item If a rate law contains Hill coefficients, they can be treated as
  part of the molecularities. The reactants of a reaction must have
  the same Hill coefficient.
\item Given two possible rate laws for the same reaction, we may
  define new rate laws by taking their geometric or harmonic mean
  (this is, for instance, how the rate law denominators $D^{\rm
    (geom)}$ or $D^{\rm arith}$ were defined). In this case, the
  enzyme demands and costs (for the new rate law) are given by the
  geometric or arithmetic means from the original rate laws. In
  particular, the enzyme cost functions related to the denominators
  $D^{\rm geom}$ or $D^{\rm arith}$ (called EMC4geom and
  EMC4arith) represent the geometric (or arithmetic) mean of the
  original EMC3sp and EMC4cm functions.
\end{itemize}

\subsection{Flux-specific enzyme cost and pathway-specific activity}

\myparagraph{\ \\Flux-specific enzyme cost for enzymes and pathways}
\label{sec:SIpathwayCost}
The flux-specific cost $\rrl$ of an enzyme, in the context of a
certain metabolic state, is defined as the enzyme cost per unit flux.
At given metabolite levels, and assuming a linear cost function, the
flux-specific cost is a constant. At constant metabolite levels, a
doubling of the flux will require a doubling of the enzyme level, and
thus a doubling of the enzyme cost. The ratio of cost and flux remains
constant and is given by Eq.~(\ref{eq:TotalEnzymeDemand}).  A
flux-specific cost can also be defined for pathways or any sets of
reactions.  Since different reactions may carry different fluxes a
pathway (due to non-stationarity, side branches, or splitting of
molecules as between upper and lower glycolysis), we choose one flux
or production rate as the representative \emph{pathway flux} $\vPW$
and define the the \emph{pathway specific cost} by $\rrpw =
\frac{\enzymemetcost}{\vPW}$, i.e., the pathway enzyme cost divided by
the pathway flux.  In practice, the pathway flux should represent a
flux that matters for the cell's benefit (e.g., ATP production in a
glycolysis model). To compare different pathway models at identical
benefits, we could scale their fluxes to the same benefit value.
Given fixed metablite levels at the pathway boundaries\footnote{In a
  kinetic model with constant external metabolite levels, an overall
  scaling of enzyme levels will lead to a proportional scaling of
  fluxes.  With a linear enzyme cost function $\hminus(\enzyme_1,
  \enzyme_2, ..)$), this scaling will leave all flux-specific costs
  unchanged.  In reality, a change in enzyme levels is likely to
  affect metabolite levels outside the pathway, so the theoretical
  result does not exactly apply.} the flux-specific cost of a pathway
will be constant.  If all reactions in a pathway carry identical
fluxes, it is given by the sum of the reactions' flux-specific costs.
Otherwise, the pathway specific cost will be a weighted sum $\sum_{l }
\rrl' \,\vscaled_{l}$ of the reaction flux-specific costs, with
unitless relative fluxes $\vscaled_{l} = \frac{v_{l}}{v_{{\mathcal
      L}}}$ as weights (in a simple linear chain, $\vscaled_{l}=1$).

\myparagraph{Pathway specific activities}
\label{sec:SIpsa}
An enzyme's specific activity is given by the catalyzed flux divided
by the enzyme mass (in $\mu$mol/min/mg enzyme).  Specific activities
can also be defined for entire pathways \cite{bnlm:10}.  If we treat
enzyme mass (in grams/cell volume) as the cost function
$\hminus(\enzymev)$, the resulting flux-specific cost $\rrs$ (enzyme
cost per flux) is exactly the inverse of the specific activity (flux
per enzyme mass)\footnote{In this case, an  enzyme's cost weight $\hel$
will be given by the enzyme's total mass (in grams/cell volume),
divided by the concentration (number of enzyme molecules divided by
Avogadro constant and cell volume), so $\hel$ is just the enzyme's
molecular mass in Daltons.}  . This holds both for single reactions
and entire pathways.  With enzyme mass used as a cost function, a
pathway's specific cost $\rrspw = \enzymemetcost/\vPW$ yields
the amount of enzyme (in grams/cell volume) divided by the pathway
flux (in mM/s) or, in other words, the amount of enzyme (in grams)
divided by the pathway flux (in mol/s). Accordingly, the pathway
specific activity (in (mol/s)/grams enzyme) is given by
\begin{eqnarray}
\activity_{\rm pw} = \frac{\vPW}{\enzymemetcost} = \rrs\inv.
\end{eqnarray}
To express this in units of $\mu$mol/min/mg enzyme, we multiply by
$60000$.  Since the pathway specific cost is a weighted sum of
enzyme specific costs
\begin{eqnarray}
 \rrspw &=&  \sum_l \vscaled_l\, \rrl,
\end{eqnarray}
the pathway specific activity $\activity_{\rm pw}$ (referring to the
pathway flux $v_{\rm pw}$) is the weighted harmonic sum of the enzyme
specific activities $\activity_l$
\begin{eqnarray}
 \activity_{\rm pw} &=& \left[ \sum_l \vscaled_l\,  \activity_l\inv \right]\inv
\end{eqnarray}
where the $\vscaled_l=v_l/v_{\rm pw}$ are scaled (unit-less) fluxes.
This formula agrees with the formula given in \cite{bnlm:10} and
allows for non-uniform fluxes along the pathway.

\section{Enzyme cost minimization}

\subsection{Parameterizing the metabolic states of a kinetic model}
\label{sec:SIproofUniqueMapping}

The standard practice in kinetic modeling is to set up an ODE system
where enzyme levels $\enzymev$ are given (typically, due to separation
of time scales they are assumed to be fixed) and  metabolite levels
$\cv(t)$ are the free variables which evolve over time.  The
kinetic model describes the relationship between enzymes and
metabolites (via kinetic rate laws), which in turn affect the
metabolites.
\begin{eqnarray}
\dot{\cv} &=& \Nmat \vv \nonumber \\
\vv &=& \vv(\enzymev, \cv).
\end{eqnarray}
In order to find a steady state, the ODE is integrated over time
until $||\dot{\cv}||$ is small enough to be labeled as stationary.
Then we can say the system is in steady state and determine the flux
and metabolic state $(\cv(\infty), \vv(\infty))$.  In many cases, the
steady state will depend on the choice of initial conditions
$\cv(0)$. We thus define the set of all steady states as ${\mathcal S}
= \{(\cv(\infty), \vv(\infty), \enzymev)\}_{\enzymev, \cv(0)}$.  Using
this representation, determining ${\mathcal S}$ requires an exhaustive
scan of all parameters $\enzymev, \cv(0)$, which can be  time-consuming, 
and virtually impossible for large kinetic networks.
Here, we suggest an alternative representation of steady states which
is computationally simple, and is especially useful for certain types
of optimization problems.  Instead of enzyme levels as parameters, we
use the steady-state fluxes.  Then, for each given steady state,
metabolite levels $\cv$ we can derive the enzyme levels, using the
inverted kinetic rate laws discussed in the previous sections
(essentially, the value of $\enzyme$ in the EMC function).  Therefore,
we can redefine the set of steady states as ${\mathcal S} = \{(\cv,
\vv, \enzymev(\vv, \cv))\}_{\vv, \cv}$ -- where $\cv$ and $\vv$
correspond to the steady state values, like $\cv(\infty)$ and
$\vv(\infty)$ in the previous definition.  This representation of
steady states has a number of practical advantages.

\begin{proposition}
\label{setequality} \textbf{Set of metabolic states} 
Consider a kinetic model with rate laws $v_{l}=
\el\,\ratelaw_{l}(\cv)$, thermodynamically consistent rate constants
(see SI \ref{sec:SIratelaws}), a feasible flux profile $\vv$,
and bounds on metabolite levels. For any feasible 
metabolite profile $\ln \cv \in \metabolitepolytope$ there is a unique set
of enzyme levels $\el$ which realizes $\cv$.
The function $\el(\ln \cv) = v_{l}/\ratelaw_{l}(\cv)$ is differentiable
on the metabolite polytope.
\end{proposition}

\textbf{Proof:} If a metabolite profile $\cv$ is feasible for our
 flux profile $\vv$, the catalytic rates $\ratelaw_{l}(\cv)$ obtained
from the rate laws Eq.~(\ref{eq:GeneralRateLawRate}) must have the
same signs as $v_{l}$, so $\el = v_{l}/\ratelaw_{l}(\cv)$ is positive
on the entire metabolite polytope.  In particular, we know that
$\ln \cv \in \metabolitepolytope \rightarrow \ratelaw_l(\cv) \neq 0$.
Since $\ratelaw_{l}(\cv)$ is differentiable and has a constant sign on
the metabolite polytope, $\el(\ln \cv)$ is differentiable on the
metabolite polytope.

Here are some additional remarks.
\begin{itemize}
\item \textbf{Metabolite profiles parameterize the possible states} 
  Proposition \ref{setequality} guarantees that all thermodynamically
  feasible metabolite profiles can be realized by steady states of the
  kinetic model. In other words, the set ${\mathcal S}$ of metabolic
  states for a given flux profile $\vv$ can be parameterized by the
  points of the metabolite polytope.  This means that the set of
  kinetically realizable metabolite profiles in a kinetic model
  depends on the equilibrium constants, but not on other
  enzyme-specific parameters.
\item \textbf{An enzyme profile need not uniquely determine the
  metabolite profile} In ECM, the same enzyme profile may be
  realizable by different metabolite profiles; this happens, in
  particular, if simplified rate laws are used. (i) If a metabolite
  appears in a model but has no impact on any reaction, its
  concentration can be freely varied, independently of fluxes or
  enzyme levels. (ii) With the EMC0 and EMC1 functions, enzyme levels do
  not depend on metabolite levels. (iii) With EMC2 functions, $\sv = \ln
  \cv$ can be varied along directions in the nullspace of
  $\Ntot\trans$ without affecting the driving forces or enzyme
  cost. These EMC2 functions, on the metabolite polytope, have an
  invariant subspace (namely the nullspace of $\Ntot\trans$).  Under
  what conditions EMC3 and EMC4 functions (without regularization terms)
  have unique optima remains an open question.
\end{itemize}

\subsection{Enzyme-based metabolic cost functions are convex on the metabolite polytope}
\label{sec:SIConvexCost}

The enzyme cost functions
  Eq.~(\ref{eq:TotalEnzymeDemand}) are convex  on
the metabolite polytope: the cost for a metabolite
log-concentration vector, interpolated between two vectors $\sv_{\rm
a}$ and $\sv_{\rm b}$, cannot be higher than the interpolated cost:
\begin{eqnarray}
\label{eq:ExplainConvexity}
  \forall \lambda \in [0, 1]: \enzymemetcost(\lambda \,\sv_{\rm a} +
  (1-\lambda)\,\sv_{\rm b}) \leq \lambda \, \enzymemetcost(\sv_{\rm a}) +
  (1-\lambda)\,\enzymemetcost(\sv_{\rm b}).
\end{eqnarray}
To show that all enzyme-based metabolic cost functions are convex, we
consider the most general rate law with denominator (\ref{eq:DEMC4}),
written in factorized form
\begin{eqnarray}
v &=&  \enzyme \cdot  \kcatplus \cdot  \eta^{\rm enr} \cdot \eta^{\rm sat},
\end{eqnarray}
where
\begin{eqnarray}
\label{eq:SIConvexEtaDefinitions}
\eta^{\rm enr} &=& 1 - \e^{-\Theta} = 1 - 
\exp \left(\frac{1}{RT}\,\Delta_{\rm r} {G^\circ}' + \sum_i n_i \ln c_i\right) \nonumber \\
\eta^{\rm sat} &=& 
\prod_i \left(\frac{s_i}{\kmi}\right)^{-m^{\rm S}_i}
\left(\sum_{k} M_{k} \prod_j c_{i}^{m_{ik}}\right)\inv
=\left(\sum_{k} \alpha_{k} \prod_j c_{i}^{a_{ik}}\right)\inv
\end{eqnarray}
 with coefficients $\alpha_k \in \Rset_+$ and $a_{ik}\in \Rset$.  The
 regulation factor $\eta^{\rm reg}$ need not be explicitly considered because it
 can  be included in the term $\eta^{\rm sat}$. With this rate law, the
 enzyme cost  for a pathway reads 
\begin{eqnarray}
\enzymemetcost_{\rm pw} = \sum_{l} \enzymemetcostl = \sum_{l} \frac{\hel\,v_{l}}{\kcatplusl} 
\cdot \frac{1}{\eta^{\rm enr}_{l}} \cdot\frac{1}{\eta^{\rm sat}_{l}}.
\end{eqnarray}
This function is convex on the
metabolite polytope.  For the proof, the cost function
$\hminus(\enzymev)$ need not be linear; if it is nonlinear, it must be
convex. For the proof, we start with some general lemmas.

\begin{lemma}\label{lemma:exp_cvx} The function
$f(y) = -\ln(1-\e^y)$ is convex in the range $y < 0$.
\end{lemma}
\begin{proof} The second derivative 
\[\frac{\mbox{d}}{\mbox{d}y^2} f(y) = \frac{\e^y}{(1-\e^y)^2}\]  is positive for $y < 0$.
\end{proof}

\begin{lemma}\label{lemma:lse_cvx} The function 
$f(\sv) = \ln \sum_{k=1}^{n} \e^{s_k}$ is convex.
\end{lemma}
\begin{proof}

\[\nabla^2 f(\sv) = \frac{\diag(\cv) ({\bf 1}\trans \cv) - \cv\,\cv\trans}{({\bf 1}\trans \cv)^2} ~~~~~~~~ (\text{where } c_i = e^{s_i})\]

\[\forall \uv:~~ \uv \trans \nabla^2 f(\sv) \uv = 
  \frac{(\sum_i c_i u_i^2)(\sum_i c_i)- (\sum_i u_i c_i)^2}{(\sum_i c_i)^2} \ge 0
\]
since $(\sum_i u_i\, c_i)^2 \le (\sum_i c_i\, u_i^2)(\sum_i c_i)$ from
the Cauchy-Schwarz inequality. Therefore, the Hessian $\nabla^2 f(\sv)$ is
positive semi-definite, which proves that  $f(\sv)$ is convex.

\end{proof}

\begin{lemma}\label{lemma:th_cvx}
For any number  $\nu \in \mathbb{R}_+$ and vector $\nv \in \mathbb{R}^m$, the
function $- \ln (1 - \nu\, \e^{\nv\cdot \sv} )$ is convex over $\{\sv \in
\mathbb{R}^m~|~\nu\, \e^{\nv \cdot \sv} < 1\}.$
\end{lemma}
\begin{proof}
 This function is a composition of $f = -\ln(1-\e^y)$ from Lemma
 \ref{lemma:exp_cvx} with the affine transformation $y = \nv \cdot \sv
 + \ln{\nu}$, an operation which preserves convexity.
\end{proof}

\begin{lemma}\label{lemma:kin_cvx}
For any matrix ${\bf A} \in \mathbb{R}^{n \times m}$ and vectors ${\bf
  b} \in \mathbb{R}^{n}_+$, the following function is convex over
${\bf x} \in \mathbb{R}^{m}$:
\begin{eqnarray}
\ln \left( \sum_{k=1}^{n} \e^{\av_{k}\cdot\sv + b_k} \right)
\end{eqnarray}
where $\av_{i}$ is the $i$th row of $\bf A$.
\end{lemma}

\begin{proof}
 This function is a composition of $f = \ln \sum_{i=1}^{n} \e^{s_i}$
 from Lemma \ref{lemma:lse_cvx} with the affine transformation $s_i =
 \av_{i}\cdot\sv + b_i$, an operation which preserves convexity.
\end{proof}

Based on these lemmas, we can now  prove the convexity of enzyme cost functions.

\begin{lemma}
Assume that all enzyme-catalysed reactions in a model behave according
to rate laws of the type
\begin{eqnarray}
v &=&  \enzyme \cdot  \kcatplus \cdot  \eta^{\rm enr} \cdot \eta^{\rm sat},
\end{eqnarray}
with $\eta^{\rm enr}$ and $\eta^{\rm sat}$ given by 
Eq.~(\ref{eq:SIConvexEtaDefinitions}), with coefficients $\alpha_k \in
\Rset_+$ and $a_{ik}\in \Rset$. Assume that the enzyme cost function
for enzymatic reaction $l$ reads
\begin{eqnarray}
\enzymemetcostl = \frac{\hel\,v_{l}}{\el} = \frac{\hel\,v_{l}}{\kcatplusl} 
\cdot \frac{1}{\eta^{\rm enr}_{l}}
\cdot\frac{1}{\eta^{\rm sat}_{l}}\,.
\end{eqnarray}

Then the total enzyme cost $\enzymemetcost = \sum_{l} \enzymemetcostl$, as a function
of logarithmic metabolite concentrations ($\sv = \ln \cv$), is convex.
\end{lemma}

\begin{proof}
To simplify the efficiency factors, we can use the abbreviations $s_{i} \equiv
\ln c_{i}$, $\nu \equiv \exp (\Delta_{\rm r} {G^\circ}'/RT)$, and $b_k = \ln
\alpha_k$:
\begin{eqnarray}
\eta^{\rm enr} &=& 1 - \nu\,\e^{-\nv\cdot \sv} \nonumber \\
\eta^{\rm sat} &=& \left(\sum_{k=1}^{n} \e^{\av_{k}\cdot\sv + b_k}\right)\inv.
\end{eqnarray}

If we look at the natural logarithm of $\enzymemetcostl$,
\begin{eqnarray}
\ln \enzymemetcostl = \ln \left( \frac{\hel\,v_{l}}{\kcatplusl}  \right) -\ln \eta^{\rm enr}_{l} - \ln \eta^{\rm sat}_{l},
\end{eqnarray}
we see that each of the three terms in the sum is convex in $\sv$. The
first term is constant with respect to the metabolite concentrations
and therefore trivially convex.  The energetic term, $-\ln \eta^{\rm
  th} = - \ln (1 - \nu\,\e^{-\nv\cdot \sv})$, is convex according to
Lemma \ref{lemma:th_cvx}. The saturation factor, $-\ln \eta^{\rm sat} = \ln
\left(\sum_{k=1}^{n} \e^{\av_{k}\cdot\sv + b_k}\right)$, is convex
according to Lemma \ref{lemma:kin_cvx}. We conclude that $\enzymemetcostl$ is
convex too, since it is a composition of a convex function ($\e^x$)
with another convex function ($\ln \enzymemetcostl$).  Finally, the total enzyme
cost ($\enzymemetcost$) is convex since it is a sum of convex functions:
\begin{eqnarray}
\enzymemetcost = \sum_{l} \enzymemetcostl(\sv).
\end{eqnarray}
\end{proof}

\subsection{The ECM problem remains convex under metabolite and enzyme constraints}
\label{sec:SIproblemremainsconvex}

In ECM, we may introduce an upper bound on the sum of all
(non-logarithmic) metabolite levels as an extra constraint.  Unlike
our original metabolite polytope, the resulting admissible region will
have a curved surface.  However, since the sum of metabolite levels is
convex on the metabolite polytope, the new constraint leads to a
convex region, and the ECM problem remains convex. Another possible
constraint comes from predefined concentrations of conserved
moieties. A fixed concentration [ATP]+[ADP], for example, would define
a nonlinear constraint on the metabolite polytope. However, the
inequality constraint [ATP]+[ADP] $\le \const$ would lead to a convex
feasible region.  Similarly, we may postulate that the sum of enzyme
levels, or some weighted sums of enzyme levels (e.g., for enzymes
occupying a certain membrane) are bounded from above. Since these sums
are convex functions on the metabolite polytope, a bound on these sums
will define a convex set, and again, the optimality problem remains
convex.

\begin{figure}[t!]
  \begin{center}
    {\begin{tabular}{llll}
        (a) & (b) & (c) & (d) \\
        \includegraphics[height=3.8cm]{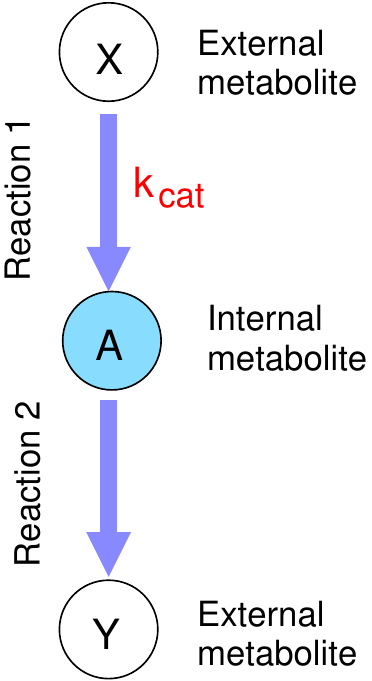} &
        \includegraphics[height=3.8cm]{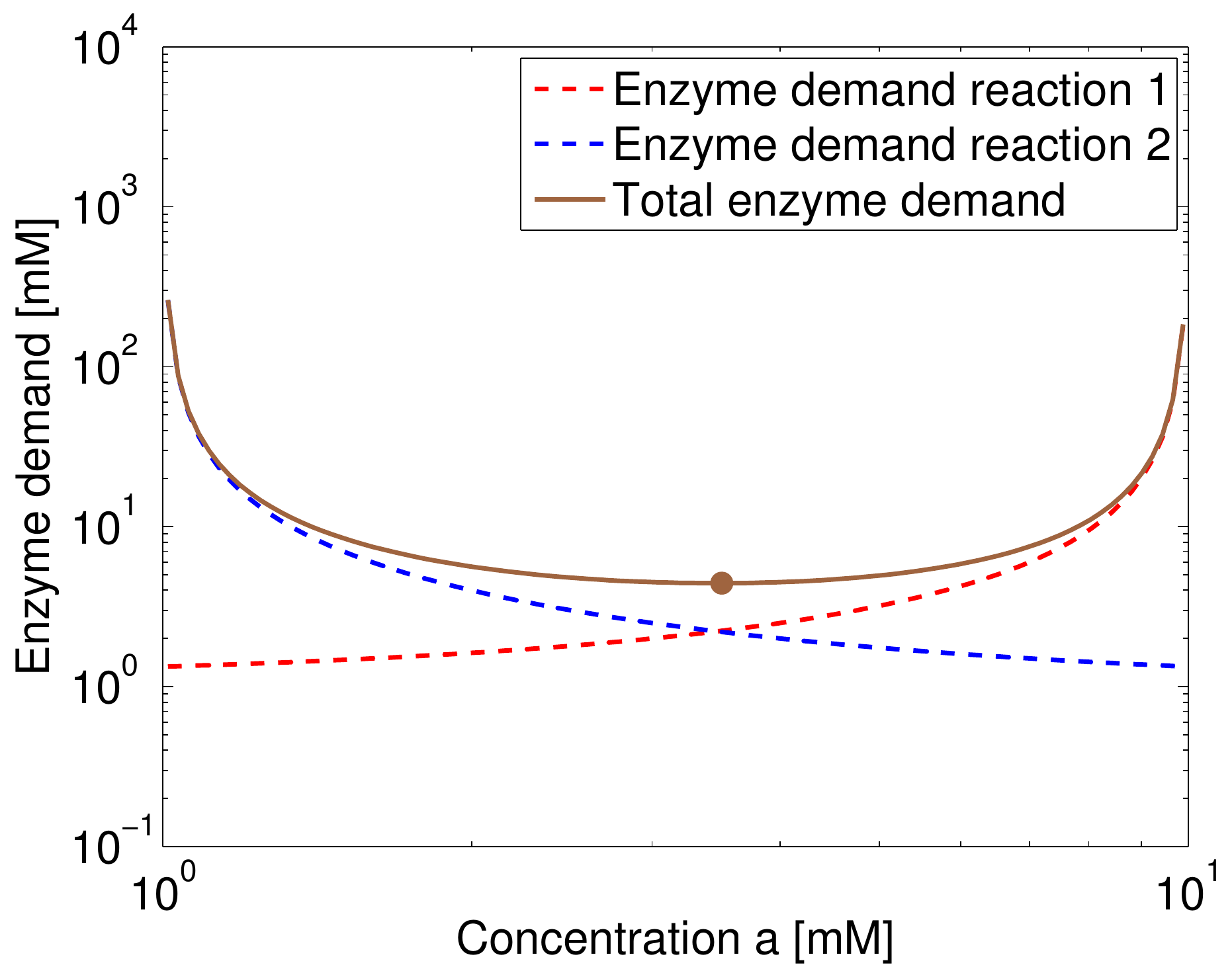} &
        \includegraphics[height=3.8cm]{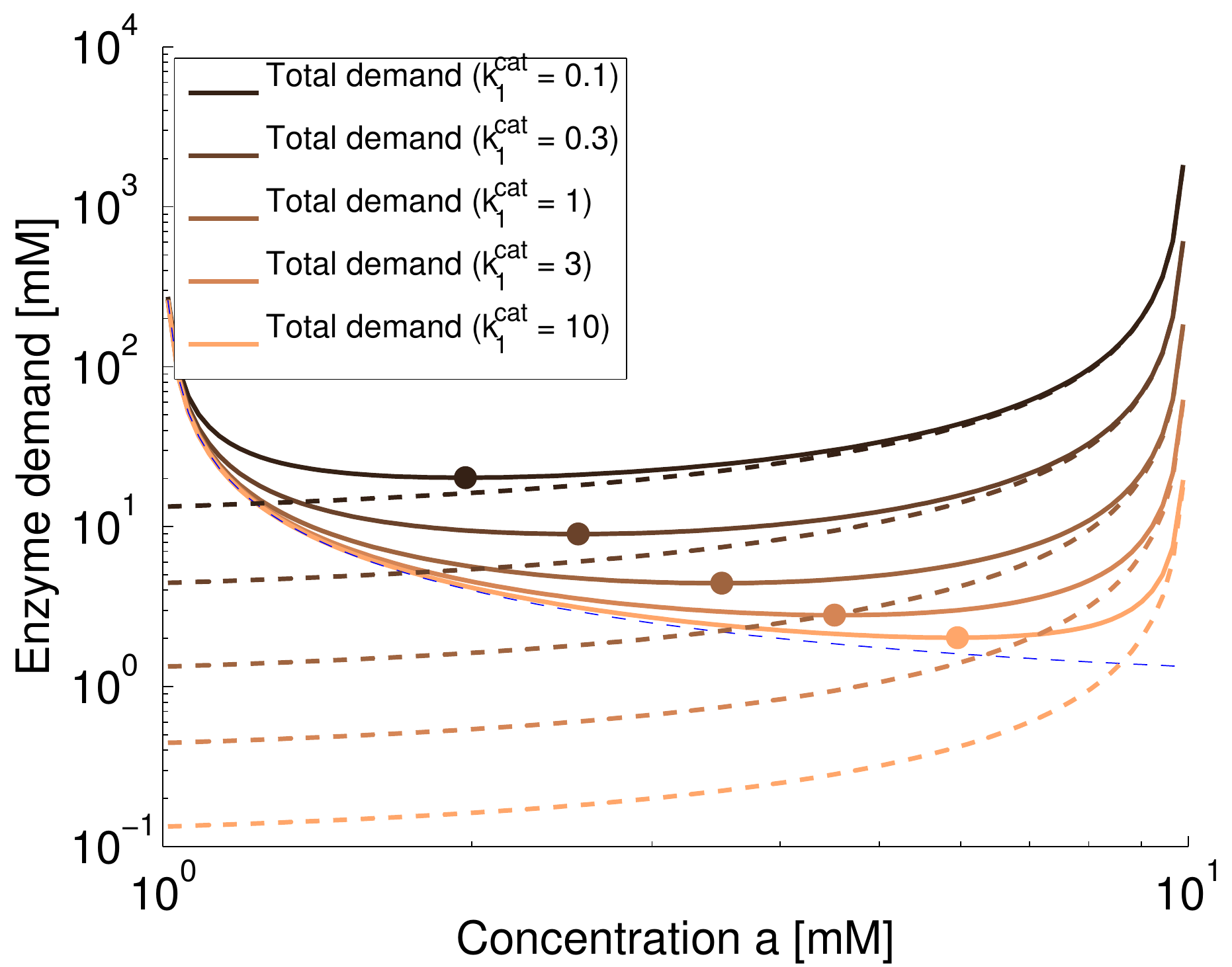}&
        \includegraphics[height=3.8cm]{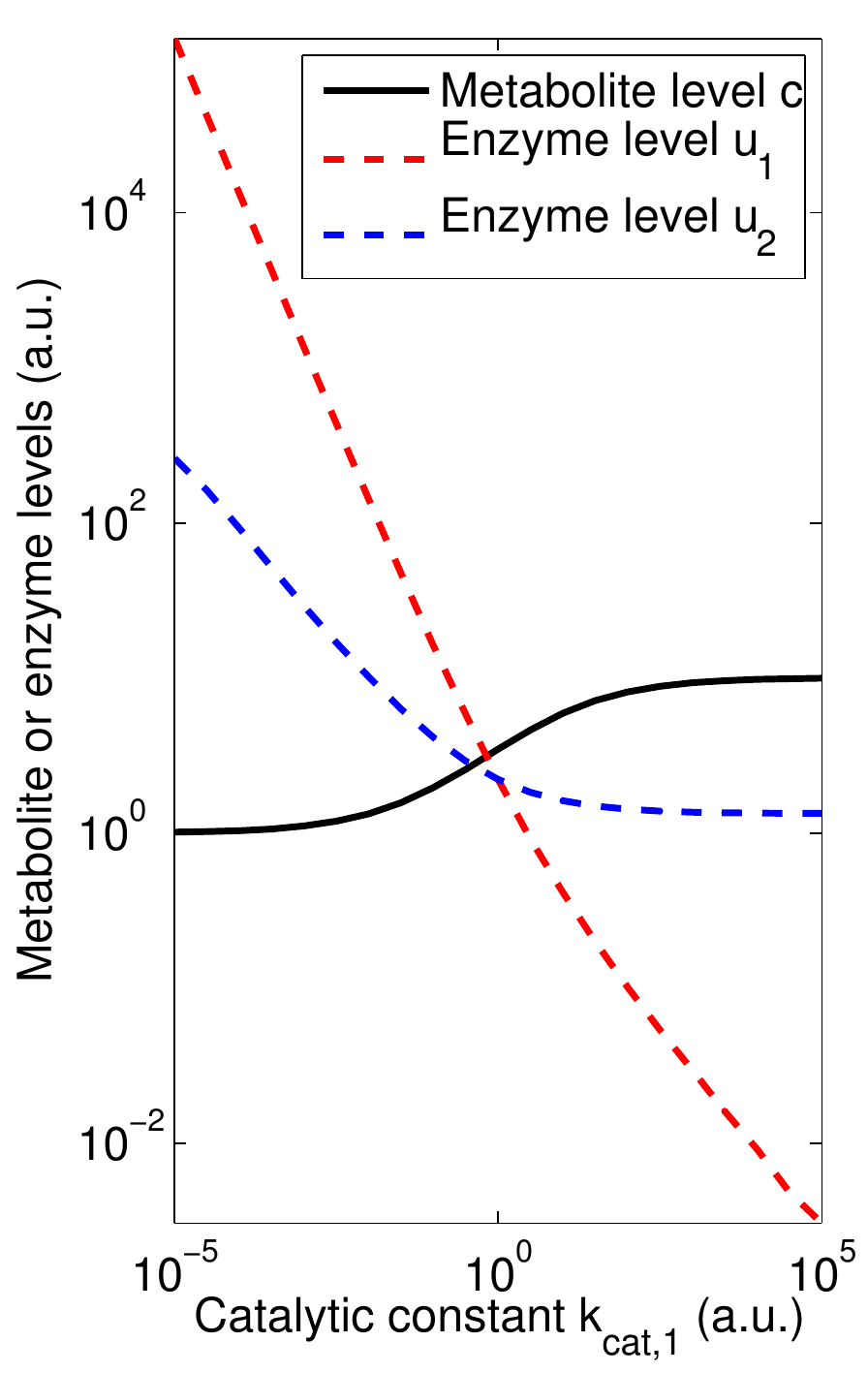}\\
      \end{tabular}
    }
    \caption{The optimal metabolic state depends on the catalytic
      constants.  (a) Two-reaction pathway with external
      concentrations $x=1$ and $y=0.1$.  The enzyme demand depends on
      the concentration $a$ of intermediate metabolite A. (b) Enzyme
      demand as a function of $a$, for a desired flux $v=1$ (a.u.) and
      assuming reversible Michaelis-Menten rate laws (all parameters
      set to 1). Close to chemical equilibrium in reaction 1 (right
      border) or reaction 2 (left border), enzyme demand diverges. The
      optimum metabolite concentration is marked by a dot.  (c) Enzyme
      demand depends on kinetic constants $\kcatplus$.  Results with
      varying $\kcatplus$ values in reaction 1 (reference value
      ${\kcatplus}=1$ s$\inv$) are shown.  Higher $\kcatplus$ values
      decrease the enzyme demand in reaction 1, shifting the optimum
      point towards higher values of $a$. (d) Optimal metabolite and
      enzyme levels shown as functions of $\kcatplus$.}
  \label{fig:single_metabolite}
  \end{center}
\end{figure}

\subsection{Optimal metabolic states depend on model parameters.}  
How do  optimal  metabolite and enzyme profiles depend on kinetic
parameters?  The metabolite profile reflects a compromise between
requirements in different reactions and depends on many 
model details. Changes in external concentrations or equilibrium constants
will shift the boundaries of the metabolite polytope, and changes in
$\kcat$ values, enzyme cost weights $\hel$, or desired fluxes $v_l$
will scale the cost of individual enzymes and shift the optimum
point (see Figure \ref{fig:fourchain} (f)).  Figure
\ref{fig:single_metabolite} shows this  for a varying
 $\kcat$ value. We consider a two-reaction pathway with 
parameters set to 1 (arbitrary units).  We note that a higher
intermediate level $a$ decreases the force in reaction 1 (i.e.,
increases its enzyme demand) and increases the force in reaction 2
(i.e., decreases its enzyme demand).  With the parameters chosen, the
total enzyme demand becomes minimal when both reactions show the same
driving force: this is the state that would also be predicted by the MDF
method (which focuses on driving forces instead of enzyme costs
\cite{nbfr:14}).  If we increase the $\kcat$ value in reaction 1, the
cost of  enzyme 1 will have a smaller impact on the overall cost,
and the optimal concentration $a$ is shifted to higher values. Since
the cost of  enzyme 2 becomes more dominant,  energy
efficiency in this reaction is increased on the expense of reaction 1.
A variation of enzyme cost weights $\hel$ or fluxes $v_{l}$ will have
similar effects as variations of $1/\kcatplusl$.

\subsection{Preemptive enzyme expression as a convex optimality problem}
\label{eq:preemptiveExpression}

Cells have to deal with varying environments which require different
fluxes and enzyme levels. Since switching takes time and the resulting 
maladaptation can be costly, a possible strategy is to
anticipate all possible (or likely) situations and to express enzymes
preemptively. In a simple strategy, the cell could express all enzymes
at a constant level and inhibit some of them in each situation
 to realise a favourable state.  The choice of optimal preemptive enzyme
levels can be formulated as an optimality problem, which turns out to be 
 convex.  We
assume a set of possible situations $\sigma$, each
characterised by different conditions (external metabolite
concentration vector $\cv^{\rm ext,\sigma}$ and other
kinetics-relevant parameters $\pv^{\sigma}$) and a necessary  flow
$\vv^{\sigma}$. For simplicity, we assume that each reaction has a fixed 
flux
direction across all conditions. Each  condition leads to a
different metabolite polytope ${\mathcal P}_\sv^\sigma$ and to
a different specific rate function $r^\sigma(\vv,\sv)$.  A preemptive
adaptation strategy is a tuple $\{\sv^\sigma\}$ of metabolite
profiles for the different situations, each located in its 
metabolite polytope 
$\sv^\sigma
\in {\mathcal P}_\sv^\sigma$.  The corresponding \emph{required enzyme
  activities} comprise the enzyme profiles required in the different
situations:
\begin{eqnarray}
   \enzyme_l^{\sigma}(\sv^{\sigma}) = \frac{v_l^\sigma}{r_l^\sigma(\sv^{\sigma})}
\end{eqnarray}
We note that each of the $\enzyme_l^{\sigma}(\sv^{\sigma})$ is a convex
function on the corresponding metabolite polytope ${\mathcal
  P}_\sv^\sigma$.  To define the overall cost of the strategy, we
determine, for each enzyme, the maximal level that it needs to show
across situations (for all other situations, we assume that the enzyme
activity will be reduced allosterically, without reducing the actual
enzyme cost).  Thus,
\begin{eqnarray}
 q^{\rm strategy}(\{\sv^\sigma\}) = \sum_l \mbox{max}_\sigma\, \enzyme_l^{\sigma}(\sv^\sigma).
\end{eqnarray}
We now show that this cost is a convex function on the product
polytope $\prod_\sigma {\mathcal P}_\sv^\sigma$. First of all, the
cost is convex if for each reaction $l$, the cost 
\begin{eqnarray}
 q^{\rm strategy}_l(\{\sv^\sigma\}) = \mbox{max}_\sigma\, \enzyme_l^{\sigma}(\sv^\sigma)
\end{eqnarray}
related to this reaction is convex. This is what we show now. From
ECM, we know that $\enzyme_l^{\sigma}(\sv^\sigma)$ is convex on ${\mathcal
  P}_\sv^\sigma$, so
\begin{eqnarray}
\enzyme_l^{\sigma}([1-\lambda]\,\sv^\sigma_A + \lambda\,\sv^\sigma_B)
\le [1-\lambda]\, \enzyme_l^{\sigma}(\sv^\sigma_A)
+ \lambda\,\enzyme_l^{\sigma}(\sv^\sigma_B).
\end{eqnarray}
Thus,
\begin{eqnarray}
 q^{\rm strategy}_l([1-\lambda]\,\{\sv^\sigma_A\} + \lambda\,\{\sv^\sigma_B\})
&=& \mbox{max}_\sigma\, \enzyme_l^{\sigma}([1-\lambda]\,\sv^\sigma_A + \lambda\,\sv^\sigma_B) \nonumber \\
&\le& \mbox{max}_\sigma\, \left( [1-\lambda]\, \enzyme_l^{\sigma}(\sv^\sigma_A)  + \lambda\,\enzyme_l^{\sigma}(\sv^\sigma_B) \right) \nonumber \\
&\le& [1-\lambda]\, \mbox{max}_\sigma\, \left( \enzyme_l^{\sigma}(\sv^\sigma_A) \right)
+ \lambda\,\mbox{max}_\sigma\, \left(\enzyme_l^{\sigma}(\sv^\sigma_B) \right) \nonumber \\
&=& [1-\lambda]\, q^{\rm strategy}_l(\{\sv^\sigma_A\})
+\lambda\, q^{\rm strategy}_l(\{\sv^\sigma_B\})
\end{eqnarray}
This shows that $ q^{\rm strategy}_l(\{\sv^\sigma\})$ is a convex
  function.  The first inequality holds because
  $\enzyme_l^{\sigma}(\sv^\sigma)$ is convex in $\sv^\sigma$. the second
  inequality holds because the maximum function
  $\mbox{max}(a_1,a_2,,a_3,...)$ is convex in its arguments.

\subsection{Non-enzymatic reactions}
\label{sec:SInonEnzymatic}

So far, we generally assumed that all reactions in a model are
enzyme-catalysed. In reality, some chemical reactions are
fast enough even without a catalyst, as is also the case for membrane
diffusion (e.g., for small molecules like O$_2$ and CO$_2$). These
processes are often counter-productive, such as spontaneous
degradation of complex compounds or leakage efflux of useful
metabolites. Furthermore, since many models of growing cells use
metabolite concentrations (not absolute amount) as variables, the
increase of cell volume dilutes these concentrations and is thus
equivalent to a global degradation rate, which might be significant
(e.g. in fast growing bacteria). These non-enzymatic processes could
have a large impact on the metabolic flows or, in the perspective
taken here, on how costly certain flows will be. Thus in general, our
flows contain, aside from enzymatic reactions, a number of
non-enzymatic reactions degrading or converting metabolites, most
probably with mass-action rate laws.  This radically changes things:
in ECM, effectively, non-enzymatic reactions put additional
constraints on metabolite concentrations, which confine the metabolite
polytope to a subspace, and may make the polytope become empty.

\textbf{In ECM, non-enzymatic reactions lead to constraints on the
  metabolite polytope, but leave the optimality problem convex}
Consider an ECM problem with non-enzymatic reactions. The rate laws
$v^{\rm non}_j = r_j(\cv)$ have the general form of reversible
reactions (with thermodynamic numerator, and some
concentration-dependent denominator), such that the functions
$1/r_j(\sv)$ will be convex on the metabolite polytope.  In ECM, with
given fluxes, the metabolite levels must be such that the flux is
realised. The non-enzymatic reactions are not scored by enzyme costs,
but they create (potentially nonlinear) equality constraints on the
metabolite polytope. In the simple case of \emph{irreversible mass
action laws}, we obtain a linear equality constraint on the metabolite
polytope, that is, the polytope is cut by a plane, and all solutions
must lie in the resulting subspace. Obviously, this leaves the
optimization problem convex. 

\subsection{The enzyme cost profile obtained by ECM is a linear combination of  metabolic control profiles}
\label{sec:SIcostAndControlCoefficients}

In kinetic models of metabolic pathways, a flux maximization at a
fixed total enzyme level will lead to a state in which enzyme levels
and flux control coefficients are proportional
\cite{hekl:1996}.  Treating metabolic pathways by ECM, we obtain
a similar, yet more general relationship between enzyme costs and
metabolic control coefficients. The cost of an enzyme is proportional
to a linear combination of metabolic control coefficients, and the
control coefficients appearing in this linear combination refer to (i)
possible stationary flux modes in the network and (ii) concentrations
of internal metabolites that are either kept fixed or hit a bound in
ECM. 

\begin{proposition}
\label{prop:lincomb} \textbf{(Enzyme cost and control
  coefficients)} In kinetic models with metabolic states obtained by
enzyme cost minimization, the profile of enzyme cost $\hel\,\el$ is
a linear combination (proof see section
\ref{sec:ProofEnzymeControl})
\begin{eqnarray*}
 \hel\,\el = \sum_{a \in \rm stat} \alpha_{a} {\mathcal C}^{j_a}_{l}
  + \sum_{b \in \rm bnd} \beta_{b} \,{\mathcal C}^{s_b}_{l}
\end{eqnarray*}
of flux control coefficients (for the independent stationary fluxes
$j^{\rm stat}_a$) and of metabolite contral coefficients (for internal
metabolites $s^{\rm bnd}_b$ that hit upper or lower bounds). The
coefficients $\beta_{i}$ assume values $\beta_{i}=0$ when a metabolite
hits none of the bounds, $\beta_{i}>0$ when it hits the lower bound,
and $\beta_{i}<0$ when it hits the upper bound.  In particular, if a
network allows for a single stationary flux mode only, the enzyme
costs show a proportionality
\begin{eqnarray}
\label{eq:ProofEnzymeControl}
  \hel\,\el \propto  \,{\mathcal C}^{\rm J}_{l} 
  + \sum_{b \in \rm bnd} \beta'_{b}\,{\mathcal C}^{\rm s_{b}}_{l}.
\end{eqnarray}
\end{proposition}

In the kinetic model, enzymes have no control over external
metabolites. Thus, if all fixed metabolites (in ECM) are considered
external (in the kinetic model), and if none of the internal
metabolites (in the kinetic model) hits a bound (in ECM), the
concentration control coefficients do not appear in the sum, then
enzyme costs are directly proportional to flux control coefficients
\begin{eqnarray}
  \hel\,\el \propto \,{\mathcal C}^{\rm J}_{l}.
\end{eqnarray}
Assuming equal cost weights for all enzymes, this yields a
proportionality between enzyme levels and flux control coefficients
as previously found in \cite{hekl:1996}.

\begin{figure}[t!]
  \begin{center}
    \begin{tabular}{lllll}
&    (a) & (b) & (c) & (d) \\
    \includegraphics[width=1.5cm]{ps-files/four_chain-eps-converted-to.pdf} & 
      \includegraphics[height=3.3cm]{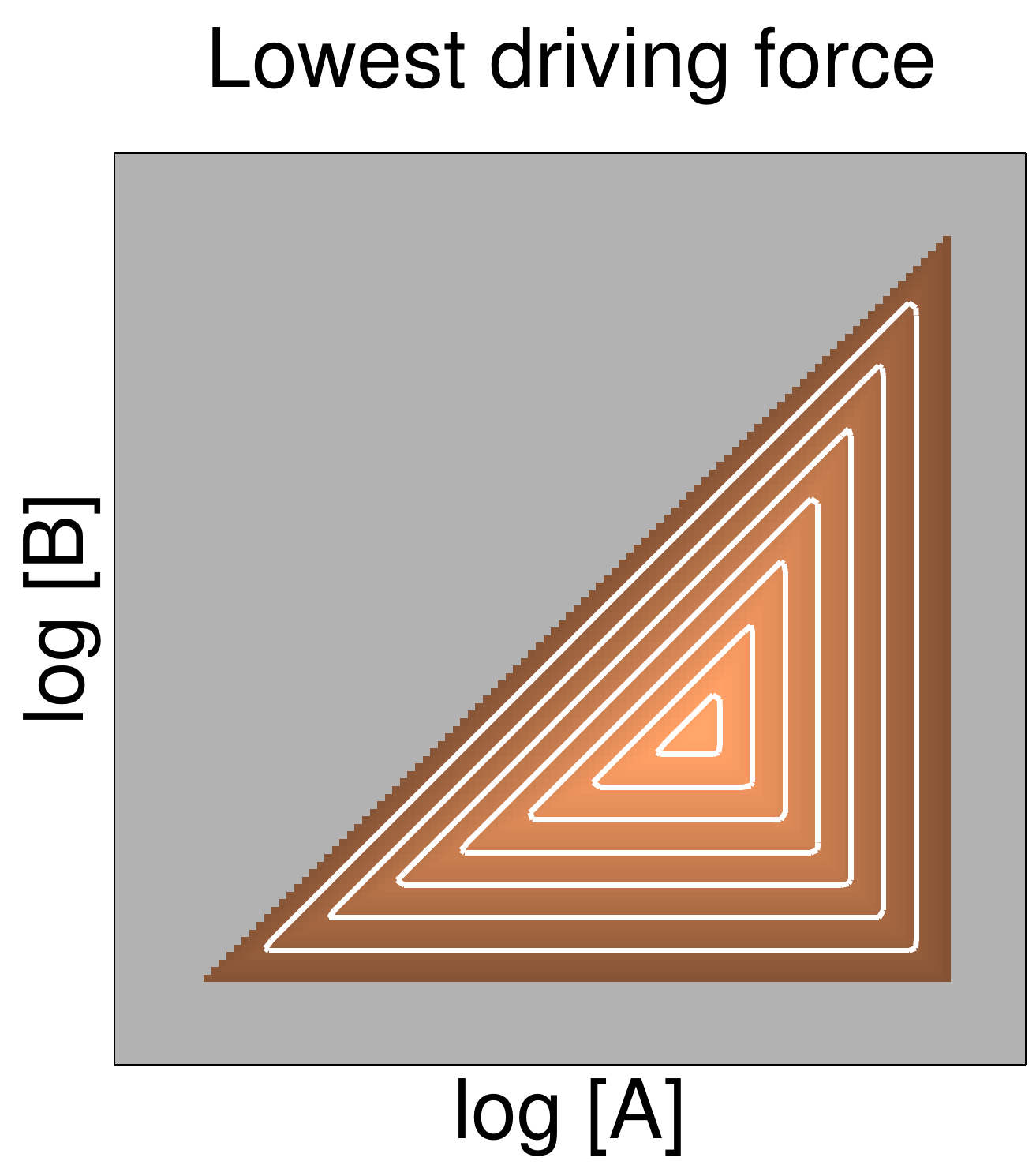} & 
      \includegraphics[height=3.3cm]{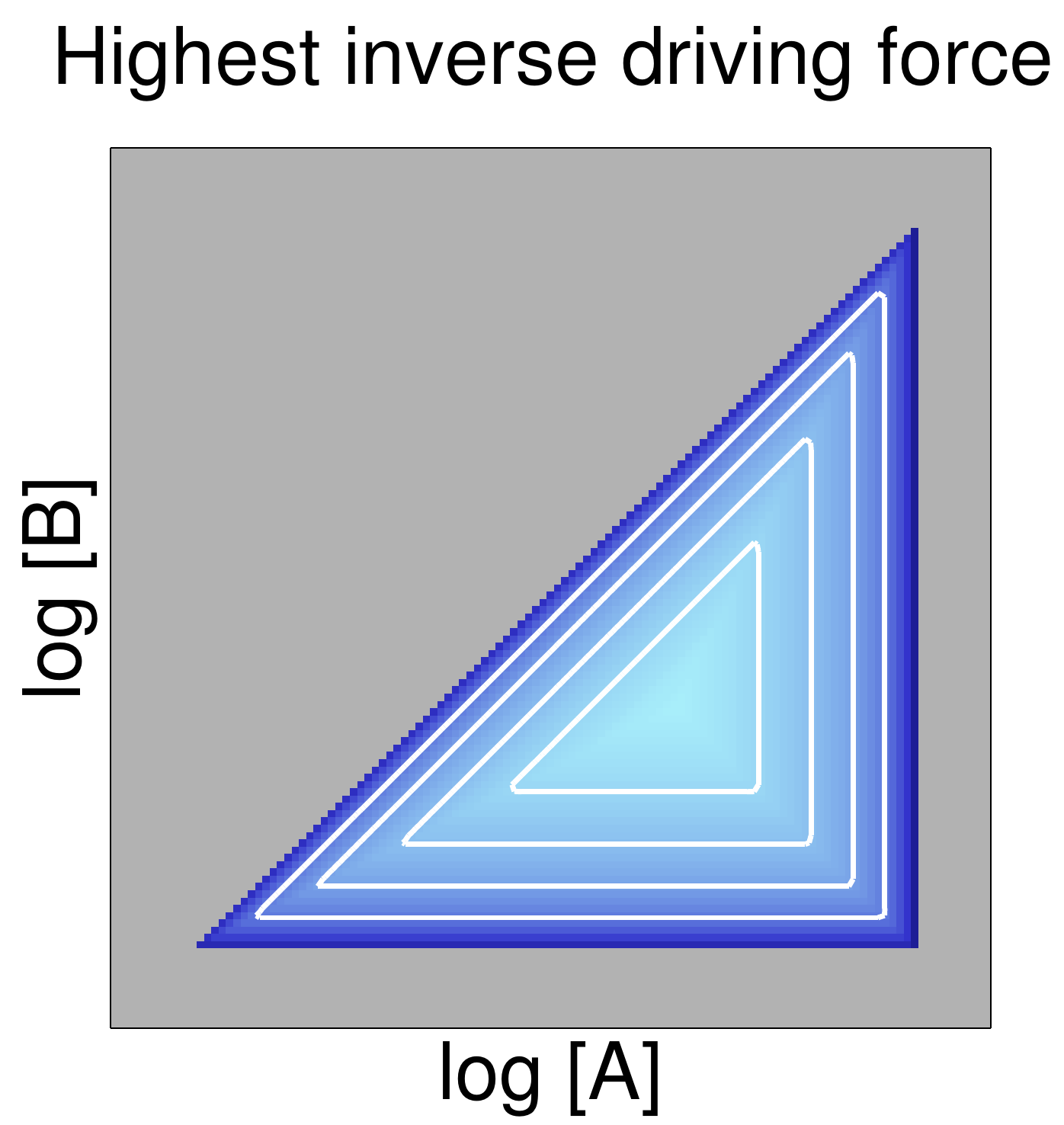}& 
      \includegraphics[height=3.3cm]{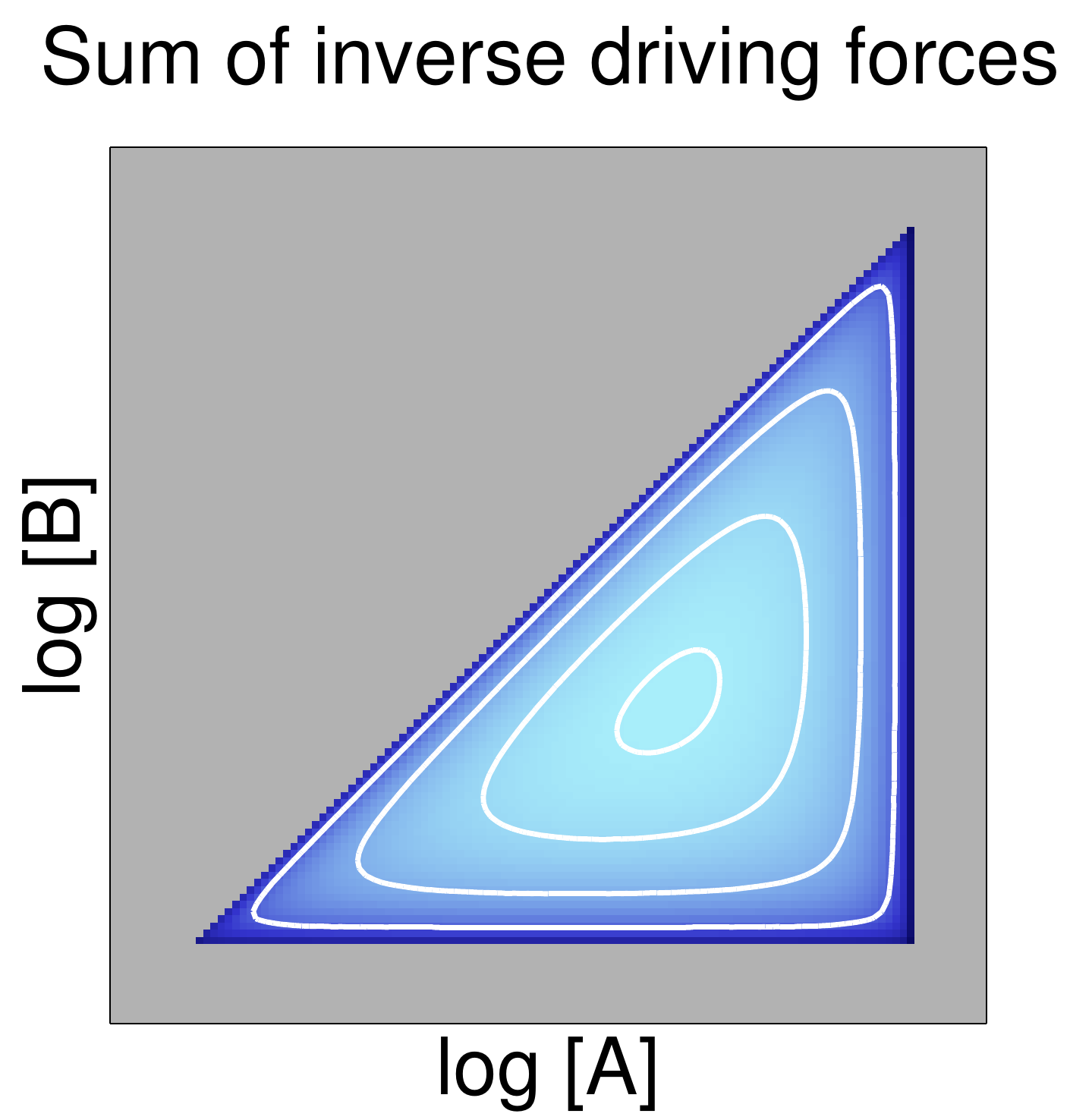}&
      \includegraphics[height=3.3cm]{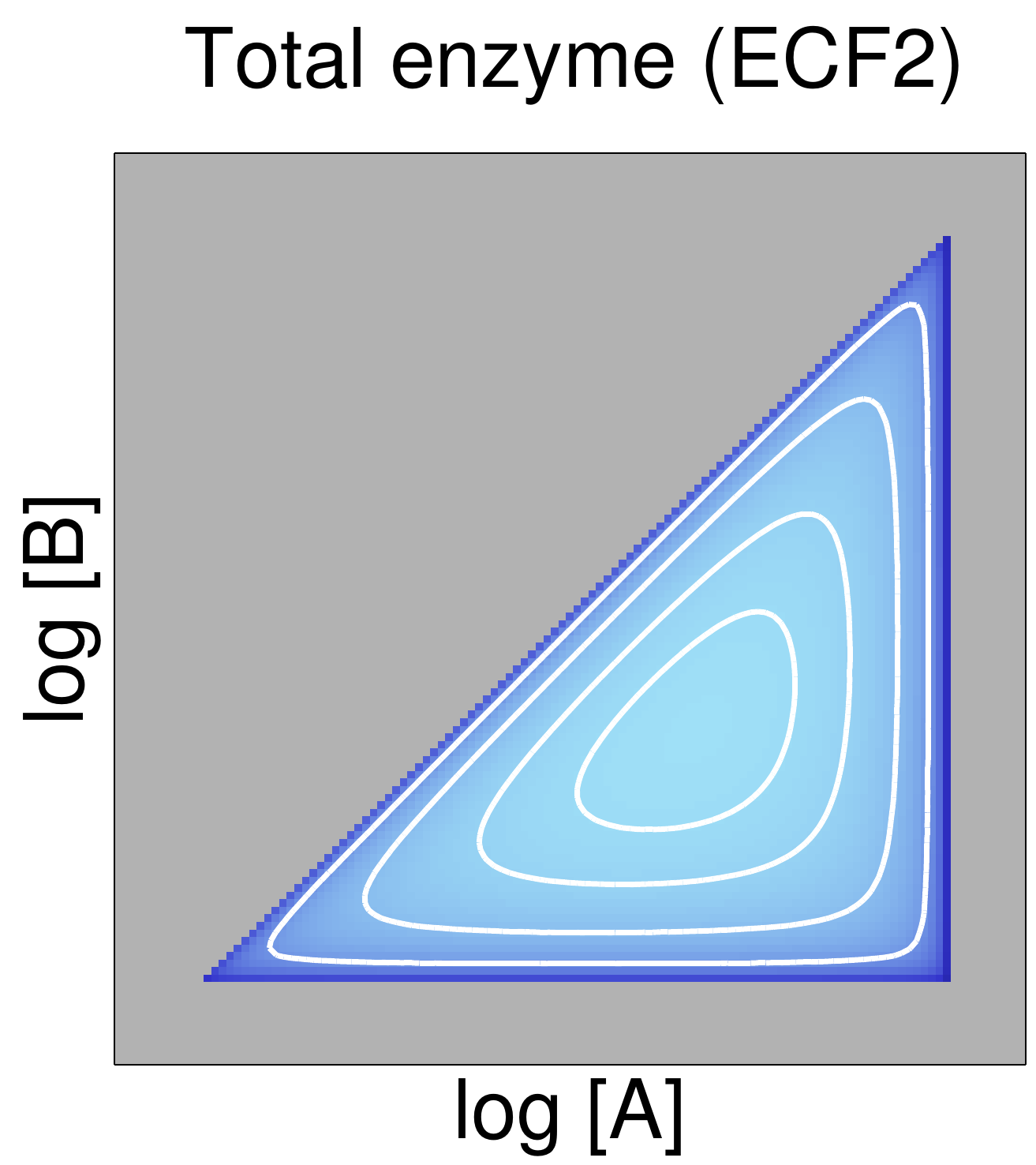}
 \end{tabular}
 \caption{Estimates of enzyme cost, computed from driving forces.  The
   schemes show different possible cost functions for a three-reaction
   pathway (left; same as in Figure \ref{fig:fourchain}).  (a) Lowest
   driving force $\Theta^{\rm min} = \mbox{min}_{l}\,\Theta_l$ along
   the pathway (colors from small (dark) to large (bright)). The
   contour lines can be used to define stricter constraints on
   metabolite profiles (see text). In the maximum point, all reactions
   have equal driving forces; this is the state that would follow from MDF
   optimization.  (b) The highest inverse driving force in the
   pathway, ($\mbox{max}_{l}\, \Theta_l\inv$), shows the same type of
   contour lines and the same optimum point (logarithmic color scale,
   from small (blue) to large (red)).  (c) The sum $\sum_{l}
   \Theta_l\inv$ of inverse driving forces along the pathway.  (d) The
   sum $\sum_l [1-\e^{-\Theta_l}]\inv$ is an EMC2s function with all
   parameters set to 1. The function (b) is always lower than (c), (c) is
   lower than (d), and (d) is lower than an EMC3 function with parameters
   set to 1 (the one shown in Figure \ref{fig:fourchain}).  All these
   functions have different optimum points.}
\label{fig:fourchainOther} 
\end{center}
\end{figure}

\begin{figure}[t!]
  \begin{center} {\begin{tabular}{llll} (a) & (b) & (c)&
    (d) \\ \includegraphics[height=3.2cm]{ps-files/two-chain-eps-converted-to.pdf}
    & \includegraphics[height=3.2cm]{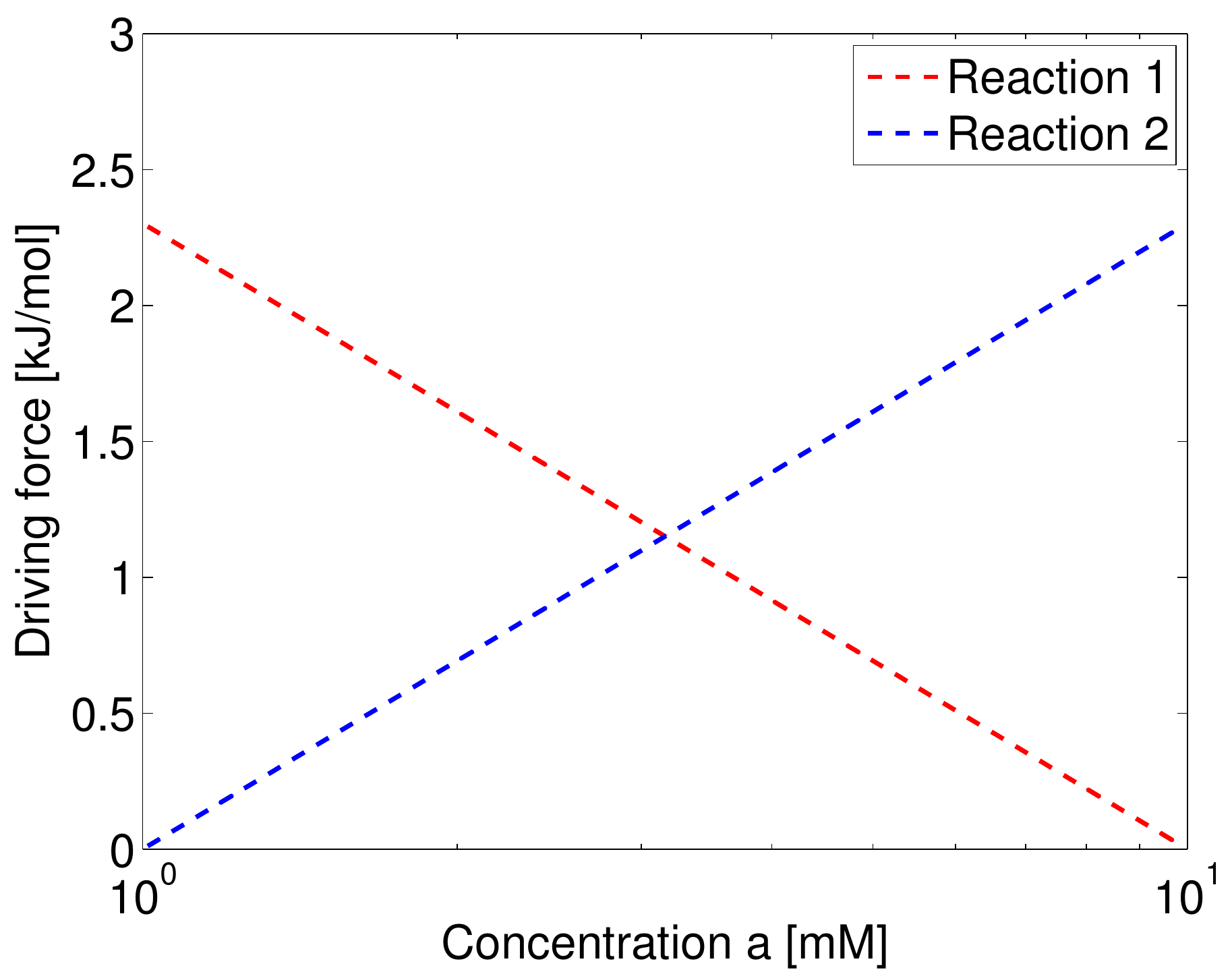}
    & \includegraphics[height=3.2cm]{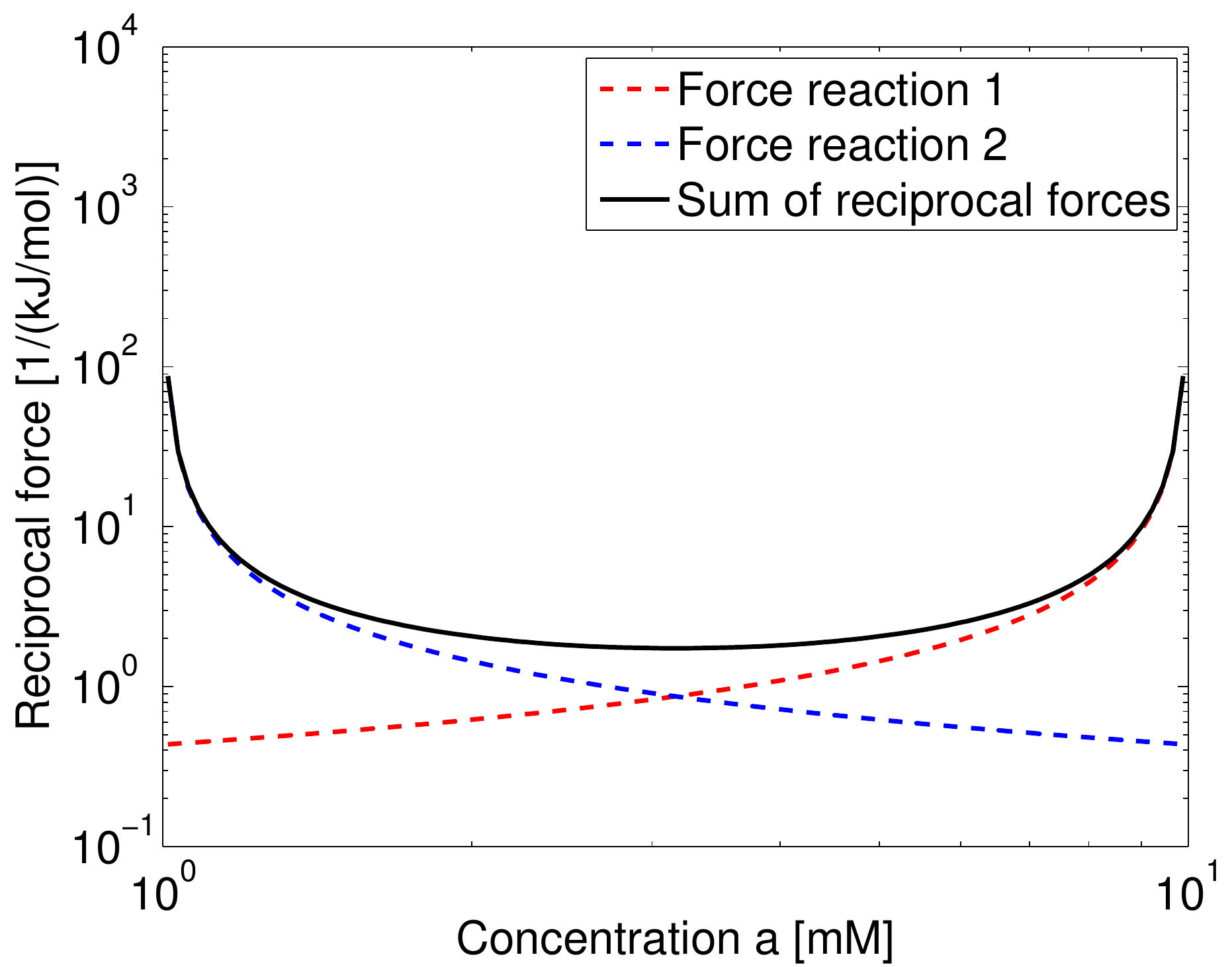}
    & \includegraphics[height=3.2cm]{ps-files/example_single_metabolite/example_single_metabolite_enzyme_demand-eps-converted-to.pdf} \end{tabular}
    } 
\end{center} 
\caption{Two-reaction pathway with external
    concentrations $x=1$ and $y=0.1$ (same model as in
    Figure \ref{fig:single_metabolite}). The driving forces depend on
    the concentration $a$ of intermediate metabolite A.  (a) Driving
    forces as functions of $a$ (note the logarithmic scale).  (b) The
    inverse driving forces $1/\Theta_l$ and their sum are convex
    functions of $\log a$. As shown in
    Figure \ref{fig:fourchainOther}, scaled inverse driving forces can
    be used as lower estimates of enzyme cost.  (c) Enzyme demands for
    desired fluxe $v=1$, assuming reversible Michaelis-Menten rate
    laws (all parameters set to 1). Close to chemical equilibrium in
    reaction 1 (right border) or reaction 2 (left border), the enzyme
    investment diverges. The optimum concentration is marked by a
    dot.}  \label{fig:single_metabolite2}
\end{figure}

\section{Energy-based  cost functions and limits on driving forces}
\label{sec:SItighterConstraints}

By studying the infeasible zones along polytope faces, we can derive tighter
  constraints on driving forces. The enzyme cost minimum can be inside
the metabolite polytope or on a P-face (see Figure Figure
\ref{fig:fourchain} (f)). Since enzyme costs rise fast near E-faces of
the polytope, the optimum point will not be located in those regions,
and it can be practical to exclude these regions from the metabolite
polytope. This simply means that we introduce positive lower bounds on
all driving forces.  However, when should a driving force count as
small? To define a threshold, we limit the enzyme cost in each
reaction by an upper bound $\enzymemetcostmax$, e.g.~five percent of
the total cost of the proteome.  We use the the general EMC equation
(Eq.~\ref{eq:TotalEnzymeDemand}) to obtain a lower bound for driving
forces ($\Theta_l$), by remembering that $\eta^{\rm enr} < 1$ and
$\eta^{\rm reg} < 1$:
\begin{eqnarray} 
\label{eq:forceThreshold} 
\frac{\hel\,v_{l}}{\kcatplusl \,\enzymemetcostmax} < 
\frac{\hel\,v_{l}}{\kcatplusl \,\enzymemetcost_{l}} =
\eta^{\rm enr}\eta^{\rm sat}\eta^{\rm reg} <
\eta^{\rm enr} = 1 - e^{\Theta_l}
 < \Theta_l,
\end{eqnarray} 
where the last step uses the fact that $1-\e^x < x$ for $x > 0$.
 The value varies between reactions (derivation
 in SI \ref{sec:ProofforceThreshold}) and depends on fluxes, $\kcat$
 values, and enzyme cost weights $\he$ (where protein mass can be used
 as a proxy).  With the constraint (\ref{eq:forceThreshold}), we
 obtain a smaller metabolite polytope in which the costly regions
 close to the previous E-faces are excluded. The same constraints can
 also be used in thermodynamics-based FBA. Usually,
 thermodynamics-based FBA requires that fluxes and driving forces 
 have the same sign (e.g.~\cite{belq:02,hohh:07,tnah:13}), but fluxes
 are allowed to be driven by infinitesimal forces. With our stricter
 (and more realistic) constraint, an FBA model would require, instead,
 that forces must be large enough to realize fluxes at plausible
 enzyme costs.

Other metabolite constraints, which ensure sufficient substrate
levels, can be derived similarly.  By putting an upper bound on the
enzyme cost, we obtain a lower bound on the saturation efficiency
$\eta^{\rm sat}$, and thus on the substrate levels.  Consider, for
instance, a reversible MM rate law for a reaction S
$\rightleftharpoons$ P in the factorized form
(\ref{eq:factorised}). Noting that $\eta^{\rm enr} < 1$ and $\eta^{\rm reg} < 1$, we obtain a
lower bound on  enzyme cost:
\begin{eqnarray}
\frac{\hel\,v_{l}}{\kcatplusl \,\enzymemetcostmax} < 
\eta^{\rm enr}\eta^{\rm sat}\eta^{\rm reg} <
\eta^{\rm sat} =
\frac{\concS/\kmS}{1+\concS/\kmS + \concP/\kmP} <
\frac{\concS}{\kmS}.
\end{eqnarray}
For multi-substrate reactions, we obtain linear inequality constraints
in log-concentration space. Using the EMC3s function, we obtain the
cost estimate
\begin{eqnarray}
\enzymemetcost > \enzymemetcostl^{\rm min} = \frac{h\, v}{\kcatplus
  \eta^{\rm sat}} > \frac{h \,v}{\kcatplus} \frac{1+ \prod_i
  \left(s_i/\kmi\right)^{m_i}} {\prod_i \left(s_i/\kmi\right)^{m_i}}>
\frac{h \,v}{\kcatplus}\, \prod_i \left(\kmi/s_i\right)^{m_i},
\end{eqnarray}
 where $s_i$ denotes substrate
concentrations and $m_{i}$ denotes substrate molecularities.  With the
upper bound $\enzymemetcost<\enzymemetcost^{\rm max}$, we obtain the
constraint 
\begin{eqnarray}
\prod_{i} s_{i}^{m_i} > \frac{h\,v\,\prod_{i} (\kmi)^{m_{i}}}
     {\kcatplus \,q^{\rm max}}\nonumber \\
\Rightarrow \quad \sum_{i} {m_i}\,\ln s_{i} >
     \ln \frac{h\,v\,\prod_{i} (\kmi)^{m_{i}}} {\kcatplus \,q^{\rm
         max}}.
\end{eqnarray}
Bounds for allosteric regulators (lower bounds for activators, upper
bounds for inhibitors) are derived in a similar way.

\paragraph{Lower estimates of flux costs; an extension of the MDF strategy}
 The total enzyme cost of a pathway, $\enzymemetcost(\sv)$, can be a
 complicated function of the log-metabolite levels. However, simple
 functions can be used as lower bounds (see Figures
 \ref{fig:fourchainOther} and \ref{fig:single_metabolite2}).  First,
 in a pathway with $N$ reactions, the total cost is always bounded by
 $N\,\min_l \enzymemetcostl(\sv)$ and $N\,\max_l
 \enzymemetcostl(\sv)$, i.e., $N$ times the lowest or the highest
 enzyme cost in the pathway. Second, the simplified EMC functions
 yield lower estimates of the cost. By combining these arguments, we
 can justify the Max-min Driving Force strategy \cite{nbfr:14}.  The
 MDF strategy is a heuristics for predicting the concentrations and
 driving forces in a pathway. It postulates that the smallest driving
 force in a pathway should be as large as possible. The MDF criterion
 is equivalent to minimizing $\mbox{max}_{l} [1-\e^{-\Theta_l}]\inv$,
 which is a lower bound on the EMC2s function $\sum_{l}
 [1-\e^{-\Theta_l}]\inv$ with all constants set to 1.  As shown in
 Fig.~\ref{fig:fourchainOther} (a) and (b), the MDF optimum is distant
 from the polytope E-faces and close to the minimum point of the EMC2
 function. Thus, the MDF strategy avoids excessive enzyme costs that
 would occur at the polytope surface. As in
 Eq.~(\ref{eq:forceThreshold}), one could devise a variant of MDF in
 which driving forces are weighted by the prefactors
 $\frac{\kcatplusl}{\hel\,v_{l}}$.

\section{Workflow for model building and metabolic optimization}
\label{sec:SIworkflow}

\begin{figure}[t!]
  \begin{center}
    \begin{tabular}{ll}
    \includegraphics[width=13.5cm]{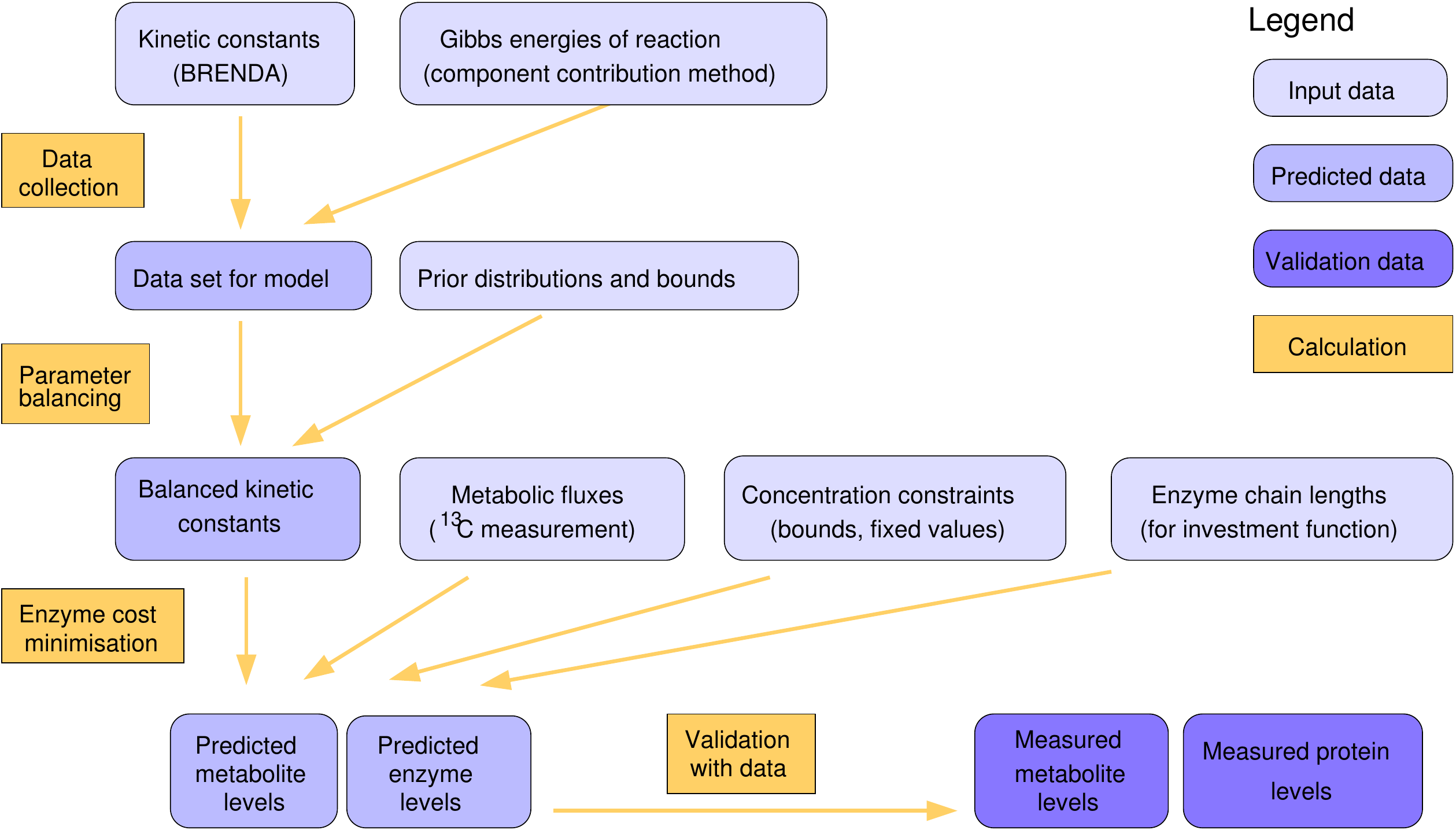}
  \end{tabular}
  \caption{Data integration in the ECM-based modelling workflow. The
    kinetics phase (data collection and parameter balancing) is
    followed by the optimization phase (ECM and validation of
    results).}
  \label{fig:workflowData}
  \end{center}
\end{figure}

\subsection{Workflow description} 
\label{sec:SIoptimization}
Our algorithm for enzyme cost minimization has two main phases. In the
kinetics phase, we collect and adjust the model parameters and
construct a model with energetically consistent fluxes (exclusion of
infeasible cycles) and rate constants (satisfying Haldane
relationships and Wegscheider conditions).  To determine consistent
model parameters, the collected rate constants and equilibrium
constants are adjusted and completed by parameter balancing.
\begin{enumerate}
\item Collect thermodynamic and kinetic data: standard chemical
  potentials $\mu^{\circ}$, equilibrium constants $\keq$,
  Michaelis-Menten constants $\km$, forward and reverse catalytic
  constants $\kcatplus$ and $\kcatminus$.
\item Set some of these quantities to  fixed values (if desired).
\item Run parameter balancing (with priors, pseudo values, and upper
  and lower bounds) to obtain a complete, consistent set of rate constants.
\end{enumerate}
In the optimization phase,  the desired pathway flux is  realized by
optimal enzyme and metabolite profiles.
\begin{enumerate}
\item Set up the kinetic model (based on the given network, flux
  profile $\vv$, and model parameters).  Redefine the reaction
  directions such that fluxes are positive, and update all parameters.
\item Choose the bounds for metabolite concentrations (tight bounds or
  fixed values for metabolites with fixed concentrations, lower and
  upper bounds for the others).
\item Determine a feasible metabolite profile $\sv = \ln \cv$ (a
  profile within the metabolite polytope) as a starting point for
  numerical optimization. We consider three alternatives: (i) Use
  linear programming to construct a set of extreme points in the
  polytope (with maximal and minimal metabolite levels $s_i$); the
  center of mass of these points is then taken as the starting
  point. (ii) Use the point in the polytope that is closest to the
  center of the predefined metabolite bounds (solution of a quadratic
  programming problem) as the starting point. (iii) Use the solution
  of the MDF problem (linear programming problem) as the starting
  point.
\item Choose an EMC function and minimize it numerically with respect to
  $\sv$ under the constraints defining the metabolite
  polytope.
\item Compute the corresponding enzyme levels and cost.
\item Based on the optimal enzyme cost, define a maximal  tolerable cost (e.g.,
   one percent higher than the optimal total
  cost) and compute individual tolerances for metabolite and enzyme
  levels as described in Methods.
\item Validate the predicted enzyme and metabolite levels with
  experimental data.
\end{enumerate}
In theory, a convex optimization should converge without problems. As
a check, we can repeat the calculation with different starting points.

\subsection{Parameter balancing yields consistent rate constants}
For our kinetic models, we need consistent sets of rate constants
($\kcatplus$, $\keq$, and $\km$ values) satisfying Wegscheider
conditions and Haldane relationships (see SI
\ref{sec:SIratelaws}). Measured parameter values may be incomplete and
contradictory. Using parameter balancing \cite{lskl:10}, we can
translate such values into complete, consistent, and plausible
parameters for a given model. Plausible parameter ranges can be
defined by prior distributions for parameter types (e.g., mean values
and a standard deviation for logarithmic $\km$ values in general).
Parameter balancing works as follows. We collect all quantities that
appear in the data or in the model ($\ln \kcatplus$, $\ln \kcatminus$,
$\ln \keq$, $\ln \km$, $\Delta_{\rm r} {G^{\circ}}'$, $\mu^{\circ}$)
and merge them into a vector $\yv$. These quantities must satisfy
Wegscheider conditions and Haldane relationships, which defines linear
equality constraints between them. Accordingly, to satisfy the
constraints in a safe way, we write all these quantities as linear
combinations of independent parameters ($\ln k^{\rm V}$, $\ln \km$,
and $\mu^{\circ}$ values), whith the definition $k^{\rm
  V}=\sqrt{\kcatplus\, \kcatminus}$.  The independent parameters,
which are collected in a vector $\sv$, can be varied without violating
any constraints. The linear dependence between the complete and the
independent parameter sets can be written as $\yv = {\bf R} \,\sv$
with a matrix ${\bf R}$ derived from the model structure.  Using this
equation as a linear regression model, we can convert an
experimentally known vector $\yv^{\rm data}$ (which may be incomplete)
into a best estimate of the underlying vector $\sv$. Using the
estimate $\sv$, we again apply ${\bf R}$ to obtain a completed,
consistent version of $\yv$. Since this regression problem is usually
underdetermined, we employ Bayesian estimation. Priors allow us to
obtain plausible estimates even from sparse data. Accordingly, the
result is not simply a point estimate of $\yv$, but a multivariate
Gaussian posterior distribution for possible parameter vectors $\yv$.
A best estimate is given by the center of the distribution; from the
covariance matrix, we obtain uncertainties of individual model
parameters as well as the correlations between them.  Parameter
balancing can handle data of different amounts or quality. If
comprehensive data are available, they will just be adjusted to
satisfy the constraints; missing or uncertain data values will be
completed with plausible values. MATLAB code for parameter balancing
and hyperparameters specifying the prior distributions are provided on
github and www.metabolic-economics.de/enzyme-cost-minimization/.

\subsection{Possible modifications of the workflow}
\label{sec:SIworkflowDetails}

The workflow can be extended in a number of ways:
\begin{itemize}
\item \textbf{External/internal and fixed/variable metabolites.} In
  kinetic models, we distinguish between \emph{internal} metabolites,
  for which a mass balance must be satisfied within the model, and
  \emph{external} metabolites, for which no mass balance is required
  (possibly assuming that other reactions, outside the model, will fix
  the mass balance). In ECM, we distinguish between fixed metabolites
  (whose concentration is predefined) and variable metabolites (whose
  concentration is determined during ECM). It is important to note
  that the two distinctions need not coincide. Nevertheless,
  metabolites at the pathway boundaries (such as initial substrates,
  final products, and cofactors, which also participate in other
  pathways) are usually the ones that will be both external (in
  kinetic models) and fixed (in ECM).
\item \textbf{Fluxes need not be stationary.} The flux distribution
  used in ECM need not be stationary (in the sense that the variable
  metabolites satisfy mass balances).  Remember that variable
  metabolites and internal metabolites are not the same!  Of course,
  stationarity is a sensible assumption for whole-metabolism models on
  a certain timescale. However, fluxes that look stationary on the
  entire metabolic network may not look stationary on an individual
  pathway model (because there may be side reactions that fix the mass
  balances, but do not appear in the model).
\item \textbf{Inactive reactions.} Inactive reactions (with a reaction
  flux $v_l=0$) do not entail any energetic constraints or enzyme
  costs and can therefore by ignored. In constrast, if a driving force
  is known to vanish, this should be used as a constraint on the
  metabolite levels.
\item \textbf{Non-enzymatic reactions.} If non-enzymatic reactions
  (typically with mass-action rate laws) are included in the
  optimality problem, they contribute to the energetic constraints,
  but not to the enzyme cost function.
\item \textbf{Spatial structure.} ECM applies to compartment models,
  in which metabolites can have different concentrations in different
  compartments. Other spatial effects, such as substrate channeling,
  are ignored. To account for substrate channeling, the increased
  substrate concentration at enzymes' catalytic sites could be
  modelled, approximately, by using effective rate constants.

\item \textbf{Constraints on the sum of metabolite levels or sums of
  enzyme levels} In addition to our bounds on individual metabolite
  levels, we can also set a bound on the total (non-logarithmic)
  metabolite concentration in the cell \cite{tnah:13}. The resulting
  ECM problem remains convex (see section
  \ref{sec:SIproblemremainsconvex}).  Alternatively, one could
  penalize large total concentrations by subtracting a concave
  function $R(\sum_i c_i)$ from the enzyme cost; in log-concentration
  space, this would yield a convex cost term. The same holds for 
constraints on the sums of some enzyme levels. 

\item \textbf{Enzyme demand and cost per flux.} Under the assumptions
  made (linear enzyme cost function; fixed external metabolite
  levels in kinetic model), enzyme demand and cost scale
  proportionally with the pathway fluxes. This holds for all EMC
  functions, but not for the MDF score \cite{nbfr:14}, which remains
  constant under a proportional scaling of pathway fluxes\footnote{The
    linear scaling of enzyme levels holds only if all fluxes are
    scaled proportionally. In branched pathways, a non-proportional
    scaling would change the flux branching ratios. The resulting
    changes in the optimal concentrations at the branch points would
    change the enzyme cost in complicated ways.}.
\item \textbf{Constraints on concentrations} Constraints on metabolite
  levels can be justified as follows. Upper bounds may reflect the
  fact that space in cells is limited, and physiological concentration
  ranges for certain compounds may be known from experience.  Some
  metabolites may have high or low levels for specific reasons: for
  instance, yeast cells (and also Dunaliella algae) use high glycerol
  concentrations to balance high external salinity; other metabolites
  may be toxic in higher concentrations.  Lower bounds are important
  when using the EMC2 functions, because these functions favor low product
  levels while lacking the saturation factor,  which prevents very
  low substrate levels.
\item \textbf{Values and uncertainties of rate constants} Different
  EMC functions require different types of rate constants for their
  calculation. All functions require forward catalytic constants $k^{\rm
    c}_{+}$; EMC2 and higher functions require equilibrium constants,
  EMC3 or higher functions require Michaelis-Menten constants. Many rate
  constants are unknown and need to be estimated.  To determine the
  rate constants for our calculations, we collect known kinetic data
  and convert them into complete, consistent parameter sets by
  parameter balancing \cite{lskl:10,slsk:13}.  Parameter balancing
  yields a joint distribution of all model parameters describing their
  individual uncertainties and correlations. A consistent, most likely
  set of parameters follows from the median values of the marginal
  distributions.  By sampling parameters from their joint
  distribution, we can obtain an ensemble of model variants with
  different consistent parameter sets. By running ECM for many such
  model variants, we can study how uncertainties in the rate constants
  affect the end result.
\item \textbf{Sampling of nearly optimal solutions} Deviations from
  the optimum metabolite profile lead to a fitness loss. For small
  Gaussian random deviations, the average loss by can be computed by
  $\trace(\cov(\sv)\inv)\,\hessian$,  where $\sv = \ln \cv$, $\cov(\sv)$
  is the covariance matrix of metabolite log-concentrations, and
  $\hessian$ is the Hessian matrix of the (non-logarithmic) cost function
  $\enzymemetcost(\sv)$ in the optimum point. To estimate the
  metabolite covariance matrix, we make an assumption inspired by
  statistical thermodynamics: we postulate that the relative
  probabilities of two metabolite vectors is given by
  $\frac{\mbox{prob}(\sv_1)}{\mbox{prob}(\sv_2)} =
  e^{-(\enzymemetcost(\sv_1)-\enzymemetcost(\sv_2))/\enzymemetcost_0}$,
  where $\enzymemetcost_0$ defines a scale of tolerable fitness
  deviations.  The metabolite covariance matrix follows directly as
  $\Cmat = \enzymemetcost_0\,\hessian\inv$. For the
  energy-based EMC2s function, the Hessian matrix in the optimum point
  can be computed analytically at least.
\item \textbf{Tolerable deviations of metabolite and enzyme levels}
  Tolerance ranges of metabolite or enzyme levels can be obtained by
  minimizing or maximizing these levels under the constraints used in
  ECM, plus the constraint that the cost must remain below some
  predefined upper bound. To speed up the calculation, an
  approximation based on the Hessian matrix of the logarithmic cost
  function can be used (see SI \ref{sec:SIHessianApproximation}).
  Alternatively, we could sample metabolite profiles with enzyme costs
  close to the optimum. Using the Metropolis-Hastings algorithm, we
  could obtain an ensemble of metabolite and enzyme profiles, where
  less costly states appear with higher probabilities (see SI
  \ref{sec:SIoptimization}).
\item \textbf{Lumped reactions} To simplify models, pathways can be
  lumped into single reactions (for parameter choices, see SI
  \ref{sec:SIlumped}).  The lumping of reactions resembles the way in
  which ECM, altogether, attributes enzyme costs or specific
  activities to entire pathways.
\end{itemize}

\begin{table}[t!]
\begin{center}
\begin{tabular}{llll}
\textbf{Data type} & \textbf{Unit} & \textbf{Provenance} &
\textbf{Reference} \\ 
\hline 
Reaction Gibbs energies & kJ/mol & Component contribution & \cite{nhmf:13}\\ 
Catalytic constants ($\kcat$) & 1/s & BRENDA & \cite{sceg:04}\\ 
Michaelis-Mention constants ($\km$) & mM & BRENDA & \cite{sceg:04}\\ 
Fluxes & mM/s & van Rijsewijk \emph{et al.} $^a$ & \cite{hann:11} \\ 
Metabolite levels$^b$ & mM & Gerosa \emph{et al.} $^a$ & \cite{gerc:15} \\
Enzyme levels$^b$ & mM & Schmidt \emph{et al.} $^a$ &
\cite{skva:15} \\
Protein lengths & AAs & \url{www.uniprot.org} & \\ \hline
\end{tabular}
\end{center}
\caption{Data used in construction of \emph{E.~coli} model by
  ECM. Processed data can be found at www.metabolic-economics.de/enzyme-cost-minimization/.
  Units refer to preprocessed data.
  $^a$ Specific data corresponding to wild-type \emph{E.~coli} BW25113, grown in batch
       culture on minimal media (M9) and glucose.
  $^b$ Data used for validation only. }
\label{tab:data}
\end{table}

\section{Model of central metabolism in \emph{E.~coli}}
\label{sec:SIecolimodel}

Our central metabolism model was built from a list of chemical
reactions as given by KEGG; compounds and reactions are denoted by
KEGG identifiers, and genes are denoted mostly by their common names
in \emph{E.~coli}. All data sources are listed in Table
\ref{tab:data}, and models and data are provided at
\url{www.metabolic-economics.de/enzyme-cost-minimization/}.  The
enzyme cost function accounts for protein composition, giving
different costs to different amino acids. However, models with equal
cost weights for all proteins, or with size-dependent protein costs
yielded similar results (results are provided on the
website). Figure
\ref{fig:EcoliPredictionsSI} shows the correlations between predicted
and measured metabolite levels, corresponding to the enzyme
predictions in Figure \ref{fig:EcoliPredictions}. More details can be
found on \url{www.metabolic-economics.de/enzyme-cost-minimization/}.

\begin{figure}[t!]
  \begin{center}
	 \includegraphics[width=0.4\textwidth]{ps-files/ecoli_model_flux_from_haverkorn_fullnames-eps-converted-to.pdf} 
         \hspace{3mm}
	 \includegraphics[width=0.55\textwidth]{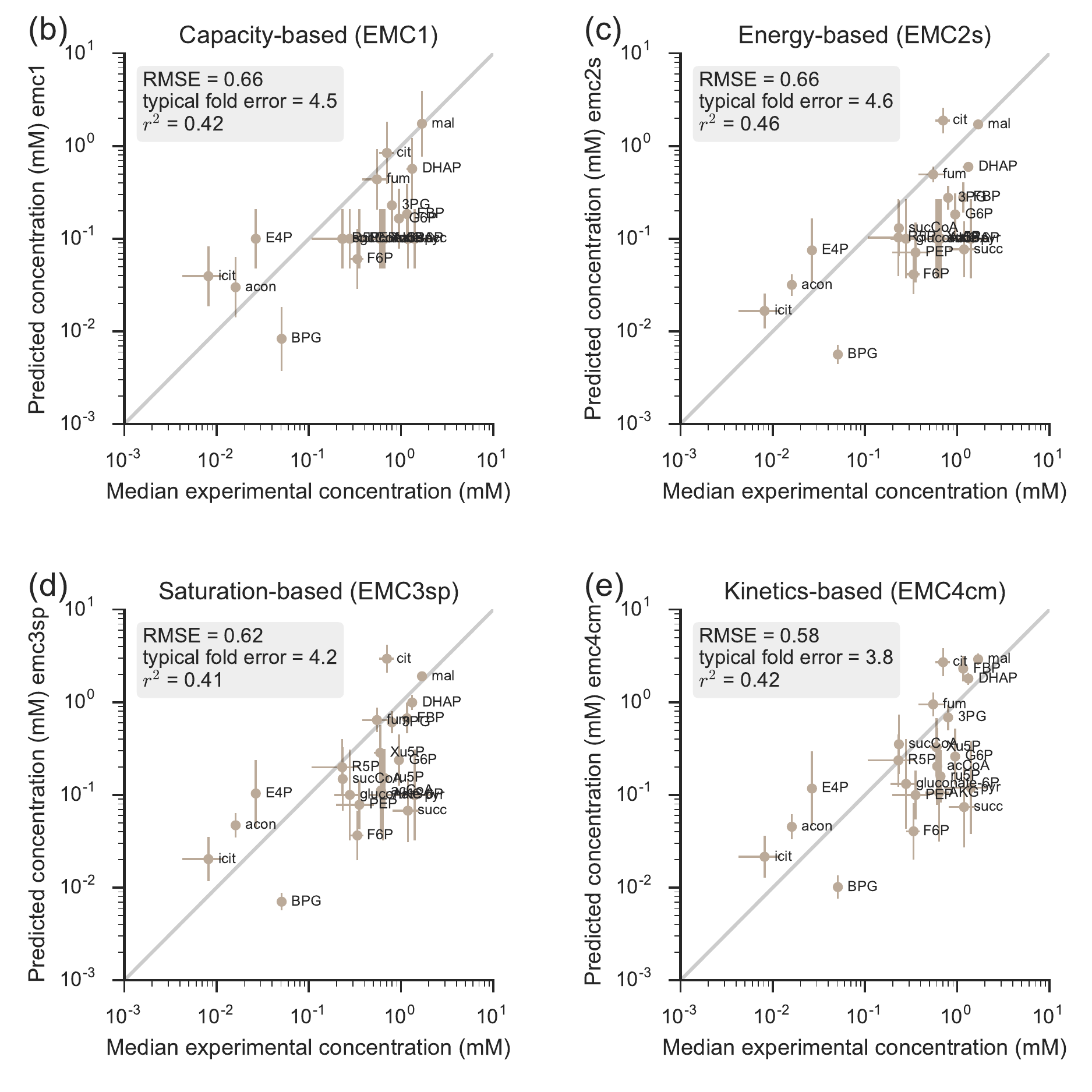}
    \caption{Metabolite levels predicted by enzymatic metabolic cost minimization.  As in Figure \ref{fig:EcoliPredictions} for
      enzyme predictions, vertical bars denote tolerance
      ranges. Horizontal lines represent uncertainties in measured
      data.  Predictions are based on fluxes from \cite{hann:11},
      $\kcatplus$ and $\km$ values from BRENDA \cite{sceg:04}, and
      validated with metabolite concentrations from \cite{gerc:15}.}
    \label{fig:EcoliPredictionsSI}
  \end{center}
\end{figure}

\section{Proofs and derivations}

\subsection{Lower bounds on driving forces, Eq.~(\ref{eq:forceThreshold})}
\label{sec:ProofforceThreshold}

Assuming that the cost for an individual enzyme  cannot exceed a
certain limit $q_l < q^{\rm max}$, we obtain Eq.~(\ref{eq:forceThreshold}) in
section \ref{sec:SItighterConstraints} as a lower bound on the driving
forces: $\Theta_l > \frac{\hel\,v_{l}}{\kcatplusl\, q^{\rm max}}$.  Noting
that $\Theta = -\Delta_{\rm r} G'/RT$, we get that the reaction Gibbs
energies are bounded by
\begin{eqnarray}
  \Delta_{\rm r} G'_{l}  &<& - RT \,\frac{\hel\,v_{l}}{q^{\rm max}\,\kcatplusl}.
\end{eqnarray}

\subsection{Parameters for lumped reactions}
\label{sec:SIlumped}

A lumped reaction describes a series of reactions as if they were
catalyzed by a single enzyme. The kinetic parameters should agree with
the original catalytic constants $\kcatplusl$, enzyme levels $\el$,
and enzyme cost weights $\hel$ of the individual reactions and yield
the right pathway flux $v=\enzyme\,\kcatplus$ and the right enzyme
cost $\enzymemetcost= h \,\enzyme$, but this still leaves some freedom
of choice.  On the one hand, we may assume that our hypothetical
lumped enzyme resembles a normal  enzyme in its kinetics and
concentration. This can be realized in different ways:
\begin{itemize}
\item Set $\kcat = \langle \kcatplusl \rangle_{\rm geom}$.  To satisfy
  $v = v_l = \kcatplusl\, \el = \langle \kcatplusl\, \el\rangle_{\rm geom} =
  \langle \kcat\rangle_{\rm geom} \langle \el\rangle_{\rm geom}$, we must set
  $\enzyme=\langle \el\rangle_{\rm geom}$.
\item Set $\kcat = \langle \kcatplusl \rangle_{\rm arith}$. To satisfy
 $v = \langle \kcatplusl\, \el\rangle_{\rm arith} = \langle
  \kcat\rangle_{\rm arith}\, \langle \frac{\kcatplusl}{\kcat} \el
  \rangle_{\rm arith}$, we must set $\enzyme = \langle
  \frac{\kcatplusl}{\kcat} \el \rangle_{\rm arith}$.
\item Set $\enzyme = \langle \enzyme_l \rangle_{\rm arith}$.  Again,
  we must set $\kcat = \langle \frac{\el}{\enzyme}
  \kcatplusl\rangle_{\rm arith}$.
\end{itemize}
In all three cases, the identity $h\, \enzyme = \sum_l \hel\, \el$
leads to the formula $h = \sum_l \hel\, \frac{\el}{\enzyme}$ for
specific cost. Since a lumped enzyme represents several real enzymes,
it will appear more costly or ``larger''. On the other hand, we can
assume that the concentration of the lumped enzyme is given by the sum
of original enzyme concentrations; this implies smaller effective
$\kcat$ values. To obtain the parameters, we can use the previous
formulae and replace $\enzyme \rightarrow n\,\enzyme$, $h \rightarrow
h/n$, and $\kcat \rightarrow \kcat/n$.

\subsection{Tolerance intervals around the minimum point of a strictly convex function}
\label{sec:SIHessianApproximation}

Consider a strictly convex function $f(\sv)$ with a global minimum
$\sv^*$.  Due to strict convexity, the Hessian $\Hm(\sv)$ is a
positive definite matrix.  To calculate tolerance intervals around the
minimum point, we choose the tolerance threshold $\tau$ (e.g. 1\% of
the minimum value) and define the tolerance subspace:
\begin{equation}
	\Stol \equiv \{ \sv~|~f(\sv) < f(\sv^*) + \tau\}.
\end{equation}
To get an explicit formula for $\Stol$, we first approximate $f$
around its minimum point by a Taylor expansion:
\begin{equation}
	f(\sv^* + \xiv) = f(\sv^*)~+~\nabla f(\sv^*) ^\top \xiv~+~\frac{1}{2} \cdot \xiv^\top \Hm(\sv^*) \xiv~+~\ldots
\end{equation}
In the minimal point, $\nabla f(\sv^*) = 0$ holds and we drop the
extra terms in the Taylor expansion to get
\begin{equation}
	f(\sv^* + \xiv) = f(\sv^*)~+~\frac{1}{2} \cdot \xiv^\top \Hm \xiv
\end{equation}
(for convenience, we use $\Hm$ to refer to the Hessian at the optimum). Therefore, the tolerance region can be approximated by:
\begin{equation}
	\Stol \approx E_{\rm tol} \equiv \{ \sv^* + \xiv~|~2\tau > \xiv^\top \Hm \xiv\}
\end{equation}

\begin{lemma}
If we define the ellipsoid $E \equiv \{\Hm^{-\mytextonehalf} \yv~|~ \yv^\top \yv < 1 \}$, then \[E_{\rm tol} = \sv^* + \sqrt{2 \tau}\cdot E\]
\end{lemma}

\begin{proof}
Since $\Hm$ is symmetric and positive definite, it is invertible and thus $\Hm^{-\mytextonehalf}$ is a unique symmetric matrix. Any $\sv \in \sv^* + \sqrt{2 \tau}\cdot E$ can be written as $\sv = \sv^* + \xiv$, where $\xiv = \sqrt{2 \tau} \cdot \Hm^{-\mytextonehalf} \yv$ and $\yv^\top \yv < 1$, therefore
\begin{equation}
	\xiv^\top \Hm \xiv = \sqrt{2 \tau} \cdot \yv^\top \Hm^{-\mytextonehalf}  \Hm \Hm^{-\mytextonehalf} \yv \cdot \sqrt{2 \tau} = 2 \tau \cdot \yv^\top \yv < 2 \tau.
\end{equation}
The reverse direction follows trivially. \qedhere
\end{proof}

\begin{corollary}
An ellipsoid is not always a convenient shape for describing the
tolerance intervals because there is dependence between the different
dimensions. For some application, it is sufficient to consider the
bounding box of $E$, which is given by $B \equiv \{ {\bf
D} \yv~|~ \yv \in [-1,1]^n \}$, where ${\bf D}$ is a matrix containing
only the diagonal values in $\Hm^{-\mytextonehalf}$ (i.e. ${\bf D}_{ii}
= \sqrt{(\Hm^{-1})_{ii}}$). Then we can approximate $\Stol$ by
\begin{equation}
 \Stol \approx \sv^* + \sqrt{2 \tau}\cdot B
\end{equation}
Therefore, for a single dimension $i$ the tolerance interval will be described by
\begin{equation}
 x^*_i \pm \sqrt{2 \tau (\Hm^{-1})_{ii}}
\end{equation}
\end{corollary}

\subsection{Enzyme costs reflects metabolic control
 (proposition \ref{prop:lincomb})}
\label{sec:ProofEnzymeControl}

Consider the ECM problem
\begin{eqnarray*}
 \mbox{Minimize}\quad \hminus(\enzymev)\qquad\mbox{subject to}\quad\jv_{\rm stat}(\enzymev) = \vv_{\rm stat}, \,\sv_{\rm bound}(\enzymev)=\cv_{\rm bound},
\end{eqnarray*}
where ``stat'' refers to independent stationary fluxes (with
running index $a$) and ``bound'' refers to metabolites that hit a
bound in the ECM solution considered (index $b$). With Lagrange
multipliers $\lambda_a$ and $\mu_b$ for the two sorts of constraints,
the optimality condition reads
\begin{eqnarray*}
 0 = \frac{\partial h}{\partial \el} + \sum_{a \in \rm stat} \lambda_{a} \frac{\partial \,j_a}{\partial \el}  + \sum_{b \in \rm bnd} \mu_{b} \,\frac{\partial s_b}{\partial \el}.
\end{eqnarray*}
After defining the enzyme cost slopes $\hel' = \frac{\partial
  h}{\partial \el}$ and multiplying the equation by $\el$, we obtain
\begin{eqnarray*}
 0 &=& \hel'\,\el + \sum_{a \in \rm stat} \lambda_{a} \frac{\partial \,j_a}{\partial \el}\el  + \sum_{b \in \rm bnd} \mu_{b} \,\frac{\partial s_b}{\partial \el}\el.
\end{eqnarray*}
We can now rewrite this in terms of control coefficients. The control
coefficients between enzymes and independent stationary fluxes are
defined by ${\mathcal C}^{j_a}_{l} = \frac{\el}{j_a} \frac{\partial
  \,j_a}{\partial \el} $, and those between enzymes and constrained
metabolites are defined by ${\mathcal C}^{s_b}_{l}= \frac{\el}{s_b}
\frac{\partial s_b}{\partial \el}$. Inserting this,
we obtain
\begin{eqnarray*}
0  &=& \hel'\,\el + \sum_{a \in \rm stat} -\alpha_{a} {\mathcal C}^{j_a}_{l} + \sum_{b \in \rm bnd} -\beta_{b} \,{\mathcal C}^{s_b}_{l},
\end{eqnarray*}
where we have defined $\alpha_{a} = -\lambda_{a}\,j_a$
and $\beta_{b} = -\mu_{b}\,s_b$, and thus
\begin{eqnarray*}
 \hel'\,\el &=& \sum_{a \in \rm stat} \alpha_{a} {\mathcal C}^{j_a}_{l} + \sum_{b \in \rm bnd} \beta_{b} \,{\mathcal C}^{s_b}_{l}.
\end{eqnarray*}
These relations hold for general non-linear cost function. In the
case of linear cost functions $\hminus(\enzymev) = \sum_l  \hel\,\enzyme_l$
(as usually assumed in ECM), the enzyme cost slopes $\hel'$ are directly given
by the cost weights $\hel'$.

\clearpage

\section{Mathematical symbols} 

  \begin{table}[h!]
\begin{center}
  \begin{tabular}{lll}
    \cellcolor{brown}    \textbf{Rate laws}&\cellcolor{brown} \textbf{Symbol} &\cellcolor{brown}\textbf{Units}\\ \hline \\[-2mm]
    Flux & $v_{l}$& mM/s\\
    Metabolite level              & $c_{i}$          & mM     \\
    Enzyme level & $\el$ & mM\\
    Rate law  & $v_{l}(\el, \cv) = \el \cdot \ratelaw_{l}(\cv)$ & mM/s \\
    Catalytic rate & $\ratelaw_{l} = v_{l}/\el$& 1/s\\
    Gibbs energy of formation (standard chem.~pot.) & $G'^{\circ}_{i}$ & kJ/mol\\
    Reaction Gibbs energy & $\Delta_{\rm r} G'_{l} = \Delta_{\rm r} {G'_l}^\circ + RT \sum_{i} n_{il}  \,\ln c_{i}$ & kJ/mol\\
    Thermodynamic driving force & $\Theta_l = - \Delta_{\rm r} G'_{l}/RT$ & unitless \\[2mm]
    \cellcolor{brown}    \textbf{Kinetic models}
    &\cellcolor{brown}  &\cellcolor{brown}\\ \hline \\[-2mm]
    Forward/backward catalytic constant  & $\kcatplus, \kcatminus$      & 1/s    \\
    Michaelis-Menten constant     & $\kmli$ & mM     \\
    Hill-like coefficient & $\gamma_l$ & unitless \\
    Molecularity for substrate (S) or product (P)   & $m^{\rm S}_{li}$,  $m^{\rm P}_{li}$   & unitless      \\
    Regulation coefficient (activator A or inhibitor I) & $m^{\rm A}_{li}, m^{\rm I}_{li}$   & unitless \\
    Scaled reactant elasticity & $\Esc_{li} = \frac{c_i}{v_l}\frac{\partial v_l}{\partial c_I}$ & unitless \\
    Scaled flux control coefficient & ${\mathcal C}^{j_a}_{l}$ & unitless\\
    Scaled concentration control coefficient &   ${\mathcal C}^{s_b}_{l}$ & unitless \\[2mm]
    \cellcolor{brown}        \textbf{Enzyme costs}
    &\cellcolor{brown}  &\cellcolor{brown}\\ \hline \\[-2mm]
    Enzyme cost   & $\hminus_l(\el) = \hel\,\el$ & D \\
    Enzyme cost weight  & $\hel$ & D/mM \\
    Protein mass & $m_l$ & Da \\
    Enzyme-based metabolic cost & $\enzymemetcost(\sv) = \sum_{l} \enzymemetcost_l(\sv) = \sum_{l} \hel \,\el(\sv,\vv)$ & D \\    
    Hessian matrix of enzyme-based metabolic cost &  $\hessian$ & D \\
    Flux-specific cost & $\rrl = \enzymemetcostl/v_l = \hel/\ratelaw_l$ & D/(mM/s)\\
    Baseline flux cost & $\rrlmin$ & D/(mM/s) \\[2mm]
    \cellcolor{brown} \textbf{Metabolic pathways}
    &\cellcolor{brown}  &\cellcolor{brown}\\ \hline \\[-2mm]
    Pathway flux (flux in representative reaction)& $\vPW$ & mM/s \\ 
    Scaled flux     & $\vscaled_{l} = v_{l}/\vPW$ & unitless \\ 
    Flux-specific cost & $\rrpw = \sum_{l  } \hel\,\el/\vPW = \sum_{l} \vscaled_{l}\,\rrl$ & D\,s/mM   \\
 \end{tabular}
\end{center}
\caption{Mathematical symbols used in ECM.  Darwin (D) is a
  hypothetical fitness unit replacing the possible fitness units in
  different models.  Reaction orientations are defined in such a way
  that fluxes are positive. Fluxes are given in units of concentration
  per time, but could also be given as amounts per time (e.g., mol/s);
  the latter choice is more practical for models with transport
  reactions.}
\label{tab:symbols}
\end{table}

\end{document}